\documentclass[10pt,journal,comsoc]{IEEEtran}
\usepackage[T1]{fontenc}

\usepackage{cite}
\usepackage{bbm}
\usepackage{bm}
\usepackage{amsmath}
\usepackage{amsthm}
\usepackage{float}
\usepackage{setspace}
\usepackage{amssymb}
\usepackage{stfloats}
\usepackage{cite}
\usepackage{ragged2e}
\usepackage[ruled,vlined,linesnumbered]{algorithm2e}
\usepackage{amsfonts}
\usepackage{mathrsfs}
\usepackage{amsmath,amsthm}
\usepackage{array,booktabs}
\usepackage{subfigure}
\usepackage{multirow}
\usepackage{cuted}
\usepackage{multicol}
\usepackage{graphicx}
\usepackage{subfigure}
\usepackage{graphicx,xcolor,bm}
\usepackage{hyperref}
\usepackage{threeparttable}
\usepackage{dcolumn}
\usepackage{setspace}
\usepackage{makecell}
\usepackage{lipsum}
\usepackage{enumerate}
\usepackage{graphicx} 		
\usepackage{subfigure}		
\usepackage{multirow}    
\usepackage{amssymb}     
\usepackage{xcolor}       
\usepackage{booktabs}    
\usepackage{tabularx}    
\usepackage{multirow}
\usepackage{array}
\usepackage{amssymb}
\usepackage{xcolor}
\usepackage{booktabs}
\usepackage{pifont}

\newtheorem{Defn}{Definition}

\newtheorem{thm}{Theorem}

\DeclareMathOperator*{\argmax}{argmax}
\usepackage{cite,bm}
\graphicspath{{figures/}}
\begin{document}
	\title{Effective Two-Stage Double Auction for Dynamic Resource Provision over Edge Networks via Discovering The Power of Overbooking}
	
\author{Sicheng Wu, Minghui Liwang, \IEEEmembership{Senior Member}, \IEEEmembership{IEEE}, Deqing Wang, \IEEEmembership{Member}, \IEEEmembership{IEEE}, Xianbin Wang, \IEEEmembership{Fellow}, \IEEEmembership{IEEE}, Chao Wu, Junyi Tang, Li Li,~\IEEEmembership{Member}, \IEEEmembership{IEEE}, Xiaoyu Xia, \IEEEmembership{Senior Member}, \IEEEmembership{IEEE}\vspace{-0.2cm}
	\thanks{This work was supported in part by National Natural Science Foundation of China under Grant nos. 62271424, 62271427; Shanghai Pujiang Programme under Grant no. 24PJD117; Shanghai Municipal Science and Technology Major Project under Grant no. 2021SHZDZX0100; Aeronautical Science Foundation of China under Grant no. 2023Z066038001; Chinese Academy of Engineering, Strategic Research and Consulting Program under Grant no. 2023-XZ-65; Key Science and Technology Project of Fujian Province under Grant no. 2023H0001.
		S. Wu (wusch2025@lzu.edu.cn) 
		and J. Tang (tjunyi2024@lzu.edu.cn) are with School of Information Science and Engineering, Lanzhou University, Lanzhou, China.
		M. Liwang (minghuiliwang@tongji.edu.cn)
		and L. Li (lili@tongji.edu.cn) are with the Department of Control Science and Engineering, the National Key Laboratory of Autonomous Intelligent Unmanned Systems, and also with Frontiers Science Center for Intelligent Autonomous Systems, Ministry of Education, Tongji University, Shanghai, China.
		D. Wang (deqing@xmu.edu.cn) is with the School of Informatics, Xiamen University, Fujian, China.
		X. Wang (xianbin.wang@uwo.ca) is with the Department of Electrical and Computer Engineering, Western University, Ontario, Canada. 
		C. Wu (chao.wu@zju.edu.cn) is with the School of Public Affairs, Zhejiang University, Hangzhou 310027, China. 
		X. Xia (xiaoyu.xia@rmit.edu.au) is with the School of Computing Technologies, RMIT University, Melbourne, Australia.
        
        Corresponding author: Minghui Liwang (minghuiliwang@tongji.edu.cn)
}}

	\IEEEtitleabstractindextext{
		\begin{abstract}
			\justifying
	To facilitate responsive and cost-effective computing service delivery over edge networks, this paper investigates a novel two-stage double auction methodology via discovering an interesting idea of resource overbooking to overcome dynamic and uncertain nature of supply of edge servers (sellers) and demand generated from mobile devices (as buyers). The proposed auction integrates multiple essential goals such as maximizing social welfare as well as accelerating the decision-making process from both short-term and long-term \textcolor{black}{perspectives (e.g., the time required to determine} winning seller-buyer pairs), by introducing a stagewise strategy: an overbooking-driven pre-double auction (OPDAuction) for determining long-term cooperations between sellers and buyers before practical resource transactions as Stage I, and a real-time backup double auction (RBDAuction) for quickly coping with residual resource demands during actual transactions. In particular, by embedding a proper overbooking rate, OPDAuction helps with facilitating trading contracts between appropriate sellers and buyers as guidance for future transactions, by allowing the booked resources to exceed theoretical supply. Then, since pre-auctions may cause risks, our RBDAuction adjusts to real-time market changes, further enhancing the overall social welfare. More importantly, we offer an interesting view to show that our proposed two-stage auction can support significant design properties such as truthfulness, individual rationality, and budget balance.
	\textcolor{black}{Extensive experiments demonstrate that our TwoSAuction achieves up to 76.8\% reduction in decision-making time compared to conventional double auctions when considering 150 buyers and 25 sellers, while maintaining superior performance in social welfare and computational scalability over dynamic edge settings.}
		\end{abstract}

		\begin{IEEEkeywords}
			Edge-assisted mobile networks, dynamics and uncertainty, two-stage double auction, overbooking, time efficiency
		\end{IEEEkeywords}
}
	
\maketitle
\IEEEdisplaynontitleabstractindextext
%
\IEEEpeerreviewmaketitle

\section{Introduction }

\IEEEPARstart{T}{he} proliferation of smart devices and their enhanced computing capabilities have enabled a wide range of innovative mobile applications, such as intelligent transportation, E-health, and smart homes \cite{b1}. Most of these applications are machine learning-driven, requiring intensive computation for massive data analysis, which exceeds the capacity of individual devices with limited resources and battery supply \cite{d1-2}. Edge computing provides a viable paradigm for resource sharing at the network edge, supporting responsive and cost-effective service delivery \cite{b2}. However, guaranteeing sufficient resources for delay-sensitive and computation-intensive tasks remains difficult due to limited edge server capacity and concurrent demands from numerous devices \cite{d1-4}. The challenge further intensifies with the increasing dynamics and heterogeneity of distributed application requirements. To bridge the gap between distributed edge resources and diverse service demands, the double auction \cite{b3} emerges as a promising solution, allowing buyers and sellers to submit bids and asks to establish mutually beneficial trades. This creates a market with built-in incentives \cite{d1-6}, enhancing overall network efficiency.

\subsection{Motivation}
\color{black}
This section outlines our key motivations in a Q\&A form. \textit{(i) Why is a market-based approach necessary for service provision in edge networks?} Traditional non-market mechanisms struggle to handle device heterogeneity, time-varying supply/demand, partial information, and self-interested participants. In contrast, market-based solutions introduce price signals and incentive-compatible rules (e.g., auctions, dynamic pricing, contracts), aligning individual incentives with system objectives, enabling scalable allocation under uncertainty, and supporting long-term participation \cite{b4}. \textit{(ii) Why adopt a double auction with monetary incentives?} Edge resource trading is inherently two-sided, involving self-interested buyers (mobile devices) and sellers (edge servers). Double auctions capture heterogeneous valuations on both sides and ensure fair, efficient matching. Monetary incentives are essential for incentive compatibility: buyers truthfully reveal valuations via payments, while sellers allocate resources to maximize revenue, ensuring fairness, sustained participation, and robustness \cite{d1-8}. \textit{(iii) Why is a conventional real-time double auction insufficient?} Despite their strengths, real-time implementations incur frequent bidding and clearing, causing high communication and computation overheads that undermine scalability in dynamic environments \cite{b5}. \color{black}\textit{(iv) Why introduce resource overbooking?} Edge resources are often underutilized due to uncertain and volatile demand. Unlike cloud platforms with abundant pooled capacity, edge nodes are geographically constrained and cannot offload shortages. Resource overbooking mitigates this by safely oversubscribing local capacity, reducing idleness and improving utilization (see Appendix \ref{ob.vs}).

\color{black}
Answering the above four questions forms our core motivation. To support responsive and cost-effective double auctions in dynamic edge networks, we design a pre-auction stage for early negotiation, fostering long-term partnerships and reducing bargaining overheads. We further pioneer the use of resource overbooking, allowing sellers to offer beyond nominal supply and buyers to request beyond immediate needs \cite{b5-1}. This, however, requires accurate estimation of historical data, as improper rates may cause economic or operational risks absent in cloud environments. To address this, we develop a two-stage double auction that integrates pre-auction and real-time decisions, strategically exploiting overbooking to manage diverse uncertainties (see Fig.~1).
\color{black}

\begin{figure}[t!]
	\centering
	\includegraphics[trim=0cm 0cm 0cm 0cm, clip, width=\columnwidth]{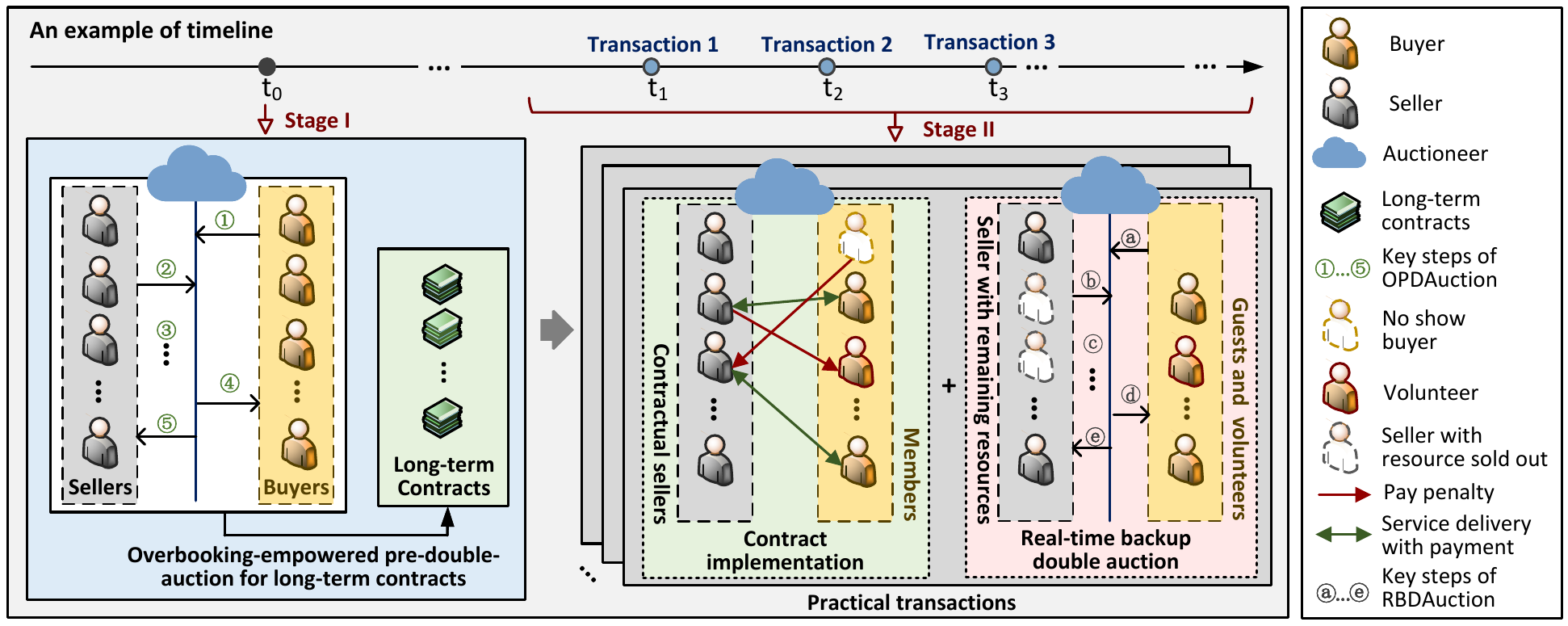}
	\caption{Schematic of our proposed two-stage double auction. During Stage I:  
		\ding{172} Buyers report their information such as demand, bid, attendance probability;  
		\ding{173} Sellers report their information such as resource capacity and ask pricing;  
		\ding{174} Auctioneer determines members, long-term contracts and overbooking rate;  
		\ding{175}\ding{176} Auction results (contract terms) are delivered to buyers and sellers.   
		During Stage II:  
		\textcircled{\scriptsize a} Volunteers and guests submit demand and bid;  
		\textcircled{\scriptsize b} Sellers with remaining resources submit their supply and ask pricing;  
		\textcircled{\scriptsize c} Auctioneer decides trading pairs and payment;  
		\textcircled{\scriptsize d}\textcircled{\scriptsize e} Auction results are delivered to buyers and sellers.}
	\label{fig_1}
\end{figure}

\subsection{Contribution}
As far as we know, this paper makes the first attempt among the existing literature to design a stagewise double auction mechanism with overbooking embedded, that facilitates resource trading between edge servers (as sellers) and mobile devices (as buyers), upon considering the dynamic and uncertain nature of edge networks. Our core principle \textit{is to maintain economic efficiency in resource trading, ensuring mutually beneficial outcomes for all parties involved, while fostering a truthful, individually rational, and risk-aware auction environment. Additionally, we aim to minimize overhead (e.g., delays) in decision-making, thereby facilitating a time-efficient auction process}\footnote{Fig. 2 presents simulation results illustrating the benefits of our proposed two-stage double auction with overbooking. While a conventional double auction may achieve slightly higher social welfare, it suffers from low time efficiency, rendering it impractical for dynamic, real-world networks.}. Key contributions are summarized below.

\noindent $\bullet$ Given the dynamic nature of edge networks, characterized by uncertain buyer participation and fluctuating resource supply of sellers, we propose a novel two-stage double auction methodology. This integrates a pre-double auction stage with overbooking (Stage I) and a real-time double auction stage as a backup (Stage II), with the goal of optimizing social welfare, ensuring time efficiency, while supporting essential auction properties.

\noindent $\bullet$ We first design an overbooking-driven pre-double auction (OPDAuction) implemented before future resource transactions. This incentivizes sellers and buyers to negotiate risk-aware, long-term trading contracts aimed at maximizing expected social welfare. Specifically, sellers are encouraged to overbook their services to account for potential buyer participation uncertainties. These contracts, comprising terms such as trading price, resource volume, and default clauses, can be directly implemented in practical transactions, thus ease the burden of seeking for auction results during subsequent stage.

\noindent $\bullet$ Given that uncertain nature can lead to unsatisfactory service quality, such as sellers failing to deliver promised resources due to overbooking or buyers not receiving the resources they require, we then design a real-time backup double auction (RBDAuction). This serves as an efficient complementary plan to further enhance social welfare.

\noindent $\bullet$ Theoretical analysis and experiments demonstrate that our two-stage double auction upholds essential properties. Furthermore, extensive evaluations showcase strong performance across multiple dimensions, such as social welfare, time efficiency in auction decision-making, individual rationality, and truthfulness.

\begin{figure}[h!]
	\centering
	\subfigure[] {\includegraphics[trim=0cm 0cm 0cm 0cm, clip, width=.241\textwidth]{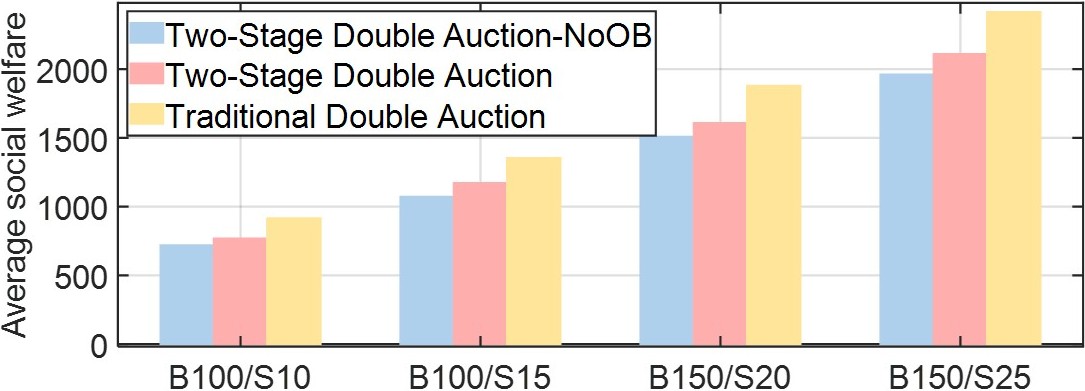}}
	\subfigure[] {\includegraphics[trim=0cm 0cm 0cm 0cm, clip, width=.241\textwidth]{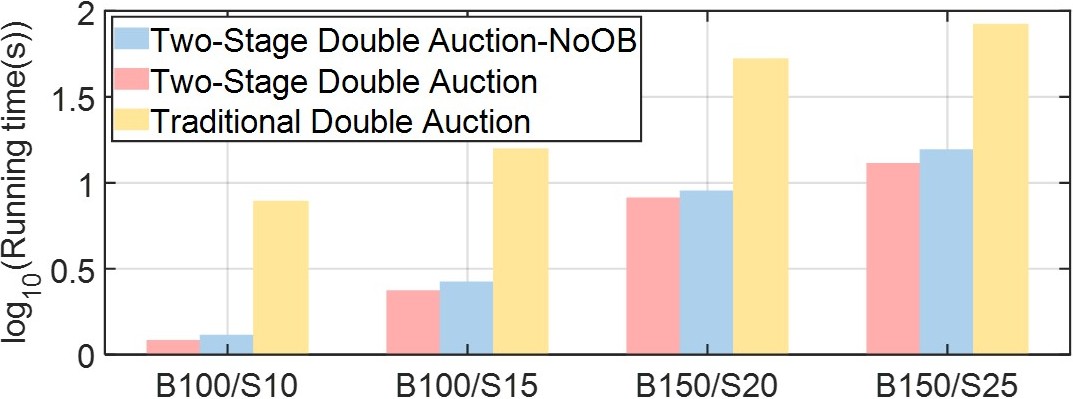}}
	\caption{Advantages brought by two-stage double auction and overbooking in terms of social welfare and time efficiency, via simulations.}
	\label{fig_E1}
\end{figure}

\color{black}
\section{Literature Review and Key Challenge}
\noindent
\textit{Investigation on Existing Literature. }\color{black}  Numerous studies have investigated auction-based resource provision mechanisms over edge networks. \textit{Li et al.} \cite{b6} developed a reinforcement learning framework for privacy-preserving resource trading in double auction markets, balancing privacy protection and resource allocation efficiency. \textit{Chen et al.} \cite{b7} proposed a blockchain-based iterative double auction mechanism for decentralized manufacturing resource sharing, achieving verifiable matching through smart contracts. \textit{Cui et al.} \cite{b8} focused on security-enhanced resource allocation, with double auction to coordinate task scheduling. \textit{Dai et al.} \cite{b9} designed a hybrid auction-game theoretic scheme for marine resource scheduling, combining Stackelberg game modeling with double auction to optimize task offloading between heterogeneous maritime devices.
 \textit{Zhang et al.} \cite{b10} addressed resource competition in mobile edge computing through a double auction model that coordinates base station allocation while mitigating service interruptions. \textit{Li et al.} \cite{b11} introduced a dual-identity double auction framework with reinforcement learning-based model selection to align incentives in personalized federated learning, tackling both data heterogeneity and participant reluctance. \textit{Qi et al.} \cite{b12} developed a group-buying coalition auction mechanism leveraging double auction principles to optimize UAV-assisted data collection through truthful sensor coalition formation. \textit{Du et al.} \cite{b13} integrated double auctions with blockchain-based smart contracts to enable verifiable resource leasing in distributed edge computing markets while ensuring social welfare maximization. \textit{Yin et al.} \cite{b14} employed multi-agent reinforcement learning within double auction frameworks to optimize strategic bidding in electricity markets with transmission congestion constraints. \textit{Zhang et al.} \cite{b15} combined Shapley value-based cooperation with double auctions to achieve adaptive computing resource transfers in wireless computing power networks.
 

\begin{table}[t!]
	\color{black}
	\centering
	\caption{Visualization of novelties of this paper through literature investigation\\(DP: Dynamic Pricing, SC: Smart Contract, PA: Pre-auction, OB: Overbooking, RC: Risk Control, T: Truthfulness, IR: Individual Rationality, BB: Budget Balance)}
	\vspace{-0.2cm}
	\label{tab:comparison}
	\scriptsize
	\resizebox{0.5\textwidth}{!}{%
	\begin{tabular}{|l|c|c|c|c|c|c|c|c|}
		\hline
		\multicolumn{1}{|c|}{\multirow{2}{*}{\textbf{Reference}}} & \multicolumn{2}{c|}{\textbf{Auction Attributes}} & \multicolumn{3}{c|}{\textbf{Advanced Functions}} & \multicolumn{3}{c|}{\textbf{Key Properties}} \\
		\cline{2-9}
		& DP & SC & PA & OB & RC & T & IR & BB \\
		\hline
		\cite{b6},\cite{b11} & $\checkmark$ &  &  &  &  & $\checkmark$ & $\checkmark$ &  \\
		\cite{b7} & $\checkmark$ & $\checkmark$ &  &  &  & $\checkmark$ & $\checkmark$ & $\checkmark$ \\
		\cite{b8},\cite{b13},\cite{b15} &  & $\checkmark$ &  &  &  & $\checkmark$ & $\checkmark$ & $\checkmark$ \\
		\cite{b9},\cite{b14} & $\checkmark$ &  &  &  &  & $\checkmark$ & $\checkmark$ &  \\
		\cite{b10} & $\checkmark$ &  &  &  &  & $\checkmark$ & $\checkmark$ &  \\
		\cite{b12} &  &  &  &  &  & $\checkmark$ & $\checkmark$ &  \\
		\hline
		our work & $\checkmark$ & $\checkmark$ & $\checkmark$ & $\checkmark$ & $\checkmark$ & $\checkmark$ & $\checkmark$ & $\checkmark$ \\
		\hline
	\end{tabular}
}
	\vspace{-0.2cm}
\end{table}

\color{black}
Although prior works have made notable contributions, they mainly focus on on-site auctions, where decisions rely solely on current network or market states. Such processes incur high overhead, as identifying winning buyer-seller pairs and clearing prices often requires multiple negotiation rounds, causing delays and extra energy costs, especially for battery-constrained mobile devices. In addition, most studies assume fixed resource availability and stable demand, assumptions rarely valid in dynamic edge networks (see Table~1). These limitations motivate the design of time-efficient auctions that adapt to network dynamics. Pre-auction decision-making has shown promise in our early studies \cite{b17,b20}, while resource overbooking further improves trading efficiency \cite{b16}. However, inaccurate estimation of uncertain factors, such as fluctuating supply, demand, or wireless conditions, may still yield suboptimal outcomes, including unreasonable pricing and application failures, which are also of our interest.

\noindent
\textit{Key Challenges.} Driven by the above, this paper explores a novel stagewise double auction methodology over dynamic and uncertain edge networks, which combines pre-decisions and real-time decisions into a two-stage process for resource provision. To be implementable, we focus on addressing the following key challenges that arise specifically in our problem setting, namely, a dynamic edge environment characterized by uncertain buyer participation, fluctuating seller supply, and the strategic integration of overbooking. These challenges are intrinsic to our design goals and differ from those commonly encountered in conventional auction frameworks\footnote{\textcolor{black}{For a discussion on the potential technical challenges that may arise when extending our mechanism to incorporate additional non-standard auction properties, please refer to Appendix H.}}:

\noindent $\bullet$ \textit{Double auctions typically have inherent properties}, such as truthfulness and individual rationality\footnote{\textcolor{black}{Following established conventions in auction literature (e.g., \cite{R1}, \cite{R2}, \cite{R3}, \cite{R4}, \cite{R5}), we focus on the core properties of truthfulness, individual rationality, and budget balance, which are widely regarded as standard requirements for double auctions. Collusion resistance, though theoretically relevant, is not treated as a standard property in these foundational works and is therefore beyond the scope of this paper; its investigation is left for future research.}}. Maintaining these properties while designing the two-stage double auction presents a significant challenge. For example, it requires careful consideration on factors such as price determination during the pre-auction process.

\noindent $\bullet$ \textit{An unique aspect of our studied auction is to determine an appropriate overbooking rate}. A low overbooking rate can limit the effectiveness of pre-made decisions in managing dynamic resource demands, while a large one can result in insufficient resource supply. Therefore, integrating overbooking into a double auction and determining proper rate remains a noteworthy challenge.

\noindent $\bullet$ \textit{The adaptation and scalability of a double auction in dynamic and uncertain network environments are crucial}. The inherently dynamic and unpredictable nature of real-world edge networks necessitates designing double auction mechanisms that are highly adaptable and scalable to effectively manage various uncertain factors.
\color{black}

\section{Overview}
We are interested in a dynamic resource trading market (see Fig. 1) over edge networks that involves three key parties: \textit{(i)} multiple resource buyers denoted by $\bm{\mathcal{B}} = \left\{b_1, \ldots, b_n, \ldots, b_{|\bm{\mathcal{B}}|}\right\}$, where each $b_n \in \bm{\mathcal{B}}$ periodically generates computation-intensive tasks~\cite{b20-1}, requiring edge resources for further computing; \textit{(ii)} multiple resource sellers (edge servers) represented by $\bm{\mathcal{S}} = \left\{s_1, \ldots, s_m, \ldots, s_{|\bm{\mathcal{S}}|}\right\}$, where each $s_m \in \bm{\mathcal{S}}$ owns a certain amount of resources that can serve buyers for certain fee; and \textit{(iii)} a neutral auction platform, playing the role of a trustworthy auctioneer that coordinates resource trading among buyers and sellers. Note that in our model, resources are quantized, e.g., resource blocks (RBs), for analytical simplicity~\cite{b20-2}.
A \textit{transaction} in our considered market refers to a trading event, in which \textit{a buyer can transfer its task to a seller for edge processing while paying for the obtained services, while a seller can serve multiple buyers simultaneously according to its resource supply.} 
To capture the uncertainty and dynamics of resource supply and demand in a trading market, we identify two key factors. First, buyers’ selfishness, willingness, and mobility can lead to \textit{fluctuating resource demand}; for example, a buyer may be absent if moving out of a seller’s coverage. Second, sellers’ resource supply can vary over time due to factors such as local workloads, affecting the \textit{availability of resources} for sale.

To ensure cost-effective and mutually beneficial resource provisioning under dynamic conditions, we propose a two-stage double auction: \textit{(i) Stage I}, the \textbf{o}verbooking-driven \textbf{p}re-\textbf{d}ouble \textbf{auction} (OPDAuction), enables resources to be traded ahead of actual transactions. Buyers and sellers negotiate long-term contracts based on historical supply-demand data while managing potential risks. Each winning buyer becomes a \textit{member} of the corresponding seller, and overbooking allows sellers to contract with more buyers than their nominal capacity, with the market setting an appropriate overbooking rate to protect profits. This pre-auction eliminates the need for negotiating transaction details (e.g., prices), reducing overhead such as auction delays. To address inaccuracies when historical estimates diverge from real conditions, \textit{(ii) Stage II}, the \textbf{r}eal-time \textbf{b}ackup \textbf{d}ouble \textbf{auction} (RBDAuction), is triggered during transactions. RBDAuction lets non-member buyers (``\textit{guests}'') and \textit{volunteers} (members lacking promised resources) compete for remaining capacity using up-to-date information, improving real-time efficiency and system-level performance (e.g., social welfare).

\subsection{Key Definition}
In the following, we will cover the unique definitions in our auction design.

\begin{Defn}
	\textit{(Member and Guest):} Buyers who have signed long-term contracts with sellers in Stage I are called members; while those who without contracts and will be engaged in Stage II are named as guests.
\end{Defn}
\begin{Defn}
	\textit{(Volunteer):} A volunteer is a member who participates in a transaction but fails to obtain the required resources stipulated in the pre-signed contract. As compensation, each volunteer receives a monetary incentive from the corresponding seller, and will compete with guests for available resources.
\end{Defn}
The aforementioned \textit{member, guest,} and \textit{volunteer} constitute the most critical personas regarding buyers in the market of our interest. Following this, we next detail crucial auction properties.
\begin{Defn}
	\textit{(Individual rationality):} An auction is individual rational if the expense of any winning buyer (paid to a seller) does not exceed its bid; while the income of any winning seller can cover its asked payment.
\end{Defn}
\begin{Defn}
	\textit{(Truthfulness):} An auction is truthful (or incentive compatible) when all the buyers and sellers declare the bids and asked payments same as their true (private) valuations. Namely, participants will not misreport their information.
\end{Defn}
\begin{Defn}
	\textit{(Budget-balance):} An auction achieves budget-balance if the overall income of the auctioneer from winning buyers can be larger than or at least equal to the total expense that it pays to winning sellers.
\end{Defn}
\begin{Defn}
	\textit{(Computational efficiency):} A mechanism is computationally efficient if it runs in polynomial time.
\end{Defn}

Our two-stage double auction aims to support the above key properties from unique perspectives (which are different from conventional auctions, please see Section 3.2.2). Besides, as a distinctive feature of our auction design, we then introduce the concept of overbooking rate, representing an innovative consideration to enhance resource provisioning in dynamic trading markets.
\begin{Defn}
	\textit{(Overbooking rate):} As the actual available resources (denoted by $R_m$ for analyzing simplicity) of seller $s_m$ is challenging to be evaluated during stage I, we consider its expected value denoted as $\overline{R_m}$ to quantify the theoretical resource supply. Accordingly, the overbooking rate of $s_m$ (denoted by $\lambda_m$) indicates the ratio of the amounts of booked resources t (e.g., the overall contractual resources) to its expected resource supply, as calculated by $\lambda_m=\max\left\{0,\frac{t-\overline{R_m}}{\overline{R_m}}\right\}$\footnote{In our designed market, overbooking represents a fundamental policy where all the sellers use the same overbooking rate to maintain the market law. Different overbooking rates among various sellers and their optimization will be studied in our future work.}.
\end{Defn}

\section{Stage I. Overbooking-driven Pre-Double Auction (OPDAuction)}
We first delve into the \textbf{o}verbooking-enabled \textbf{p}re-\textbf{d}ouble \textbf{auction} (OPDAuction) scheme that facilitates early agreements on resource trading between buyers and sellers, prior to future transactions. Specifically, OPDAuction focuses on identifying suitable members for sellers and establishing their long-term contracts. In the following, we begin by detailing key models.

\subsection{Key Modeling}
\subsubsection{Modeling of long-term contract}
The long-term contract $ \mathbb{C}_{m,n}^\mathsf{b\leftrightarrow s}$ between buyer $ b_n $ and seller $ s_m $ in stage I is modeled by a 5-tuple $ \mathbb{C}_{m,n}^\mathsf{b\leftrightarrow s}=\left\{t_n,p_n^\mathsf{b},r_m^\mathsf{s},q_{m,n}^\mathsf{b\to s},q_{m,n}^\mathsf{s\to b} \right\} $, where $t_n$ represents the amount of trading resources (quantized by RBs for analytical simplicity \cite{b20-23}); $ p_n^\mathsf{b} $ indicates the unit payment (per RB) from $ b_n $ to the contractual seller; $ r_m^\mathsf{s} $ is the unit reward (per RB) of $ s_m $ for offering services; $ q_{m,n}^\mathsf{b\to s} $ and $ q_{m,n}^\mathsf{s\to b} $ describe the default clause, referring to the unit penalty (per RB) when $ b_n $, and $ s_m $ breaks the contract, respectively. For example, when buyer $ b_n $ is absent from a transaction and will not purchase the preserved resources, it has to pay a penalty (e.g., $ q_{m,n}^\mathsf{b\to s} $) to seller $ s_m $. Similarly, the limited resource supply of sellers can also result in failures to afford promissory computing services, causing possible compensations (e.g., $ q_{m,n}^\mathsf{s\to b} $) to buyers. Our proposed OPDAuction in Stage I makes efforts to maximize the expectation of social welfare, e.g., the overall expected profit of all the three parties (see Section 3.2.2), by mapping buyers to feasible sellers while optimizing their contract terms as well as a proper overbooking rate.

\subsubsection{Modeling of buyers}

Considering multiple resource buyers collected in set $ \bm{\mathcal{B}} $, where each  $b_n\in\bm{\mathcal{B}}$ is described by a 4-tuple $b_n=\left\{t_n,v_{m,n},bid_{m,n},\alpha_n \right\}$. Specifically, $t_n$ denotes its required resources; $v_{m,n}$ represents the unit valuation (per RB) of $b_n$ as benefited from enjoying computing service offered by $s_m$, which varies across different sellers due to heterogeneous attributes of edge servers, e.g., hardware settings. Specifically, $v_{m,n}$ \textcolor{black}{is a private information of} $ b_n $, unknown to \textcolor{black}{sellers or the auctioneer}. Besides, $bid_{m,n}$ refers to the unit price (per RB) that $ b_n $ is willing to afford to $ s_m $.

To capture the \textcolor{black}{dynamic and} random nature of edge networks, we model the uncertain buyers' participation through independent but non-identical Bernoulli trials. Specifically, $\alpha_n$ obeys $\alpha_n\sim \text{\textbf{Ber}}\left\{(1,0),(\mathbbm{a}_n,1-\mathbbm{a}_n)\right\}$, where $\alpha_n=1$ (with probability $\mathbbm{a}_n$) indicates $b_n$'s attendance. Such an assumption satisfies two critical properties: \textit{(i)} \emph{statistical independence} among buyers aligns with the auction process across multiple sellers' decision-making; \textit{(ii)} \emph{Poisson convergence} emerges when considering collective behaviors -- the total number of active buyers $\sum_{n}\alpha_n$ asymptotically follows a Poisson distribution under mild conditions\footnote{Although our model preserves
	heterogeneity through distinct an, this theoretical foundation
	justifies using Bernoulli trials to approximate Poisson distributed
	mass events in large-scale systems, while maintaining
	individual-level controllability for behavior diversity \cite{b20-3}.} \cite{b20-3}.
To facilitate our analysis, let $x_{m,n}$ represent the assignment between buyers and sellers, where $x_{m,n}=1$ indicates that buyer $ b_n $ wins the service of seller $ s_m $ in OPDAuction and thus becomes a member; while $x_{m,n}=0$, otherwise. We then use $\bm{X}=\left\{x_{m,n} | b_n \in \bm{\mathcal{B}}, s_m \in \bm{\mathcal{S}}\right\}$  to denote the profile of $x_{m,n}$ for notational simplicity. 
\textcolor{black}{This binary assignment is linked to the participation indicator $\alpha_n$: a buyer can only be matched ($x_{m,n}=1$) if they are active ($\alpha_n=1$), thereby preventing inactive buyers from obtaining resources}\footnote{\textcolor{black}{Notably, this modeling choice leverages the Bernoulli–Poisson relationship: it preserves buyer heterogeneity via individual $\mathbbm{a}_n$ parameters while enabling analytical tractability through Poisson approximation.}
This balances mathematical rigor with practical considerations for edge computing markets where participants exhibit \emph{sporadic yet correlated} engagement patterns.}.

\noindent\textbf{1) Utility and expected utility of $b_n$:} The utility of a buyer $b_n\in \bm{\mathcal{B}}$ is denoted by $U_n^\mathsf{B} \left(t_n,p_n^\mathsf{b},q_{m,n}^\mathsf{b\to s},q_{m,n}^\mathsf{s\to b} \right)$, which comprises three key aspects: \textit{(i)} the profit that $b_n$ receives from enjoying computing service offered by $s_m$; \textit{(ii)} The compensation that $b_n$ obtains from the contractual seller $s_m$ for being selected as a volunteer; and \textit{(iii)} The penalty that $b_n$ pays to the seller when it breaks the contract (e.g., $\alpha_n=0$). Correspondingly, the utility of $b_n$ can be calculated as 
\begin{equation}\label{key}
	\begin{aligned}
		&U_n^\mathsf{B} \left(t_n,p_n^\mathsf{b},q_{m,n}^\mathsf{b\to s},q_{m,n}^\mathsf{s\to b} \right)=
		\\&\sum_{s_m\in \bm{\mathcal{S}}}x_{m,n}t_n\alpha_n\left(M_n\left(v_{m,n}-p_n^\mathsf{b}\right)+(1-M_n) q_{m,n}^\mathsf{s\to b}\right)
		\\&-\sum_{s_m\in \bm{\mathcal{S}}}x_{m,n}t_n(1-\alpha_n)q_{m,n}^\mathsf{b\to s},
	\end{aligned}
\end{equation}
where $ p_n^\mathsf{b} $ is the unit payment of $ b_n $; and $ M_n $ denotes a binary variable for volunteers. Specifically, we have $ M_n= 1 $ at probability $ 1-\mathbb{P}_n $, describing that buyer $ b_n $ attends a practical transaction and enjoys computing service from $ s_m $; while $ M_n= 0 $ at probability $ \mathbb{P}_n $ depicts that $ b_n $ is selected as a volunteer (derivation of $\mathbb{P}_n$ is given by (5)). Let $\mathbf{M}=\left\{M_1,...,M_n,...,M_{|\bm{\mathcal{B}}|}\right\} $ be the profile of $ M_n $ for notational simplicity. 
Since OPDAuction is implemented before practical transactions, obtaining the practical value of $U_n^\mathsf{B}$ is challenging due to the uncertain values of $ \alpha_n $, we accordingly calculate its expectation as 
\begin{equation}\label{key}
	\begin{aligned}
		&\overline{U_n^\mathsf{B}} \left(t_n,p_n^\mathsf{b},q_{m,n}^\mathsf{b\to s},q_{m,n}^\mathsf{s\to b} \right)=
		\\&\sum_{s_m\in \bm{\mathcal{S}}}x_{m,n}t_n\mathbbm{a}_n\left(\left(1-\mathbb{P}_n\right)\left(v_{m,n}-p_n^\mathsf{b}\right)+\mathbb{P}_n q_{m,n}^\mathsf{s\to b}\right)
		\\&-\sum_{s_m\in \bm{\mathcal{S}}}x_{m,n}t_n(1-\mathbbm{a}_n)q_{m,n}^\mathsf{b\to s}.
	\end{aligned}
\end{equation}
\noindent\textbf{2) Risk analysis of $b_n$:} Given that our considered trading market encompasses both profits (e.g., positive utility for buyers) and risks, we evaluate two key risks that each buyer may encounter during a transaction. First, we define the risk of a member $ b_n $ (not a volunteer) experiencing a non-positive utility, abbreviated as ``BRisk'', as the probability that $ b_n $'s utility falls too close to or below its acceptable minimum $ U^\mathsf{min} $:
\begin{equation}\label{key}
	\begin{aligned}
		&\mathcal{R}_{n}^\mathsf{BRisk} \left(t_n,p_n^\mathsf{b},q_{m,n}^\mathsf{b\to s} \right)
		\\&=\text{Pr}\left(\frac{t_n\left(\alpha_n\left(v_{m,n}-p_n^\mathsf{b}\right)-\left(1-\alpha_n\right)q_{m,n}^\mathsf{b\to s}\right)}{U^\mathsf{min}}\leq\xi_1\right)
		\\&=\begin{cases}
			0,\mathbbm{c}_1<0\\
			1-\mathbbm{a}_n,0\leq\mathbbm{c}_1<1\\
			1,1\leq\mathbbm{c}_1\\
		\end{cases},
	\end{aligned}
\end{equation}
where $ U^\mathsf{min} $ is a positive value approaching to zero, $ \xi_1 $ denotes a positive threshold coefficient, and $ \mathbbm{c}_1=\frac{\frac{U^\mathsf{min}\xi_1}{t_n}+q_{m,n}^\mathsf{b\to s}}{v_{m,n}+q_{m,n}^\mathsf{b\to s}-p_n^\mathsf{b}} $ represents a constant for notational simplicity. Derivation of $ \mathcal{R}_n^\mathsf{BRisk} $ is detailed in Appendix D.
Then, the implementation of resource overbooking may force some members to take on volunteer roles and forgo access to edge services. Although these buyers can be compensated, this arrangement might lead to a less favorable trading experience for them. Therefore, we also consider the risk of a member being selected as a volunteer, referred to as ``VRisk'', during an actual transaction (see (4)). This risk is defined as the probability that a member $ b_n $ participates in a transaction but is chosen as a volunteer owing to the insufficient resources of its contractual seller.
\begin{equation}\label{key}
	\begin{split}
		& \mathcal{R}_n^\mathsf{VRisk}= \\
		&\text{Pr}\Biggl(t_n-\sum_{s_m \in \bm{\mathcal{S}}} x_{m, n}\left(\mathbbm{d}_m\mathbbm{r}_m-\sum_{b_{n'}\in \bm{\mathcal{B}}^-}\mathbbm{a}_{n'}x_{m,{n'}}t_{n'}\right)\\ 
		&\geqslant 0,\alpha_n = 1\Biggr)=\mathbbm{a}_n\mathbb{P}_n,
	\end{split}
\end{equation}
where $ \bm{\mathcal{B}}^-=\bm{\mathcal{B}}\setminus{\{b_n\}} $ denotes the set of buyers excluding $ b_n $. Moreover, we have $\mathbb{P}_n$ as the following (5).
\begin{equation}\label{key}
	\small{
	\begin{aligned}
		&\mathbb{P}_n = \\
		&1-\frac{\sum_{s_m\in \bm{\mathcal{S}}}x_{m,n}\left[\mathbbm{d}_m\mathbbm{r}_m(1-\mathbbm{d}_m)+\sum_{b_{n'}\in \bm{\mathcal{B}}^-}\mathbbm{a}_{n'}x_{m,{n'}}(t_{n'})^2\right]}{\left(t_n - \sum_{s_m\in \bm{\mathcal{S}}}x_{m,n}\left(\mathbbm{d}_m\mathbbm{r}_m-\sum_{b_{n'}\in \bm{\mathcal{B}}^-}\mathbbm{a}_{n'}x_{m,{n'}}t_{n'}\right)\right)^2}.
	\end{aligned}
}
\end{equation}

Apparently, a large value of $ \mathcal{R}_n^\mathsf{VRisk} $ indicates an increased likelihood of the buyer $ b_n $ chosen as a volunteer, which negatively impacts their trading experience. Parameters $  \mathbbm{d}_m $ and $ \mathbbm{r}_m $ associated with $ s_m $ will be elaborated in the following section, and the detailed derivation of $ \mathcal{R}_n^\mathsf{VRisk} $ is provided in Appendix D.
\subsubsection{Modeling of sellers}
A seller $ s_m\in\bm{\mathcal{S}} $ is denoted by a triple $ s_m=\left\{c_m,{ask}_m,R_m\right\} $, where $ c_m $ represents the unit cost (per RB) for processing buyers' tasks, e.g., consumed energy and hardware cost, \textcolor{black}{which is private information for} $ s_m $; while $ {ask}_m $ denotes the unit price (per RB) asked by $ s_m $ for providing services. 
Given that the actual resource supply can vary over time, we use $ R_m $ to quantify the number of available resources that $ s_m $ can offer during each transaction. Without loss of generality, $ R_m $ is modeled as a random variable obeying a Binomial distribution denoted by $ R_m\sim{\text{\textbf{Bin}}\left(\mathbbm{d}_m,\mathbbm{r}_m\right)} $ with parameters $ \mathbbm{d}_m $ and $ \mathbbm{r}_m $, where $ \text{Pr}\left\{R_m=k\right\}=\text{C}_{\mathbbm{d}_m}^k\mathbbm{r}_m^k\left(1-\mathbbm{r}_m\right)^{\mathsf{\mathbbm{d}_m-}k} $. Specifically, $ \mathbbm{r}_m $ represents the probability that one RB is available, while $ 1-\mathbbm{r}_m $, otherwise. 

\noindent\textbf{1) Utility and expected Utility of $s_m$:} The utility of seller $ s_m\in\bm{\mathcal{S}} $ involves two key aspects: \textit{(i)} the income obtained from members (not volunteers), e.g., the payment from attendant members and the penalty of absent members; and \textit{(ii)} the compensation from $ s_m $ to volunteers. Accordingly, the utility of $ s_m $ is expressed as:
\begin{equation}\label{key}
	\begin{aligned}
		&U_m^\mathsf{S} \left(r_m^\mathsf{s},q_{m,n}^\mathsf{b\to s},q_{m,n}^\mathsf{s\to b} \right)=
		\\&\sum_{b_n\in \bm{\mathcal{B}}}x_{m,n}t_n\alpha_n\left(M_n\left(r_m^\mathsf{s}-c_m\right)-(1-M_n) q_{m,n}^\mathsf{s\to b}\right)
		\\&+\sum_{b_n\in \bm{\mathcal{B}}}x_{m,n}t_n(1-\alpha_n)q_{m,n}^\mathsf{b\to s},
	\end{aligned}
\end{equation}
where $ r_m^\mathsf{s} $ denotes the unit reward (per RB) that $ s_m $ obtains for offering services. Similar to buyers, since it is challenging to ascertain the practical value of $ U_m^\mathsf{S} $ during the first stage, \textcolor{black}{we compute} its expected value as
\begin{equation}\label{key}
	\begin{aligned}
		&\overline{U_m^\mathsf{S}} \left(r_m^\mathsf{s},q_{m,n}^\mathsf{b\to s},q_{m,n}^\mathsf{s\to b} \right)=
		\\&\sum_{b_n\in \bm{\mathcal{B}}}x_{m,n}t_n\mathbbm{a}_n\left(\left(1-\mathbb{P}_n\right)\left(r_m^\mathsf{s}-c_m\right)-\mathbb{P}_n q_{m,n}^\mathsf{s\to b}\right)
		\\&+\sum_{b_n\in \bm{\mathcal{B}}}x_{m,n}t_n(1-\mathbbm{a}_n)q_{m,n}^\mathsf{b\to s}.
	\end{aligned}
\end{equation}
\noindent\textbf{2) Risk analysis of $s_m$:} Our OPDAuction is also designed to be risk-aware for sellers. Accordingly, we define the risk of $s_m$ receiving an unsatisfactory utility (abbreviated to ``SRisk'') as the probability that $ U_m^\mathsf{S} \left(r_m^\mathsf{s},q_{m,n}^\mathsf{b\to s},q_{m,n}^\mathsf{s\to b} \right) $ approaches too close or falls below its expectation, shown as 
\begin{equation}\label{key}
	\begin{aligned}
		&\mathcal{R}_m^\mathsf{SRisk} \left(r_m^\mathsf{s},q_{m,n}^\mathsf{b\to s},q_{m,n}^\mathsf{s\to b} \right)
		\\&=\text{Pr}\left(\frac{U_m^\mathsf{S} \left(r_m^\mathsf{s},q_{m,n}^\mathsf{b\to s},q_{m,n}^\mathsf{s\to b} \right)}{\overline{U_m^\mathsf{S}} \left(r_m^\mathsf{s},q_{m,n}^\mathsf{b\to s},q_{m,n}^\mathsf{s\to b} \right)}\leq\xi_2\right)
		\\&=\frac{\sum_{n=1}^\mathsf{|\bm{\mathcal{B}}|} (x_{m,n})^2 (t_n)^2 \mathbb{P}_n(1-\mathbb{P}_n)\left(M_n \mathbbm{c}_2-\mathbbm{c}_3\right)^2}{\left(\sum_{n=1}^\mathsf{|\bm{\mathcal{B}}|} x_{m, n} t_n \mathbb{P}_n\left(M_n \mathbbm{c}_2-\mathbbm{c}_3\right) - \mathbbm{c}_4\right)^2},
	\end{aligned}
\end{equation}
where $ \xi_2 $ represents a positive threshold coefficient approaching to 1. Also, we use constants $ \mathbbm{c}_2=\left(r_m^\mathsf{s}-c_m+q_{m,n}^\mathsf{s\rightarrow b}\right) $, $ \mathbbm{c}_3=q_{m,n}^\mathsf{s\rightarrow b}+q_{m,n}^\mathsf{b\rightarrow s} $, and $ \mathbbm{c}_4=\overline{U_m^\mathsf{S}}\xi_2-\sum_{b_n\in\bm{\mathcal{B}}}{x_{m,n}t_nq_{m,n}^\mathsf{b\rightarrow s}} $ to simplify the complicated calculations associated with (8), for notational simplicity. Derivation of $ \mathcal{R}_m^\mathsf{SRisk} $ is detailed in Appendix D.

\subsubsection{Modeling of the auctioneer}

An auctioneer plays a pivotal role in coordinating the auction process, including collecting bids/asks, determining winning sellers/members, designing long-term contracts, and optimizing the overbooking rate. The utility of the auctioneer in Stage I can generally be defined as the difference between total income received of sellers and total expenses paid from buyers, as expressed by:
\begin{equation}\label{key}
	\begin{aligned}
		&U^\mathsf{P}\left(p_n^\mathsf{b},r_m^\mathsf{s}\right)=
		\\&\sum_{b_n\in \bm{\mathcal{B}}}\sum_{s_m\in \bm{\mathcal{S}}}x_{m,n}t_n\left(\alpha_nM_n+\mu\left(1-\alpha_n\right)\right)\left(p_n^\mathsf{b}-r_m^\mathsf{s}\right),
	\end{aligned}
\end{equation}
where the $\mu$ in (9) represents a penalty factor, intentionally introduced for buyers and sellers. Regarding the penalties specified in long-term contracts, we have the following considerations: We set a penalty factor, $\mu$, whose value ranges between 0 and 1. Similarly, we further show the expected utility of the auctioneer as 
\begin{equation}\label{key}
	\begin{aligned}
		&\overline{U^\mathsf{P}}\left(p_n^\mathsf{b},r_m^\mathsf{s}\right)=
		\\&\sum_{b_n\in \bm{\mathcal{B}}}\sum_{s_m\in \bm{\mathcal{S}}}x_{m,n}t_n\left(\mathbbm{a}_n\left(1-\mathbb{P}_n\right)+1/2\left(1-\mathbbm{a}_n\right)\right)\left(p_n^\mathsf{b}-r_m^\mathsf{s}\right).
	\end{aligned}
\end{equation}

In subsequent sections, our algorithm design will systematically ensure that the auctioneer’s income remains nonnegative (also see Appendix C). Therefore, we will not address the associated risks for the auctioneer in detail.

\subsection{Design of OPDAuction}
This section details problem formulation and solution design of OPDAuction in Stage I.
\subsubsection{Problem formulation}
A fundamental and critical goal within an auction market is to maximize its social welfare. In our model, social welfare is defined by the collective utilities of three parties (e.g., buyers, sellers, auctioneer), as given by:
\begin{equation}\label{key}
	\begin{aligned}
		&U^\mathsf{SW}=\sum_{b_n\in \bm{\mathcal{B}}}U_n^\mathsf{B}+\sum_{s_m\in \bm{\mathcal{S}}}U_m^\mathsf{S}+U^\mathsf{P}
		\\&=\sum_{b_n\in \bm{\mathcal{B}}}\sum_{s_m\in \bm{\mathcal{S}}}x_{m,n}\alpha_nM_nt_n\left(v_{m,n}-c_m\right).
	\end{aligned}
\end{equation}

Interestingly, since our proposed OPDAuction implements an unique auction procedure prior to actual transactions, our emphasis during this stage relies on the expectation of social welfare, as illustrated by (12).
\begin{equation}\label{key}
	\begin{aligned}
		\overline{U^\mathsf{SW}}=\sum_{b_n\in \bm{\mathcal{B}}}\sum_{s_m\in \bm{\mathcal{S}}}x_{m,n}\mathbbm{a}_nt_n\left(1-\mathbb{P}_n\right)\left(v_{m,n}-c_m\right).
	\end{aligned}
\end{equation}

In this context, the primary goal of OPDAuction is to \textit{identify the winning pairs of sellers and buyers (i.e., determining the members for each seller, which involves figuring out all the feasible $ x_{m,n} $), establish their long-term contracts (e.g., $ \mathbb{C}_{m,n}^\mathsf{b\leftrightarrow s} $), and determine the appropriate overbooking rate (e.g., $ \lambda_m $).} These issues are encapsulated in the following optimization problem $ \bm{\mathcal{F}_1} $, given by (13).
\begin{equation}\label{key}
	\begin{aligned}
		&\bm{\mathcal{F}_1}:\mathop{\text{argmax}}\limits_{x_{m,n},\mathbb{C}_{m,n}^\mathsf{b\leftrightarrow s},\lambda_m}\overline{U^\mathsf{SW}}
	\end{aligned}
\end{equation}
\begin{subequations}{
		\begin{align}
			\text{s.t.}~~~~&\mathcal{R}_n^\mathsf{BRisk}\le\xi^\mathsf{M},\forall b_n\in\bm{\mathcal{B}},\tag{C1}\\
			&\mathcal{R}_n^\mathsf{VRisk}\le\xi^\mathsf{V},\forall b_n\in\bm{\mathcal{B}},\tag{C2}\\
			&\mathcal{R}_m^\mathsf{SRisk}\le\xi^\mathsf{S},\forall s_m\in\bm{\mathcal{S}},\tag{C3}\\
			&\sum_{b_n\in\bm{\mathcal{B}}}{x_{m,n}t_n}\le\left(1+\lambda_m\right)\overline{R_m},\ \forall\ s_m\in\bm{\mathcal{S}},\tag{C4}\\
			&\sum_{s_m\in\bm{\mathcal{S}}} x_{m,n}\le1,\ \forall\ b_n\in\bm{\mathcal{B}},\tag{C5}\\
			&x_{m,n}\in\left\{0,1\right\},\forall\ b_n\in\bm{\mathcal{B}},\forall\ s_m\in\bm{\mathcal{S}}.\tag{C6}
	\end{align}}
\end{subequations}
\setcounter{equation}{13}
Here, $ \xi^\mathsf{M} $, $ \xi^\mathsf{V} $, and $ \xi^\mathsf{S}$ are positive threshold coefficients for risk assessment, while constraints (C1)–(C3) manage risks for both buyers and sellers. Constraint (C4) ensures that overbooked resources do not exceed anticipated supply, (C5) maps each buyer to only one seller, and (C6) enforces the binary nature of $ x_{m,n} $. Problem $ \bm{\mathcal{F}_1} $ is a mixed integer linear program (MILP) with discrete variables (e.g., $ x{m,n} $) and continuous variables (e.g., elements of $ \mathbb{C}_{m,n}^\mathsf{b\leftrightarrow s} $ and $ \lambda_m $), and is generally NP-hard \cite{b21}. The probabilistic risk constraints (C1)–(C3) further increase complexity, making $ \bm{\mathcal{F}_1} $ challenging to solve. To address this, we propose OPDAuction to identify feasible buyer-seller matchings—determining members for each seller and their contract terms (e.g., payments and penalties). OPDAuction also adjusts the overbooking rate to mitigate potential losses and improve expected social welfare under market dynamics.

\subsubsection{Solution design}
The proposed OPDAuction in Stage~I is sealed-bid, private, and collusion-free, meaning all sellers and buyers submit their asks and bids simultaneously to the auctioneer. It involves three key subproblems: $\bm{\mathcal{F}_{1a}}$ for member selection (winning seller-buyer determination), $\bm{\mathcal{F}_{1b}}$ for designing long-term contracts, and $\bm{\mathcal{F}_{1c}}$ for optimizing the overbooking rate to enhance resource provision performance. Specifically, $\bm{\mathcal{F}_{1a}}$ and $\bm{\mathcal{F}_{1b}}$ find feasible buyer-seller assignments and establish contract terms (e.g., payments and penalties), while $\bm{\mathcal{F}_{1c}}$ accounts for dynamic and uncertain resource demand/supply, allowing sellers to reserve more resources for members than their theoretical capacity. For clarity, Fig.~3 provides an overview of OPDAuction in the considered market. Strategies for the three subproblems are detailed in Algorithms~1–3: Algorithms~1 and 2 address $\bm{\mathcal{F}_{1a}}$ and $\bm{\mathcal{F}_{1b}}$, respectively, while Algorithm~3 tests various overbooking rates to finalize feasible long-term contracts for $\bm{\mathcal{F}_{1c}}$. In the following, we show how we address these sub-problems.
\begin{figure}[b!]
	\centering
	\includegraphics[trim=0cm 0cm 0cm 0cm, clip, width=\columnwidth]{Algorithm_1.pdf}
	\caption{Flow diagram of our proposed OPDAuction.}
	\label{fig_1}
\end{figure}

\noindent $\bullet$ \textbf{Member Determination (abbreviated as OPDAuction-MemberD, Algorithm 1):} OPDAuction-MemberD establishes seller-buyer mappings through three coordinated phases. To facilitate better reading, more details can be found in Appendix A. 

\noindent\textit{Phase 1: Participant sorting.} Buyers are sorted in $\mathcal{L}_b$ by non-ascending average bid $bid^\mathsf{Avg}_n = \frac{1}{|\mathcal{S}|}\sum bid_{m,n}$ (see Appendix A, Eq. 18), while sellers are ordered in $\mathcal{L}_s$ by non-descending $ask_m$ (see Appendix A, Eq. 21). This creates bid-ask aligned sequences:
\begin{equation}
	\mathcal{L}_b: bid^{\mathsf{List}}_1 \succ \cdots \succ bid^{\mathsf{List}}_{|\mathcal{B}|}, \quad \mathcal{L}_s: ask^{\mathsf{List}}_1 \prec \cdots \prec ask^{\mathsf{List}}_{|\mathcal{S}|},
\end{equation}

\noindent\textit{Phase 2: Key index identification.} We determine pivotal indices $(k_b^*,k_s^*)$ that should satisfying the following (15). This ensures budget balance through nested search from $k_b=|\mathcal{B}|-1$ to $k_b=1$ (lines 7-18). The optimal $(k_b^*,k_s^*)$ maximizes feasible pairs while keeping $MaxNumB \leq \sum_{m=1}^\mathsf{k_s^*} \overline{R_m}(1+\lambda_m)$.
\begin{equation}
	\begin{cases}
		bid^{\mathsf{List}}_{k_b^*+1} \geq ask^{\mathsf{List}}_{k_s^*+1} \\
		bid^{\mathsf{List}}_{k_b^*+2} < ask^{\mathsf{List}}_{k_s^*+2} \text{ or }k_b^\ast+1=\left|\bm{\mathcal{B}}\right|\ \text{or} \ k_s^\ast+1=\left|\bm{\mathcal{S}}\right|
	\end{cases}
\end{equation}

\noindent\textit{Phase 3: Dynamic matching.} Top $k_b^*$ buyers and $k_s^*$ sellers enter 0-1 knapsack-based matching (lines 21-32). For each seller $s_m$ with capacity $h_m = \overline{R_m}(1+\lambda_m)$, we solve:
\begin{equation}
	\max \sum (bid_{m,n} - ask_m)x_{m,n}, \quad \text{s.t.} \sum t_nx_{m,n} \leq h_m,
\end{equation}
where $x_{m,n} \in \{0,1\}$, and we use dynamic programming (DP) with weights $v=t_n$ and values $w=bid_{m,n}$. Unmatched buyers retain eligibility for other sellers through bid list updating (zeroing current bids). Individual rationality is enforced by requiring $bid_{m,n} \geq bid^\mathsf{Avg}_n$ for successful matches.

\begin{algorithm}[]
	{\footnotesize  
		\caption{\small{OPDAuction-MemberD}}
		\LinesNumbered 
		{\bf{Input :}} 
		$ \bm{\mathcal{B}},t_n,v_{m,n},bid_{m,n},\mathbbm{a}_n,\bm{\mathcal{S}},c_m,{ask}_m,\mathbbm{d}_m,\mathbbm{r}_m $
		
		{\bf{Output :}} 
		$ \bm{X}^* $
		
		{\bf{Initialization :}} 
		$ bid^\mathsf{Avg}_n \leftarrow mean({bid}_{m,n}),MaxNumB \leftarrow 0, k_b^\ast \leftarrow 0, k_s^\ast \leftarrow 0,U_{m',n'}\leftarrow{bid}_{m',n'}-{ask}_{m'}, U\leftarrow\left[U_{1,1}\cdots U_{1,n'};\cdots;U_{m',1}\cdots U_{m',n'}\right]$
		
		\textbf{\# Phase 1: List generation} 
		
		\text{Generating lists $\mathcal{L}_b$ and $\mathcal{L}_s$}\
		
		\textbf{\# Phase 2: Key index determination} 
		
		\For{
			each $ n' =\left|\bm{\mathcal{B}}\right|-1, ..., 1 $
		}{
			\For{
				each $ m' = \left|\bm{\mathcal{S}}\right|-1, ..., 1 $
			}{
				\If{$ \overline{{bid}_{n'+1}} \ge {ask}_{m'+1} $ and $ (n' + 1 = \left|\bm{\mathcal{B}}\right| $ or $ m' + 1 = \left|\bm{\mathcal{S}}\right| $ or  $ \overline{{bid}_{n'+2}} < {ask}_{m'+2}) $}
				{
					$ R\leftarrow\sum_{k=1}^\mathsf{m'}\overline{R_{k}},\forall s_{k}\in\bm{\mathcal{S}} $
					
					\For{each $ n' = 1, ..., n $}{
						$ R\leftarrow R-t_{n'},\forall b_{n'}\in\bm{\mathcal{B}} $
						
						\If{R<0}{
							\If{$ n'> MaxNumB + 1 $}
							{$ MaxNumB \leftarrow n'-1 $
								
								$ k_b^\ast\leftarrow n'-1,k_s^\ast\leftarrow m' $
								
								break}
						}
					}
				}

			}
		}
		\textbf{\# Phase 3: Buyer-Seller matching} 

		\textbf{$ h_{m'}\leftarrow\overline{R_{m'}}\times\left(1+\lambda_{m'}\right) $}\
		
		\For{each $ m' = k_s, ..., 1 $}
		{$ UMAX\leftarrow max\left(U\right) $\
			
			\For{each $ n' = k_b^\ast, ..., 1 $}
			{
				\If{$ U_{m',n'}={UMAX}_{n'} $}
				{$ v\leftarrow\left[v,t_{n'}\right] $\
					
					$ w\leftarrow\left[w,{bid}_{m',n'}\right] $\
					
					$ tt\leftarrow 01\text{KP}\left(c,v,w\right) $\
					
					\For{each $  n' = k_b^\ast, ..., 1 $}{
						\If{$ tt_{n'}=0 $}
						{$ U_{m',n'}\leftarrow0 $}
						\ElseIf{$ {bid}_{m',n'}\geq\overline{{bid}_{n'}} $}{
							$ X_{m',n'}\leftarrow1 $\
							
							$ h_{m'}\leftarrow h_{m'}-t_{n'} $\
							
							\textbf{Return} $\bm{X}^*$
						}	
					}
				}
			}
		}
					
}\end{algorithm}

\noindent $\bullet$ \textbf{Long-term contract design (abbreviated as OPDAuction-ContractD, Algorithm 2):} To resolve subproblem $\bm{\mathcal{F}_{1b}}$, we develop OPDAuction-ContractD via borrowing the idea of binary search to finalize payments between matched buyer-seller pairs identified previously (more details are found in Appendix A). This algorithm replaces $bid^{\mathsf{List}}_{n'}$ with ${bid}_{m',n'}$ during auction sorting (lines 5-8), ensuring contract negotiations reflect genuine buyer valuations for specific seller services. This adjustment preserves truthful bidding representation while streamlining price discovery.
For each winning buyer $b_{n'} \in \mathcal{L}_b$, we establish price bounds using $bid^{\mathsf{List}}_{k_b+1}$ (lower) and $bid^{\mathsf{List}}_n$ (upper), guaranteeing final payments stay between seller asks and buyer bids (line 8). The core mechanism employs binary search \cite{b25} to efficiently identify the minimal acceptable price within this range (lines 9-15), balancing buyer affordability with seller compensation requirements. A parallel process handles seller pricing (lines 17-32), where the lower bound $low = {ask}_{m'}$ ensures sellers receive at least their ask price while optimizing buyer costs.
OPDAuction-ContractD concludes by outputting payment vectors (line 20) that capture all financial settlements, having systematically reduced negotiation complexity through constrained search spaces and bid-ask alignment. This dual-boundary approach maintains individual rationality while preventing price inflation, crucial for maintaining auction equilibrium. 

\begin{algorithm}[]
	{\footnotesize  
		\caption{\small{OPDAuction-ContractD}}
		\LinesNumbered 
		
		{\bf{Input :}} 
		$ \bm{\mathcal{B}},t_n,v_{m,n},bid_{m,n},\mathbbm{a}_n,\bm{\mathcal{S}},c_m,{ask}_m,\mathbbm{d}_m,\mathbbm{r}_m, \bm{X}^* $
		
		{\bf{Output :}} 
		$ \bm{p}^\mathsf{\bm{b}*},\bm{r}^\mathsf{\bm{s}*}, \bm{q^\mathsf{s\to b}}, \bm{q^\mathsf{b\to s}} $
		
		\textbf{\# Phase 1: Initialization}\\ 
		\textbf{$ p_{n'}^\mathsf{b} \leftarrow 0, r_{m'}^\mathsf{s} \leftarrow 0 $}\\
		\For{
			each $ n' =\left|\bm{\mathcal{B}}\right|-1, ..., 1 $
		}{
			\For{
				each $ m' = \left|\bm{\mathcal{S}}\right|-1, ..., 1 $
			}{
				\If{$ x_{m',n'}=1 $}
				{$ \overline{{bid}_{n'}}\leftarrow{bid}_{m',n'} $}
			}
		}
		\textbf{\# Phase 2: Buyer's payment determination}\\ 
		\For{
			each $ n' =\left|\bm{\mathcal{B}}\right|-1, ..., 1 $
		}{
			\If {$ \sum_{m'\in\bm{\mathcal{S}}} x_{m',n'}=1 $
			}{
				$ high\ \leftarrow{bid}_{n'},\ low\ \leftarrow{bid}_{k_b^*+1},\ tmp\ \leftarrow{bid}_{n'} $
				
				\While {low\ <\ high
				}{
					$ {bid}_{n'}\ \leftarrow\ (high\ +\ low)/2 $
					
					$ \bm{X^\mathsf{Temp}} \leftarrow$ OPDAuction-MemberD 
					
					\If {$ \sum_{m'\in\bm{\mathcal{S}}} x_{m',n'}^\mathsf{Temp}=1 $
					}{
						$ high\ \leftarrow\ {bid}_{n'} $
					}
					\Else{$ low\ \leftarrow\ {bid}_{n'} $
						
						$ p_{n'}^\mathsf{b}\ \leftarrow\ {bid}_{n'},\ {bid}_{n'}\ \leftarrow\ tmp,\ q_{m',n'}^\mathsf{s\to b}\ \leftarrow\ \mu p_{n'}^\mathsf{b} $
					}
				}	
			}
		}
		\textbf{\# Phase 3: Seller's reward determination}\\ 
		\For{
			each $ m' = \left|\bm{\mathcal{S}}\right|-1, ..., 1 $
		}{
			\If {$ \sum_{n'\in\bm{\mathcal{B}}} x_{m',n'}=1 $
			}{
				$ high\ \leftarrow{ask}_{m'},\ low\ \leftarrow{ask}_{k_s^*+1},\ tmp\ \leftarrow{ask}_{m'} $
				
				\While{low\ <\ high
				}{
					$ {ask}_{m'}\ \leftarrow\ (high\ +\ low)/2$
					
					$ \bm{X^\mathsf{Temp}} \leftarrow$ OPDAuction-MemberD
					
					\If {$ \sum_{n'\in\bm{\mathcal{B}}} x_{m',n'}^\mathsf{Temp} =1 $
					}{
						$ low\ \leftarrow\ {ask}_{m'} $
					}
					\Else{$ high\ \leftarrow\ {ask}_{m'} $
						
						$ r_{m'}^\mathsf{s}\leftarrow {ask}_{m'}, {ask}_{m'} \leftarrow tmp, q_{m',n'}^\mathsf{b\to s} \leftarrow \mu r_{m'}^\mathsf{s} $
					}
				} 
			}
		}
		\Return{$ \bm{p}^\mathsf{\bm{b}*},\bm{r}^\mathsf{\bm{s}*}, \bm{q^\mathsf{s\to b}}, \bm{q^\mathsf{b\to s}}  $.}

	}				
\end{algorithm}

\noindent\textcolor{black}{$\bullet$ \textbf{An illustrative numerical example:} To illustrate the operation of Algorithms 1-2, we consider 5 sellers and 5 buyers: Sellers $\{s_1,\dots,s_5\}$ have capacities $\{5,4,6,3,7\}$, idle probabilities $\{0.8,0.7,0.9,0.6,0.85\}$, and ask prices $\{2.0,2.2,1.8,2.5,1.9\}$; buyers $\{b_1,\dots,b_5\}$ have demands $\{3,2,4,1,3\}$, attendance probabilities $\{0.9,0.8,0.7,0.6,0.5\}$, and bid prices $\{3.0,2.8,2.5,2.2,2.0\}$. With overbooking rate $\lambda^*=0.2$, Algorithm 1 first sorts buyers by bid density ($b_1 > b_2 > b_3 > b_4 > b_5$) and sellers by ask price ($s_3 < s_5 < s_1 < s_2 < s_4$). It determines $k_b^*=4$ buyers and $k_s^*=2$ sellers for matching. The optimal matching pairs $b_1$-$b_2$ with $s_3$ (utility 5.6) and $b_3$-$b_4$ with $s_5$ (utility 2.7). Algorithm 2 then determines payments: $b_1$ pays 2.07/unit (vs. bid 3.0), $b_2$ pays 2.15/unit (vs. bid 2.8), $b_3$ pays 2.05/unit (vs. bid 2.5), $b_4$ pays 2.02/unit (vs. bid 2.2); $s_3$ receives 1.96/unit (vs. ask 1.8), $s_5$ receives 1.98/unit (vs. ask 1.9). This satisfies individual rationality (all participants have positive utility), truthfulness (no incentive to misreport), and budget balance (auctioneer profit = 1.03 > 0). Details are moved to Appendix E.}

\noindent$\bullet$ \textbf{Overbooking rate optimization (abbreviated as OPDAuction-OverbookROpt, Algorithm 3):} To balance dynamic resource demands and market efficiency, OPDAuction-OverbookROpt systematically determines the optimal overbooking rate $\lambda^*$ through risk-aware parallel evaluation (detailed explanations are provided by Appendix A). The process begins by initializing candidate rates $\Lambda = \{1\%,(1+\Delta\lambda)\%,...,100\%\}$ with precomputed risk thresholds ($\xi^\mathsf{S}$, $\xi^\mathsf{B}$, $\xi^\mathsf{V}$). For each candidate $\lambda$, we first calculate sellers' adjusted capacities $\overline{R_m} = c_m \times (1+\lambda)$ and validate basic capacity constraints through binary search.
The algorithm then performs parallel evaluations by simultaneously: \textit{(i)} executing OPDAuction-MemberD to identify winning pairs under current $\lambda$, \textit{(ii)} determining contract terms $[p^\mathsf{b},r^\mathsf{s}]$ via OPDAuction-ContractD, and \textit{(iii)} assessing risks through $\mathcal{R}^\mathsf{SRisk}_m = f(\mathbbm{d}_m, \overline{R_m})$ and $\mathcal{R}^\mathsf{BRisk}_n = g(\mathbbm{a}_n, t_n)$. Candidates exceeding risk thresholds are automatically pruned while recording feasible solutions in $\bm{U^\#}$, $\bm{p}$, $\bm{r}$, and $\bm{X^\#}$. Meanwhile, social welfare $USW = \sum_{m,n} (bid_{m,n}-ask_m)x_{m,n}$ guides search direction prioritization.
Finally, the golden-section search refines \cite{b25-2} neighborhood solutions when multiple $\lambda$ values demonstrate comparable welfare, selecting configurations that maximize both welfare and successful matches. This coordinated evaluation of risk profiles and economic efficiency ensures balanced market participation through optimal contracts, maintaining computational efficiency. 

The above Algorithms 1-3 collectively form our designed OPDAuction, with more details shown in Appendix A. From OPDAuction, we obtain long-term contracts between proper buyers and sellers, which will be implemented during Stage II directly without any further negotiation, thus facilitating service delivery.

\begin{algorithm}[]
	{\footnotesize
		\caption{\small{OPDAuction-OverbookROpt}}
		\LinesNumbered
		
		{\bf{Input :}} 
		$\bm{\mathcal{B}},t_n,v_{m,n},bid_{m,n},\mathbbm{a}_n,\bm{\mathcal{S}},c_m,{ask}_m,\mathbbm{d}_m,\mathbbm{r}_m$
		
		{\bf{Output :}} 
		$\bm{X}^*,\bm{p}^\mathsf{\bm{b}*},\bm{r}^\mathsf{\bm{s}*},\lambda^*$
		
		\textbf{\# Phase 1: Preprocessing} \\ 
		Precompute $\xi^\mathsf{S}, \xi^\mathsf{B}, \xi^\mathsf{V}$ thresholds \\
		Initialize candidate set $\Lambda \gets \{1\%,(1+\Delta\lambda)\%,...,100\%\}$ (sorted) \\
		Precompute seller capacities $\overline{R_m} \gets c_m \times (1+\lambda), \forall \lambda \in \Lambda$
		
		\textbf{\# Phase 2: Binary Search for Optimal $\lambda$} \\
		$low \gets 1$, $high \gets 100$, $best \gets \emptyset$ \\
		\While{$low \leq high$}{
			$mid \gets \lfloor (low+high)/2 \rfloor$ \\
			$\lambda \gets \Lambda[mid]$ \\
			
			\textbf{\# Pruned Candidate Evaluation} \\
			\If{$\lambda$ violates capacity constraints}{
				Update search direction \\
				\textbf{continue}
			}
			
			$\bm{X} \gets$ OPDAuction-MemberD($\lambda$) \\
			$[p^\mathsf{b},r^\mathsf{s}] \gets$ OPDAuction-ContractD($\bm{X}$) \\
			
			\textbf{\# Parallel Risk Assessment} \\
			Compute $\mathcal{R}^\mathsf{SRisk}, \mathcal{R}^\mathsf{BRisk}, \mathcal{R}^\mathsf{VRisk}$ in parallel:
			\begin{itemize}
				\item[] Seller risks: $\mathcal{R}^\mathsf{SRisk}_m = f(\mathbbm{d}_m, \overline{R_m})$ 
				\item[] Buyer risks: $\mathcal{R}^\mathsf{BRisk}_n = g(\mathbbm{a}_n, t_n)$
			\end{itemize}
			
			\textbf{\# Early Termination Checks} \\
			\If{$\exists m \ \mathcal{R}^\mathsf{SRisk}_m > \xi^\mathsf{S}$}{
				Mark $\lambda$ as infeasible \\
				Update $high \gets mid-1$
			}
			\Else{
				Calculate $USW \gets \sum_{m,n} (bid_{m,n}-ask_m) \cdot x_{m,n}$ \\
				Update $best \gets \argmax(best, USW)$ \\
				Update $low \gets mid+1$
			}
		}
		
		\textbf{\# Phase 3: Refinement Search} \\
		\If{$best$ has multiple candidates}{
			Perform golden-section search in neighborhood \\
			Verify feasibility constraints for final selection
		}
		
		\textbf{\# Final Output} \\
		\Return optimal $\bm{X}^*, \bm{p}^\mathsf{\bm{b}*}, \bm{r}^\mathsf{\bm{s}*}, \lambda^*$ from $best$
	}
\end{algorithm}

\section{Stage II. Realtime Backup Double Auction (RBDAuction)}
Although constraints associated with risk control in Stage I (constraints (C1)-(C3)) can make the risks within acceptable thresholds, the market can still face unexpected situations due to its dynamic nature, we thus proceed to Stage II, which takes place during actual resource transactions. In this stage, members and their respective sellers fulfill their pre-signed contracts. This involves either utilizing and paying for edge services, or serving as volunteers in exchange for compensation (when a member is participated), or paying penalty when a member is absent. 
To further enhance practical social welfare, we \textit{encourage volunteers, guests, and sellers with surplus resources to participate in an additional double auction process,} referred to as the RBDAuction. Apparently, auction decisions of RBDAuction are based on the current network/market conditions, such as available resource supply and the presence of members and guests who actively attend (i.e., $\alpha_j=1$).

In RBDAuction, we first filter out a matrix $ \bm{X}^\mathsf{**} $ from the obtained $\bm{X}^*$ during Stage I that records the seller-member pairs that have successfully completed resource transactions. Due to their dynamic participation, we consider a new set of buyers $ \tilde{\bm{\mathcal{B}}}=\left\{\tilde{b}_1,\ldots,\tilde{b}_j,\ldots,\tilde{b}_{\left|\bm{\mathcal{B}}\right|}\right\} $, which includes guests and volunteers. Also, we re-denote the set of sellers $ \tilde{\bm{\mathcal{S}}}=\left\{\tilde{s}_1,\ldots,\tilde{s}_i,\ldots,\tilde{s}_{\left|\bm{\mathcal{S}}\right|}\right\} $, consisting of sellers with remaining available resources. After completing the supplementary RBDAuction in Stage II, we can obtain new matching pairs as recorded in matrix $ \tilde{\bm{X}} =\left\{{\tilde{x}}_{i,j}|\tilde{b}_j\in\tilde{\bm{\mathcal{B}}},\tilde{s}_i\in\tilde{\bm{\mathcal{S}}}\right\} $, where ${\tilde{x}}_{i,j}=1$ (or $0$) indicates a partnership (or not) between $ \tilde{s}_i $ and $ \tilde{b}_j $ during each transaction. Apparently, matrices $ \bm{X}^\mathsf{**} $, $ \tilde{\bm{X}} $, and their corresponding monetary elements (e.g., price) collectively form the trading decision of RBDAuction.

Accordingly, RBDAuction works for obtaining proper trading pairs between buyers with unmet resource demands and sellers with surplus resource supply, with the aim of maximizing the practical social welfare denoted by $\tilde{U}^\mathsf{SW}$.
Similar to Stage I, we formulate the optimization problem in Stage II as the following $\bm{\mathcal{F}_2}$:
\begin{equation}\label{key}
	\begin{aligned}
		&\bm{\mathcal{F}_2}:\underset{\tilde{x}_{i,j},\tilde{p}^\mathsf{b},\tilde{r}^\mathsf{s}}{\arg\max}~~{\tilde{U}^\mathsf{SW}}
	\end{aligned}
\end{equation}
\begin{equation}\label{key}\tag{C7}
	\begin{aligned}
		&\text{s.t.} \sum_{\tilde{b}_j\in\tilde{\bm{\mathcal{B}}}}{{\tilde{x}}_{i,j}\tilde{t}_j}\le{\tilde{R}_{i}},\ \forall\ \tilde{s}_i\in\tilde{\bm{\mathcal{S}}},
	\end{aligned}\nonumber
\end{equation}
\begin{equation}\label{key}\tag{C8}
	\begin{aligned}
		&\sum_{\tilde{s}_i\in\tilde{\bm{\mathcal{S}}}} {\tilde{x}}_{i,j}\le1,\ \forall\ \tilde{b}_j\in\tilde{\bm{\mathcal{B}}},
	\end{aligned}\nonumber
\end{equation}
\begin{equation}\label{key}\tag{C9}
	\begin{aligned}
		&{\tilde{x}}_{i,j}\in\left\{0,1\right\},\forall\ \tilde{b}_j\in\tilde{\bm{\mathcal{B}}},\forall\ \tilde{s}_i\in\tilde{\bm{\mathcal{S}}}.
	\end{aligned}\nonumber
\end{equation}
where $\tilde{t}_j$ represents the resource block demand of buyer $\tilde{b}_j$, and $\tilde{R}_{i}$ signifies the remaining idle resources of seller \( \tilde{s}_i \). The variables and parameters such as $\tilde{U}^\mathsf{SW}$, $\tilde{t}_j$, and $\tilde{R}_{i}$ mentioned above will be defined and modeled in Appendix B. Constraint (C7) restricts the remaining practical resources. Constraint (C8) shows that one buyer can only be matched to at most one seller, while constraint (C9) describes the binary nature of ${\tilde{x}}_{i,j}$. To facilitate RBDAuction, we first determine the buyers who can be engaged. Then, the winning seller-buyer determination and their pricing can be figured out by using similar algorithms to Algorithms 1-2 introduced in OPDAuction, to keep the consistency and fairness of the two stages. Thus, as constrained by limited space and to achieve a better readability, we omit their details here. Instead, we show the pseudo-code of RBDAuction as well as the detailed analysis in Appendix B.

\section {Key Auction Property Analysis}
Proofs of key properties
such as individual rationality, truthfulness, budget-balance,
and computational efficiency of the proposed two-stage double
auction are moved to Appendix C, to achieve better coherence during reading. Moreover, since we offer an unique view of integrating pre-double auction, risks analysis and overbooking, the above properties can be different from conventional auctions, also representing one of the most interesting parts of this paper.

\section{Evaluations}
This section conducts comprehensive experiments to validate the performance of our proposed two-stage double auction (referred to as ``TwoSAuction'' for notational simplicity), using MATLAB 2022b.
\subsection{Simulation Setup and Evaluation Metrics}
Most existing papers only focus on manual numerical parameter settings in experiments, such as \cite{b8}, \cite{b13}, \cite{b15}. Differently, without loss of generality, we incorporate both real-world data and numerical data to guide key parameter settings in our experiments, since this paper offers an unique view of auction design, leading to that some information are hard to get from real-world datasets, e.g., price and penalty. Accordingly, \textcolor{black}{to capture diverse market sizes,} key parameters are considered as: $\left|\bm{\mathcal{B}}\right|\in \left[50,100,150,200\right]$, $ t_n\in \left[0,10\right] $, $v_{m,n}\in \left[0,10\right]$, $ \mathbbm{a}_n\in\left[50\%,100\%\right]$, $\left|\bm{\mathcal{S}}\right|\in \left[10,15,20,25\right]$, $R_m\in \left[1,100\right]$, $ r_m\in\left[0,1\right] $, $ c_m\in\left[0,10\right] $, $\xi^\mathsf{S}=\xi^\mathsf{B}=\xi^\mathsf{V}=50\%$. 

\color{black}
Note that our simulation adopts a hybrid parameter setup: observable factors such as user mobility are directly measured, while subjective or strategic parameters (e.g., private valuations, bidding/asking prices, penalty factors) are synthetically generated within reasonable ranges reported in the literature. These economic and behavioral parameters are crucial for auction design. To model the primary source of uncertainty in edge networks, namely \textit{dynamic buyer participation}, we employ the \textit{Chicago Taxi Trips dataset}, focusing on a downtown hotspot (latitude: 41.38°N–41.40°N, longitude: 87.35°W–87.33°W). Each simulated buyer (vehicle or mobile user) is assigned a participation probability ($\mathbbm{a}n$) per time slot, derived from the empirical frequency of the corresponding vehicle’s presence, thereby capturing spatio-temporal dynamics. Other parameters, including resource demands ($t_n \in [0, 10]$), unit valuations ($v_{m,n} \in [0, 10]$), seller costs ($c_m \in [0, 10]$), and capacities ($R_m \in [1, 100]$), follow standard ranges widely used in edge computing and auction literature. While not tied to a single dataset, these values are guided by empirical studies and established modeling practices, ensuring that the simulation environment is both realistic and comparable to prior work \cite{b28}.

\color{black}
Moreover, for default clauses, a penalty factor ($\mu$) ranging from 0 to 1 is applied. For example, if a buyer fails to participate in a transaction, the matched seller receives compensation set at $\mu r_m^\mathsf{s}$. Conversely, if the seller fails to provide promised resources, buyers volunteering to forgo services due to overbooking are compensated at $\mu p_n^\mathsf{b}$. This structure ensures fairness, balancing the interests of both parties by addressing discrepancies between contractual commitments and transaction outcomes.
To offer thorough evaluations, we also employ the following key metrics:

\noindent 
$\bullet$ \textit{Social welfare (SW)}: This metric quantifies the collective utility of all participants involved in the auction, further reflecting the economic efficiency.

\noindent 
$\bullet$ \textit{Time consumed by auction decision-making (TimeADM)}: This metric measures the time required to finalize the auction decisions, including identifying the winning seller-buyer pairs while determining trading prices. It is indicative of the time efficiency and practicality of the auction process.


\noindent 
$\bullet$ \textit{Property Analysis}: This assessment verifies crucial design properties of an auction such as truthfulness and individual rationality, via simulations.

The above metrics provide a comprehensive assessment of TwoSAuction. In the following, performance is evaluated in two steps: \textit{(i)} comparing with conventional auctions to highlight improvements (Section 7.2), and \textit{(ii)} assessing the impact of overbooking by comparison with methods that do not allow it (Section 7.3).

\subsection{Performance Evaluation vs. General Auction methods}
We first introduce the following comparative auction mechanisms as general benchmark methods\footnote{\textcolor{black}{Clarification on algorithm origins: Benchmark methods evaluate TwoSAuction. CRDAuction follows conventional real-time double auctions \cite{b25}, adapted to a multi-seller, multi-buyer setting with heterogeneous valuations. SSPDAuction tests Stage~I alone, while VRAuction, CRAuction, and RSAuction use greedy heuristics prioritizing buyer value, seller cost, and resource supply, respectively, to assess simplified strategies.}}: 

\noindent
$\bullet$ \textit{Conventional real-time double auction (CRDAuction):} CRDAuction embodies a prevalent approach in existing resource trading market, aiming to optimize the practical social welfare based on real-time network/market conditions.

\noindent
$\bullet$ \textit{Single-stage pre-double auction (SSPDAuction):} SSPDAuction promotes the engagement of all sellers and buyers in long-term contracts, mitigating possible risks prior to transactions without backup plans, i.e., implement Stage I only.

\noindent
$\bullet$ \textit{Value-raising-preferred auction (VRAuction):} Buyers prioritize purchasing resources from sellers who offer higher unit values (e.g., $v_{m,n}$).

\noindent
$\bullet$ \textit{Cost-reduction-driven auction (CRAuction):} Sellers prioritize serving buyer with lower unit service costs (e.g., $c_m$).

\noindent
$\bullet$ \textit{Resource supply-promoted auction (RSAuction):} Buyers show a preference to sellers who possess more sufficient resources, seeking to ensure their resource demands (e.g., $d_m$).


\begin{figure*}[t]
	\centering
	\begin{minipage}[t]{0.329\textwidth} 
		\subfigure[SW]{\includegraphics[trim=0.7cm 0cm 0cm 0cm, clip, width=0.45\textwidth,height=0.39\textwidth]{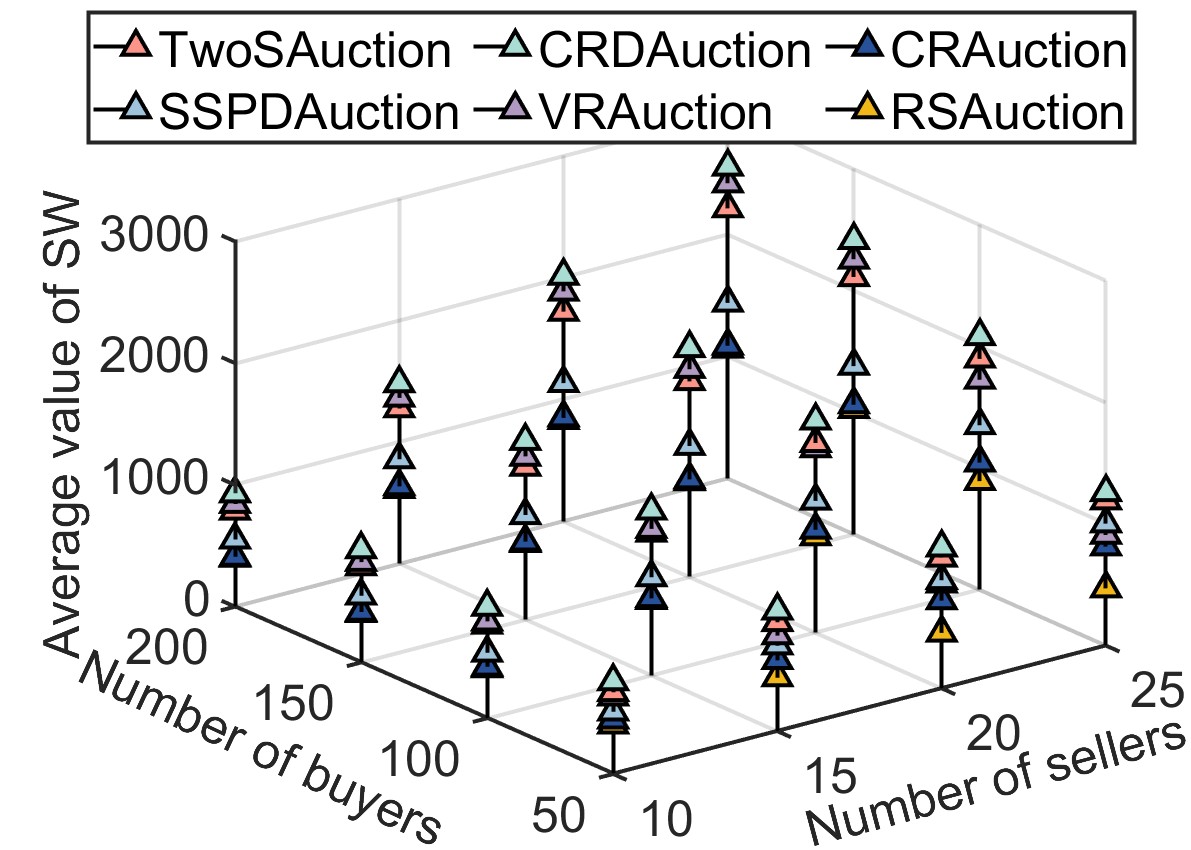}}
		\subfigure[TimeADM]{\includegraphics[trim=0.5cm 0cm 0.15cm 0cm, clip, width=0.45\textwidth,height=0.39\textwidth]{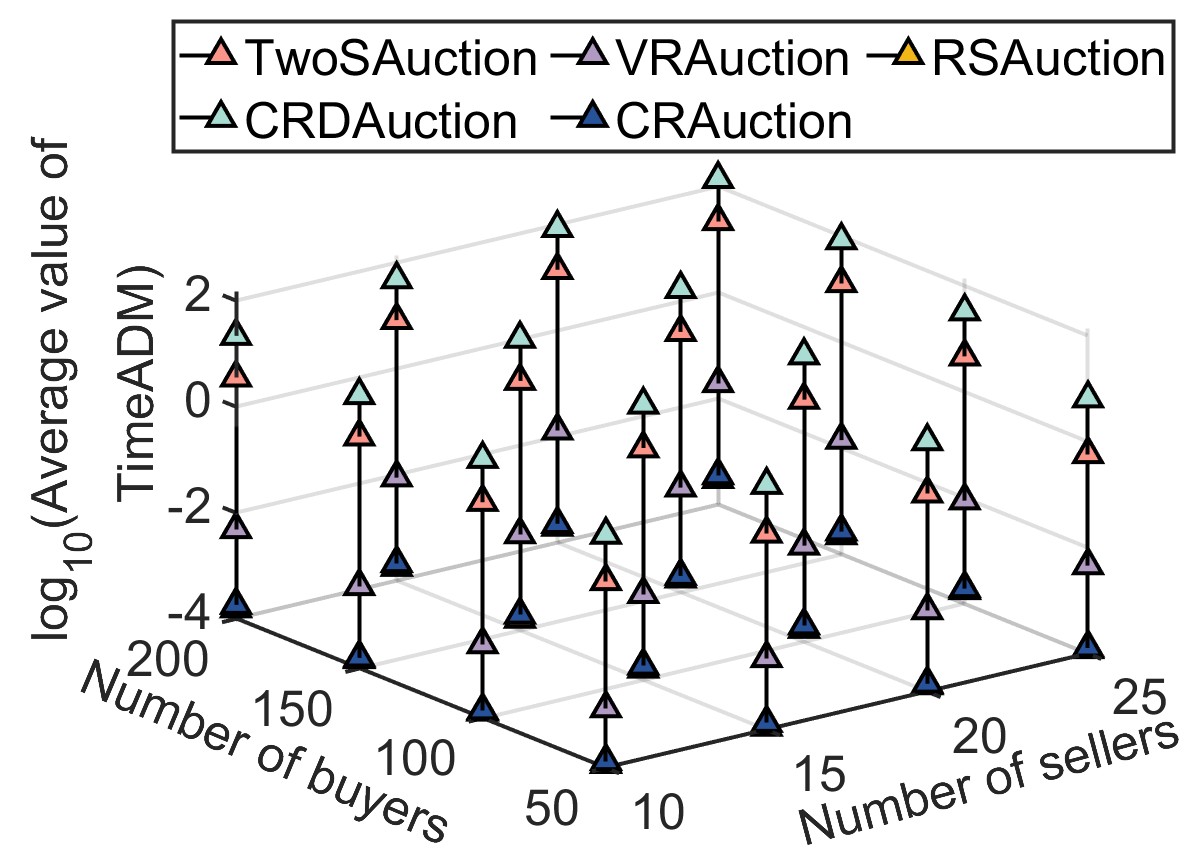}}
		\caption{Performance comparison.}
		\label{fig:main1}
	\end{minipage}
	\hfill 
	\begin{minipage}[t]{0.329\textwidth}
		\subfigure[Buyers]{\includegraphics[trim=0cm 0cm 0cm 0cm, clip, width=0.45\columnwidth,height=0.4\textwidth]{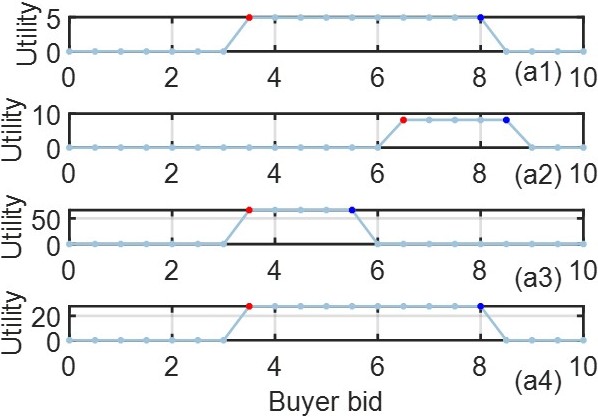}}
		\subfigure[Sellers]{\includegraphics[trim=0cm 0cm 0cm 0cm, clip, width=0.45\columnwidth,height=0.39\textwidth]{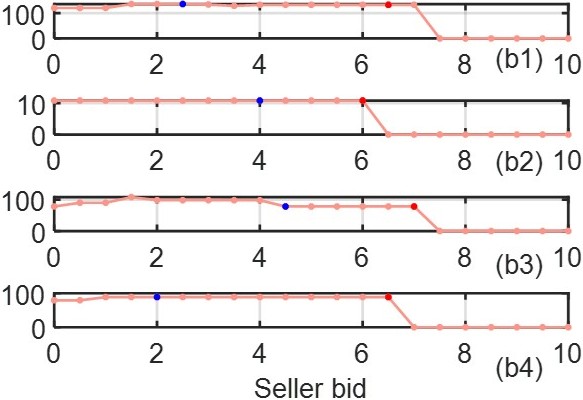}}
		\caption{Truthfulness.}
		\label{fig:main2}
	\end{minipage}
	\hfill
	\begin{minipage}[t]{0.329\textwidth}
\subfigure[Sellers]{\includegraphics[trim=0cm 0cm 0cm 0cm, clip, width=0.45\textwidth,height=0.38\textwidth]{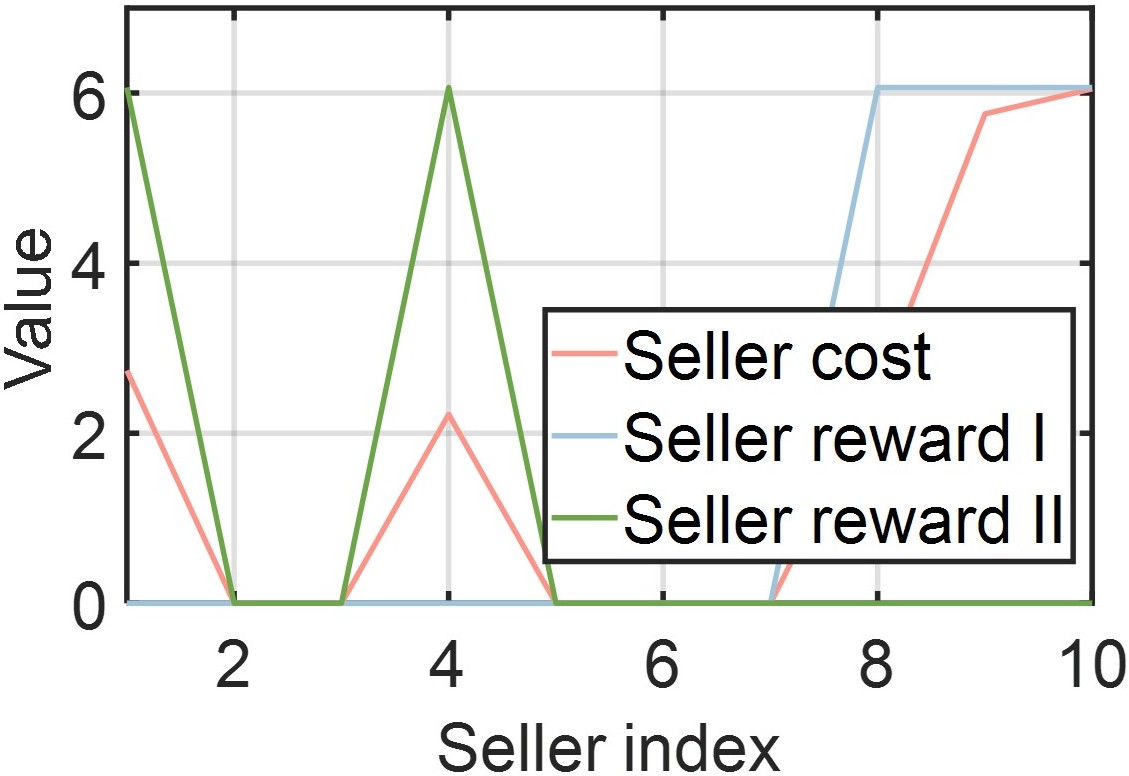}}
\subfigure[Buyers]{\includegraphics[trim=0cm 0cm 0cm 0cm, clip,width=0.45\textwidth,height=0.39\textwidth]{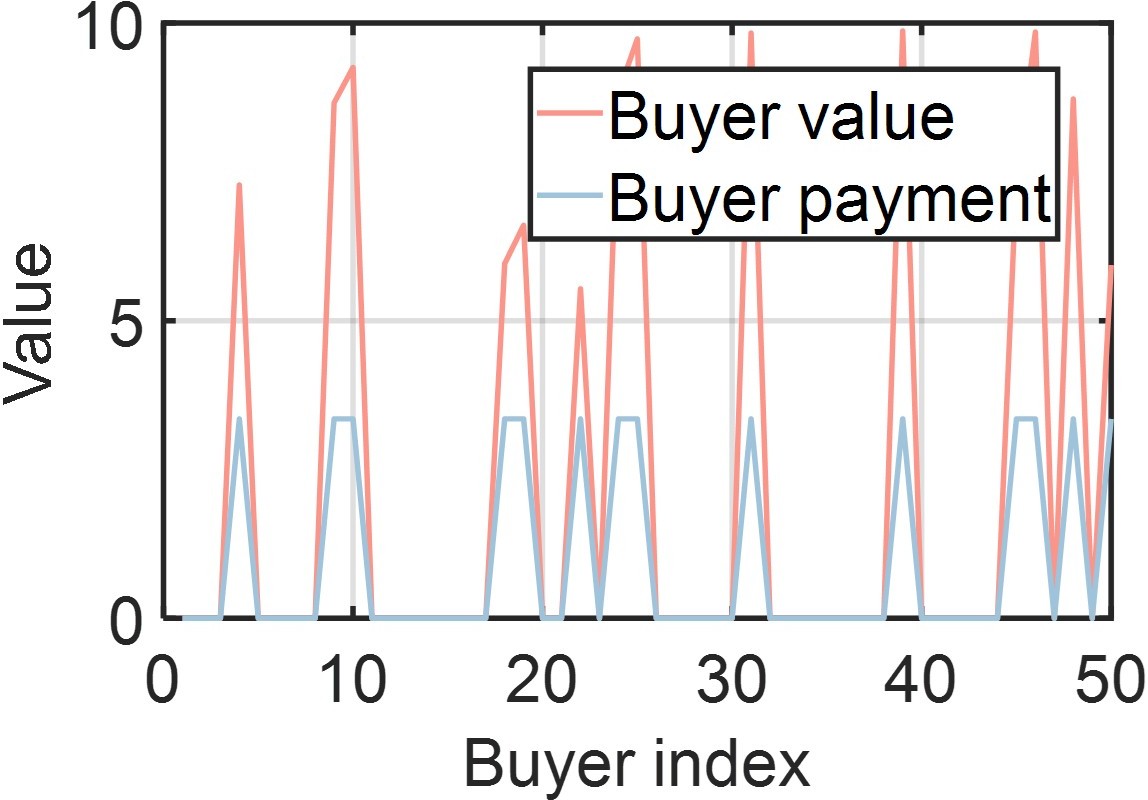}}
		\caption{Individual rationality.}
		\label{fig:main3}
	\end{minipage}
\end{figure*}

{\color{black}{Fig.~4(a) compares the average social welfare (SW) under different numbers of buyers/sellers across 300 experiments. CRDAuction achieves the highest SW, as its decisions rely solely on current network/market conditions, reflecting the SW optimum. TwoSAuction closely matches CRDAuction, though it occasionally underperforms slightly compared with VRAuction, which uses greedy heuristics to maximize buyer profits. However, VRAuction lacks bargaining, compromising fundamental auction properties such as truthfulness and individual rationality, making it impractical for real-world markets. TwoSAuction significantly outperforms SSPDAuction, CRAuction, and RSAuction in average SW. SSPDAuction lacks a backup mechanism, while CRAuction and RSAuction prioritize cost or resource supply, limiting buyer gains. Fig.~4(b) reports the TimeADM per transaction (MATLAB running time) across 300 experiments, using a logarithmic scale to highlight differences. SSPDAuction is excluded since its pre-signed contracts yield TimeADM = 0. CRDAuction incurs much higher TimeADM due to real-time decision-making, especially in large-scale markets. TwoSAuction reduces TimeADM by 76.8\% with 150 buyers and 25 sellers, thanks to OPDAuction, which minimizes bargaining during actual transactions. This advantage grows with more participants. While VRAuction, CRAuction, and RSAuction achieve lower TimeADM by avoiding bargaining, they fail to preserve essential auction properties.}}

As truthfulness and individual rationality are two \textcolor{black}{significant} auction attributes, we provide detailed analysis for both buyers and sellers across various problem sizes in Fig.~5, thereby verifying the theoretical proofs. Specifically, Fig.~5(a) contains four subplots for different problem sizes (200/100 for (a1), 100/15 for (a2), 50/25 for (a3), 50/20 for (a4); and 200/10 for (b1), 50/25 for (b2), 200/15 for (b3), 200/20 for (b4)), where a buyer is randomly selected as an example without loss of generality. The buyer’s truthful bid is marked by a red dot, and the critical transaction price (index $k_b^\ast$) is shown as a blue dot. Clearly, bids in \textcolor{black}{TwoSAuction} reflect true valuations: misreporting does not improve utilities, as low bids may fail to meet the minimum transaction price, while high bids risk transaction failures. In Fig.~5(b), sellers’ critical transaction prices (index $k_s^\ast$) are shown as blue dots, while actual asking prices are red. Unlike conventional auctions, Fig.~5(b) reveals interesting patterns. Sellers near the top of the list $\mathcal{L}_s$ have greater opportunities to contract with buyers. Since \textcolor{black}{TwoSAuction} has two stages, a seller participating in both may set different prices. In phase~3 of Algorithm~1, sellers are matched based on ask rankings, with lower asks prioritized for unmatched buyers. Thus, a seller’s position in $\mathcal{L}_s$ strongly affects its final utility. Except when the ask exceeds the critical price causing failure, seller utility shows only minor fluctuations as the ask varies. Lowering the ask to improve ranking rarely changes final contracts, which depend on resource-demand compatibility (i.e., a knapsack problem), and Algorithm~3 prevents risky mismatches. Thus, misreporting has limited effect, preserving seller truthfulness. This highlights the impact of our pre-auction design and risk management, distinguishing our work from prior studies.
\vspace{-0.1cm}
We next evaluate individual rationality in Fig.~6 with 10 sellers and 50 buyers. Fig.~6(a) compares sellers’ unit costs with their received rewards in Stage~I (Seller reward I'') and Stage~II (Seller reward II''); values remain 0 if a seller fails in either stage. Clearly, winners always receive rewards covering their costs, confirming individual rationality for sellers. Fig.~6(b) shows that buyers’ true valuations always cover their payments, verifying individual rationality for buyers. Overall, \textcolor{black}{TwoSAuction} achieves superior social welfare and time efficiency compared to representative benchmarks while maintaining key auction properties, providing a useful reference for resource provisioning in dynamic, uncertain networks.

\subsection{Performance Evaluation vs. Overbooking and Risk Control}
Since this paper stands for a first attempt to incorporate overbooking and risk management into auction design, we conduct an innovative view to explore their potential benefits. Specifically, we involve the following methods as benchmarks, where the former two aim to highlight the advantages of overbooking, while the later three show the importance of controlling  diverse risks.

\noindent
$\bullet$ \textit{TwoSAuction\_NoOB}: This auction applies our TwoSAuction without overbooking in Stage I.

\noindent
$\bullet$ \textit{SSPDAuction\_NoOB}: This method only considers a pre-double auction procedure without overbooking.

\color{black}
\noindent
$\bullet$ \textit{TwoSAuction\_noBRisk}: Remove risk constraint (C1).

\noindent
$\bullet$ \textit{TwoSAuction\_noVRisk}: Remove risk constraint (C2).

\noindent
$\bullet$ \textit{TwoSAuction\_noSRisk}: Remove risk constraint (C3).
\color{black}

\begin{figure*}[htbp]
	\centering
	\subfigure[] {\includegraphics[trim=0cm 0cm 0cm 0cm, clip, width=.16\textwidth]{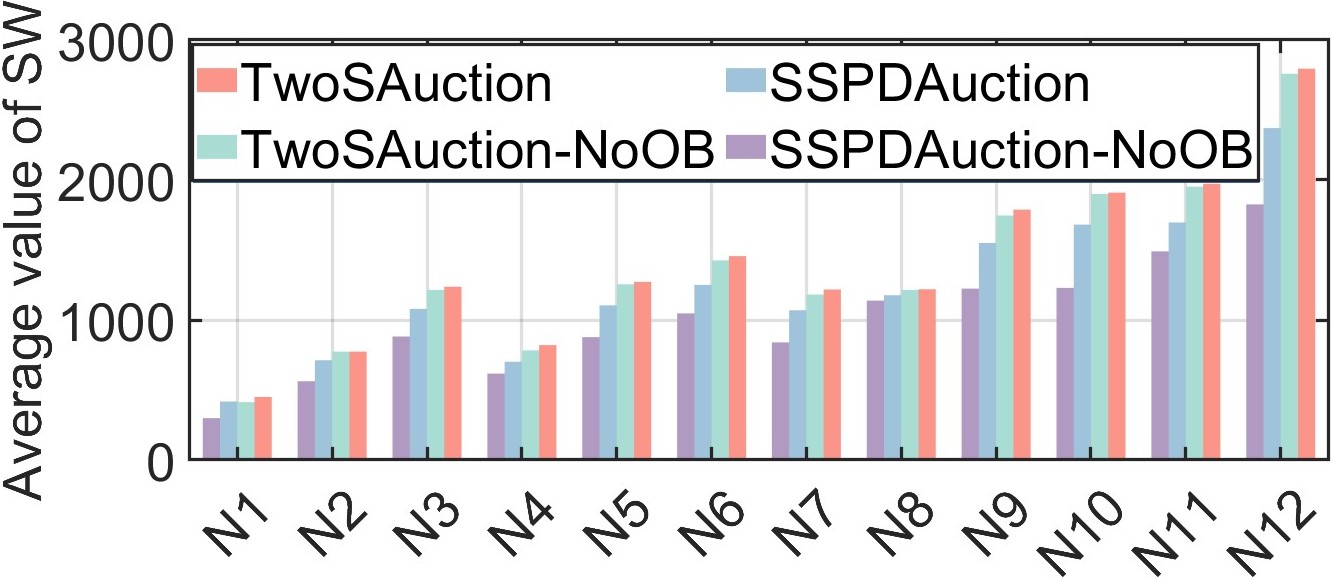}}
	\subfigure[] {\includegraphics[trim=0cm 0cm 0cm 0cm, clip, width=.16\textwidth]{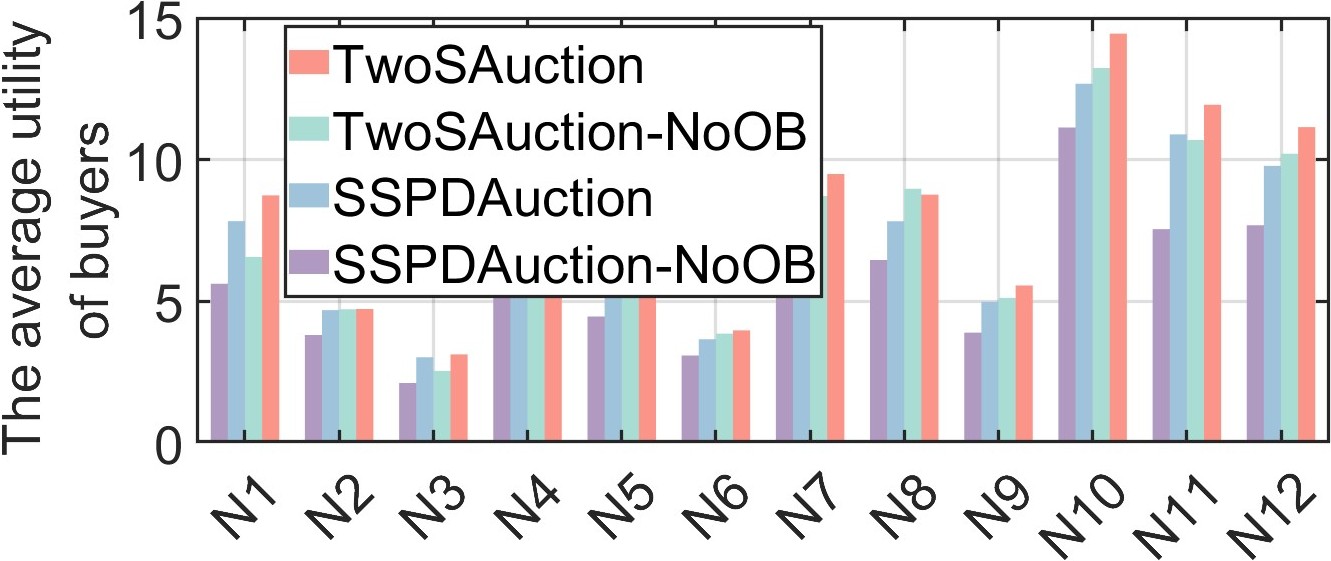}}
	\subfigure[] {\includegraphics[trim=0cm 0cm 0cm 0cm, clip, width=.16\textwidth]{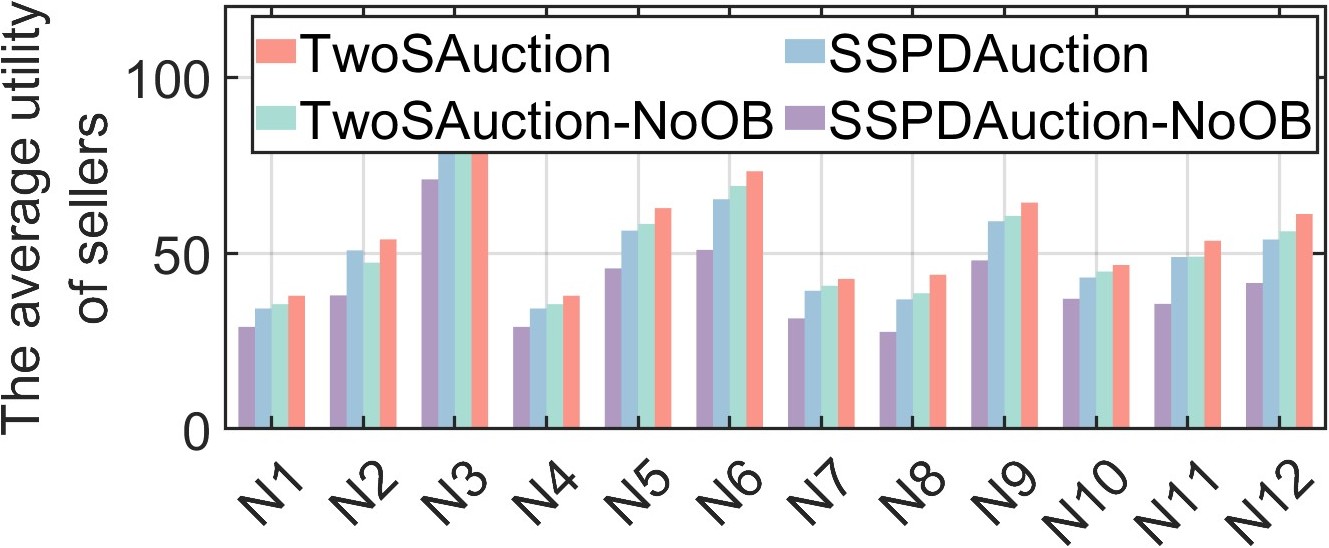}}
	\subfigure[] {\includegraphics[trim=0cm 0cm 0cm 0cm, clip, width=.16\textwidth,height=0.069\textwidth]{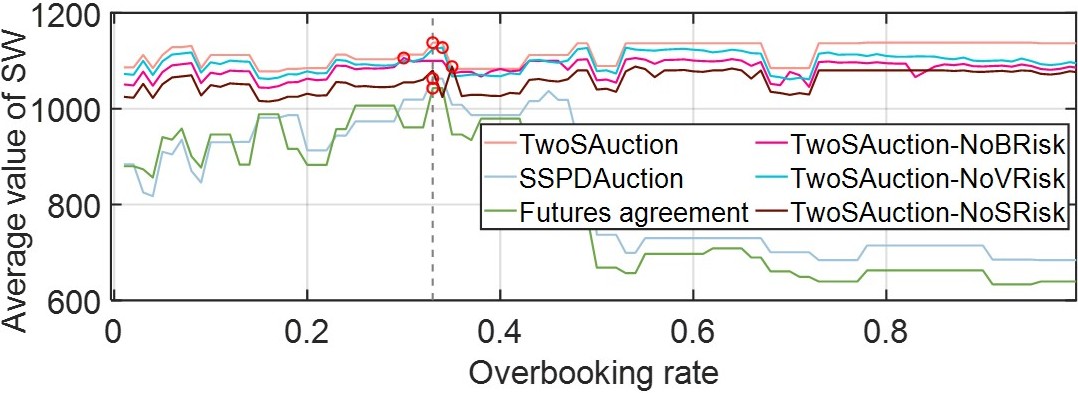}}
	\subfigure[] {\includegraphics[trim=0cm 0cm 0cm 0cm, clip, width=.16\textwidth,height=0.069\textwidth]{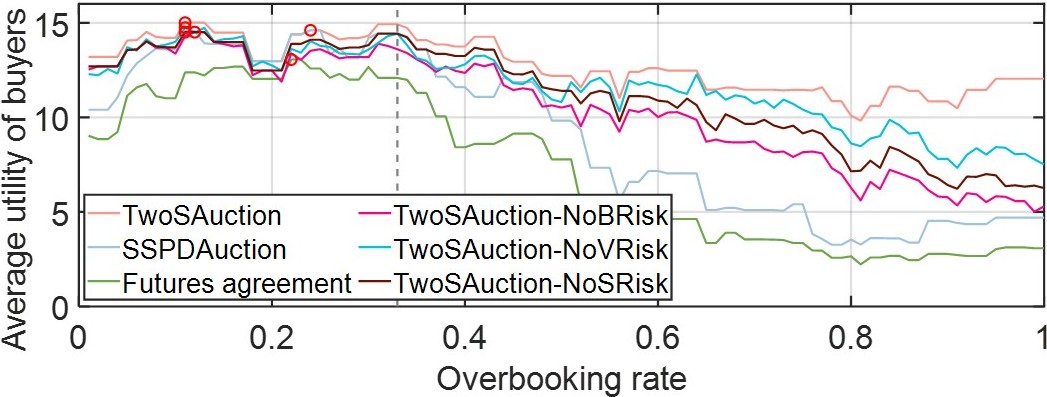}}
	\subfigure[] {\includegraphics[trim=0cm 0cm 0cm 0cm, clip, width=.16\textwidth,height=0.069\textwidth]{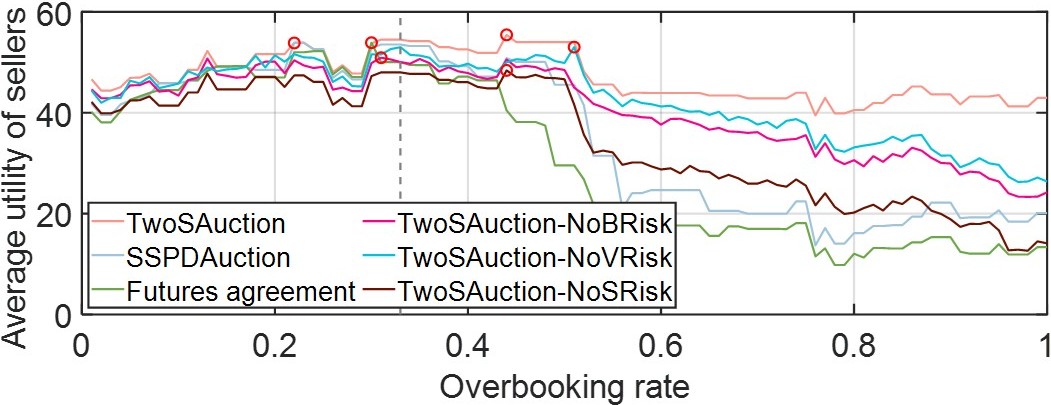}}
	\caption{\textcolor{black}{Effects of overbooking policies and user scales on social welfare, buyer utility, and seller utility, where we have 100/10 (N1), 150/10 (N2), 200/10 (N3), 100/15 (N4), 150/15 (N5), 200/15 (N6), 100/20 (N7), 150/20 (N8), 200/20 (N9), 100/25 (N10), 150/25 (N11), 200/25 (N12).}}
	\label{fig_E1}
\end{figure*}

\color{black}
To comprehensively evaluate the impact of overbooking and risk control on our proposed mechanism, we conduct a series of comparative experiments presented in Fig. 7. \textit{First, Figs. 7(a)--7(c) focus on demonstrating the fundamental advantage of incorporating an overbooking policy under varying problem scales}. As shown in Fig. 7(a), methods that leverage the RBDAuction as a backup plan consistently achieve superior social welfare (SW) compared to others, owing to its robust design as an effective fallback option. Crucially, overbooking enhances the performance of both our proposed TwoSAuction and SSPDAuction, demonstrating clear advantages over methods that do not incorporate overbooking (e.g., TwoSAuction\_NoOB). This trend holds true across different numbers of buyers and sellers, validating the generalizability of our approach. Similarly, Figs. 7(b) and 7(c) illustrate that overbooking brings tangible benefits to both buyers and sellers, respectively, by enabling more efficient resource utilization and mitigating the negative impacts of demand-supply mismatches. \textit{Then, Figs. 7(d)--7(f) delve deeper into the nuanced interplay between the overbooking rate and our integrated risk control mechanisms}. This set of experiments examines the performance under varying overbooking rates (from 0\% to 100\% in increments of 1\%) while keeping other variables constant. We present the following key insights:

\noindent $\bullet$ The critical role of the overbooking rate: As depicted in all three subfigures, there exists an optimal overbooking rate (approximately 33\%, marked by the dashed vertical line) that maximizes social welfare and participant utilities. The red circles highlight the peak values for each curve, showing that this optimal point yields the highest SW for our TwoSAuction (Fig. 7(d)) and near-optimal utilities for both buyers (Fig. 7(e)) and sellers (Fig. 7(f)). Notably, while the 33\% rate may not yield the absolute maximum utility for individual parties (e.g., sellers in Fig. 7(f)), it strikes a balance that maximizes overall system efficiency. Furthermore, the stability of our TwoSAuction's performance (e.g., relatively flat SW curve from 0.8 to 0.9 in Fig. 7(d)) underscores the effectiveness of our RBDAuction in utilizing available resources, ensuring favorable outcomes even at high overbooking rates.

\noindent $\bullet$ The indispensable function of risk control: 
\textit{(i)} In Fig. 7(d) regarding SW, removing any risk control leads to a significant degradation in performance, especially at higher overbooking rates (>50\%). For instance, {TwoSAuction-NoSRisk suffers a sharp drop in SW, indicating that without seller risk control, the system becomes unstable and inefficient.
\textit{(ii)} In Fig. 7(e) regarding buyers' utility, the removal of buyer risk control (TwoSAuction-NoBRisk) causes a dramatic plunge in buyer utility at high overbooking rates. This validates our theoretical analysis: without BRisk, buyers are incentivized to over-participate, only to face a high probability of failure and negative utility when their bids cannot be fulfilled.
\textit{(iii)} In Fig. 7(f) regarding sellers' utility, the absence of seller risk control (TwoSAuction-NoSRisk) results in a catastrophic collapse of seller utility at high overbooking rates. This occurs because sellers, unbound by SRisk, over-commit resources beyond their capacity, leading to massive default penalties and reputational damage, ultimately destroying their own profits.
\textit{(iv)} The TwoSAuction-NoVRisk curve, while showing less severe degradation, still exhibits increased volatility at high overbooking rates, confirming that VRisk plays a crucial auxiliary role in maintaining the stability of Stage II.

In summary, these experiments demonstrate that \textcolor{black}{TwoSAuction}, with its calibrated overbooking rate and integrated risk controls, provides a stable and efficient framework. Ablation studies confirm that removing any single risk control component compromises system stability and performance, validating our core design philosophy.
\color{black}

\begin{figure}[t!]
	\centering
	\subfigure[] {\includegraphics[trim=0cm 0.7cm 0cm 0cm, clip, width=.241\textwidth]{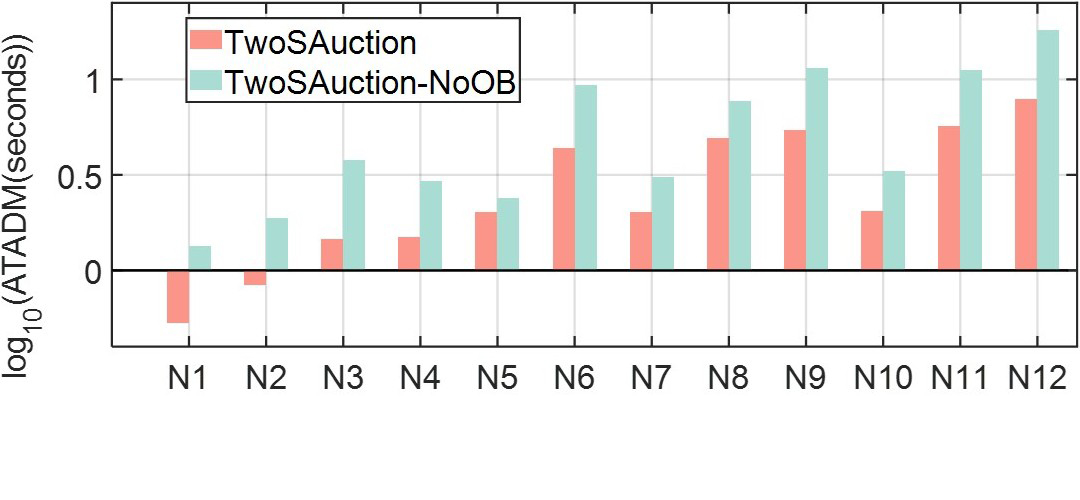}}
	\subfigure[] {\includegraphics[trim=0cm 0cm 0cm 0cm, clip, width=.241\textwidth,height=0.099\textwidth]{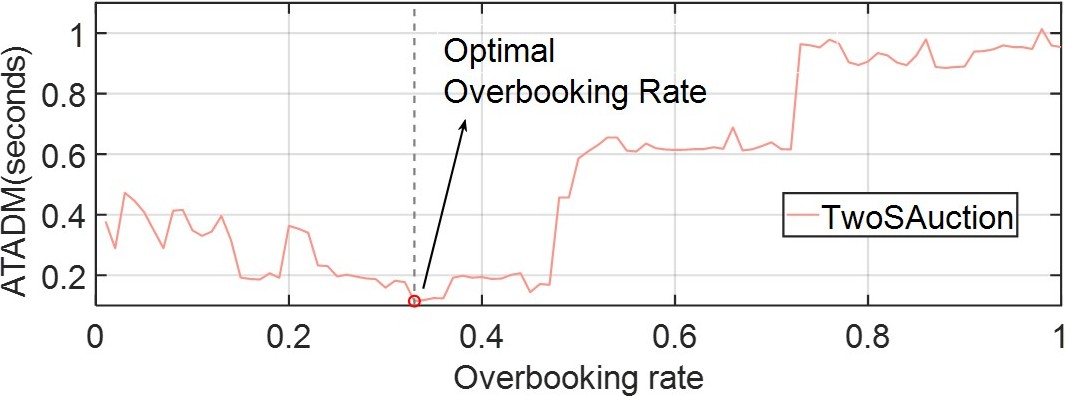}}
	\caption{(a) Analysis of time efficiency on a fixed overbooking rate 33\%. (b) Analysis of time efficiency upon having different overbooking rates.}
	\label{fig_E1}
\end{figure}


Fig. 8 illustrates the time efficiency impacted by overbooking, as reflected by average TimeADM. Given an overbooking rate (i.e., 33\%), the time consumed by our TwoSAuction is significantly lower than that of TwoSAuction-NoOB. This improvement is attributed to the increased number of pre-signed contracts enabled by overbooking, which can be directly fulfilled during each transaction, thereby enhancing time efficiency. Fig. 8(b) presents the curve of TimeADM across various overbooking rates, with the red circle on the y-axis marking a rate of 33\%. It is clear that a proper overbooking rate can significantly accelerate the auction process, promoting a time-efficient trading market.

\section{Conclusion and Future Work}
\noindent
This paper presents a novel two-stage double auction for resource scheduling in dynamic edge networks. By incorporating resource overbooking, our approach promotes long-term and real-time cooperation between edge servers and mobile devices. The OPDAuction enables long-term contracts for efficient resource provisioning under uncertainty, while the RBDAuction dynamically addresses residual demand to optimize allocation and social welfare. Experiments demonstrate that our mechanism outperforms conventional auctions while preserving key design properties. Future research may focus on refining the overbooking strategy to further mitigate the inherent risks of pre-auctions, enhancing overall robustness and reliability. Integrating machine learning techniques to predict and dynamically adapt to complex, fluctuating resource demand could improve the responsiveness and intelligence of resource scheduling. Additionally, extending the framework to multi-market and multi-resource scenarios would provide insights into its scalability, supporting effective deployment in real-world edge computing environments.


\newpage
\clearpage
\appendices
\section{Details of Algorithms}

\subsection{Member determination (Algorithm 1)}

Member determination represents a key feature of our OPDAuction, which is abbreviated as OPDAuction-MemberD for notational simplicity, shown in Algorithm 1. It aims to achieve feasible mappings between buyers and sellers, thereby supporting  the signing of long-term contracts between them.

\textcolor{black}{To bridge the gap between resource supply and demand, we first reorder buyers and sellers based on their bids and asks}\footnote{In the context of a double auction, the term ``ask'' refers to the price requested by a seller for their offered services, as commonly used in existing literature, which signifies the payment amount a seller seeks to receive. To streamline the expressions, we adopt the term ``ask'' as a noun to denote the asked price or payment requested by resource sellers, a convention also supported by existing works \cite{b25}.}\textcolor{black}{, respectively.} Given that each buyer can offer different bids for various sellers, we organize buyers in list $\mathcal{L}_b$ by following a non-ascending order of the average value of their bids (line 5), calculated by (18).
\begin{equation}\label{key}
	\begin{aligned}
		&bid^\mathsf{Avg}_{n'}=\frac{\sum_{s_{m'}\in\bm{\mathcal{S}}}{bid}_{m',n'}}{\left|\bm{\mathcal{S}}\right|},
	\end{aligned}
\end{equation}
while the list $\mathcal{L}_b$ of buyers can be expressed as 
\begin{equation}\label{key}
	\begin{aligned}
		\mathcal{L}_b\!=\!\left\{\!\left(\!k_b,b_{n'},bid^\mathsf{Avg}_{n'}\right)\middle| b_{n'}\!\in\!\bm{\mathcal{B}},\! \text{ non-ascending \! on \! } bid^\mathsf{Avg}_{n'}\!\right\}\!\!.
	\end{aligned}
\end{equation}

To capture the position of different buyers in list $\mathcal{L}_b$, we utilize a notation $k_b$ to indicate the index of each specific buyer in $ \mathcal{L}_b $. Also, let $ \mathcal{L}_b(k_b) $ refer to the buyer corresponding to index $ k_b $ in list $\mathcal{L}_b$, and  $ bid^{\mathsf{List}}_{k_b} $ denote the average value of bid of the buyer $ \mathcal{L}_b(k_b) $. Note that in $\mathcal{L}_b $, we have
\begin{equation}\label{key}
	\begin{aligned}
		&bid^{\mathsf{List}}_{1}\succ bid^{\mathsf{List}}_{2}\succ\cdots\succ bid^{\mathsf{List}}_{\left|\bm{\mathcal{B}}\right|}.
	\end{aligned}
\end{equation}

Similarly, sellers are sorted in a non-descending order according to their asked prices (line 6), for which we introduce list $ \mathcal{L}_s $ as (21).
\begin{equation}\label{key}
	\begin{aligned}
		&\mathcal{L}_s=\left\{(k_s,s_{m'},{ask}_{m'})|s_{m'}\in\bm{\mathcal{S}},\text{non-decending\ on\ } {ask}_{m'}\right\}.
	\end{aligned}
\end{equation}

Here, $ k_s $ represents the new index of sellers in list $ \mathcal{L}_s $. For better analysis, let $ \mathcal{L}_s(k_s) $ denote the seller corresponding to index $ k_s $ in list $ \mathcal{L}_s $, and $ ask^{\mathsf{List}}_{k_s} $ be the asked price of seller $ \mathcal{L}_s(k_s) $. Note that in $  \mathcal{L}_s $, we have: 
\begin{equation}\label{key}
	\begin{aligned}
		&ask^{\mathsf{List}}_{1}\prec ask^{\mathsf{List}}_{2}\prec \cdots\prec ask^{\mathsf{List}}_{\left|\bm{\mathcal{S}}\right|}.
	\end{aligned}
\end{equation}
where symbols ``$ \succ $'' and ``$ \prec $'' are used to indicate the order of elements within lists $ \mathcal{L}_b $ and $ \mathcal{L}_s $. For instance, ``$\succ$'' denotes that the element on left-hand side is positioned before that on its right-hand side in list $\mathcal{L}_b$, while ``$\prec$'' indicates a similar relationship in list $\mathcal{L}_s$.
To determine winning sellers and their members, the next crucial step is to find the key buyer and the key seller in given lists $ \mathcal{L}_b $ and $ \mathcal{L}_s $ (lines 7-18). To distinguish, we denote their corresponding indices as $ k_b^\ast $ (for the key buyer) and $ k_s^\ast $ (for the key seller), respectively. To maintain the property of budget balance, the values of $ k_b^\ast $ and $ k_s^\ast $ should satisfy the following two conditions:
\begin{equation}\label{key}
	\begin{aligned}
		&bid^{\mathsf{List}}_{k_b^\ast+1}\geq ask^{\mathsf{List}}_{k_s^\ast+1} \text{ and }bid^{\mathsf{List}}_{k_b^\ast+2}<ask^{\mathsf{List}}_{k_s^\ast+2},
	\end{aligned}
\end{equation}
\begin{equation}\label{key}
	\begin{aligned}
		&bid^{\mathsf{List}}_{k_b^\ast+1}\geq ask^{\mathsf{List}}_{k_s^\ast+1}
		\\&\text{and }(k_b^\ast+1=\left|\bm{\mathcal{B}}\right|\ \text{or} \ k_s^\ast+1=\left|\bm{\mathcal{S}}\right|),
	\end{aligned}
\end{equation}
where (23) ensures that bids $bid^{\mathsf{List}}_{1}, bid^{\mathsf{List}}_{2}, \ldots, bid^{\mathsf{List}}_{k_b^\ast+1}$ are sufficient to meet or exceed asks $ask^{\mathsf{List}}_{1}, ask^{\mathsf{List}}_{2}, \ldots, ask^{\mathsf{List}}_{k_s^\ast+1}$, while $ask^{\mathsf{List}}_{k_s^\ast+2}$ exceeds $bid^{\mathsf{List}}_{k_b^\ast+2}$. This implies that the revenue of the winning seller will meet or exceed their asking prices. When $ \mathcal{L}_b\left(k_b^\ast+1\right)$ refers to the last buyer in $ \mathcal{L}_b$, or $ \mathcal{L}_s\left(k_s^\ast+1\right)$ is the last seller in $ \mathcal{L}_s$, this indicates a continuous pricing structure without any pivot points. Consequently, we define the index of the last buyer/seller in these lists as the pivotal index, detailed by condition (24).

Recall that $\bm{\mathcal{F}_1}$ aims to maximize the expected value of social welfare, which involves aggregating the difference between profits gained by members and the costs incurred by successful sellers. An equivalent alternative consideration is to maximize the number of successful seller-buyer pairs, i.e., determining as many members for sellers as possible while keeping risks within acceptable limits, with the aim of effectively increasing the expected social welfare. To achieve such a goal, our proposed OPDAuction-MemberD given in Algorithm 1 is structured around three key phases. First, we start with generating lists $ \mathcal{L}_b $ and $\mathcal{L}_s$ (Phase 1, lines 4-5). According to the given lists, in Phase 2, we consider buyers from indexes $ k_b=(\left|\bm{\mathcal{B}}\right|-1)$ to $ k_b=1$ (line 8). Suppose that $n'$ and $m'$ denote two constant indices, if $ \mathcal{L}_b\left(n'\right) $ and $ \mathcal{L}_s\left(m'\right) $ satisfy conditions (19) and (20), OPDAuction-MemberD assesses the resource capacity of the first $ m' $ sellers and determines the maximum number of buyers that can be accommodated by these sellers (lines 12-18). Notably, the resource capacity does not necessarily have to match the number of winning buyers, as this step may not guarantee the property of truthfulness. When OPDAuction-MemberD finds better values of $ k_b^\ast $ and $ k_s^\ast $ that enable a larger number of winning buyer-seller pairs, it updates $ MaxNumB $, $ k_b^\ast $, and $ k_s^\ast $ accordingly, where $ MaxNumB $ denotes the number of buyers that selected sellers can serve (lines 15-17).\vfill

When coming to buyer-seller matching (Phase 3, lines 18-32), our proposed OPDAuction-MemberD considers only the top $ k_b^\ast $ buyers in $ \mathcal{L}_b $ and the top $ k_s^\ast $ sellers in $ \mathcal{L}_s $ (lines 21-23). Such a consideration can support the property of budget-balance, ensuring that the auctioneer will not incur a deficit. Then, we focus on selecting the optimal group of members for each seller to maximize the expectation of social welfare. This can be formulated as a 0-1 knapsack problem and effectively solved by using dynamic programming \cite{b26} (lines 24-27). Specifically, we start with creating an empty matrix $ U $ to record the difference between buyers' bids and sellers' asks. Then, according to $ \mathcal{L}_b $ and $ \mathcal{L}_s $, we methodically compile the essential elements required for a 0-1 knapsack problem (01KP), including: \textit{(i)} for each seller, we denote $ h_m $ as the amount of available RBs after overbooking, serving as the capacity of the knapsack (line 20); \textit{(ii)} from matrix $ U $, a set of buyers can be selected, with the optimal selection among them determined by solving 01KP. This is aimed at maximizing the profit for seller $ \mathcal{L}_s\left(m'\right) $; \textit{(iii)} we use the respective number of RBs required by each buyer as the weight array $ v $ in the 01KP (line 25); and \textit{iv)} the bids from buyers for $ \mathcal{L}_s\left(m'\right) $ are considered as the values $ w$ of items in 01KP, which is then resolved through dynamic programming (line 26).

For buyers who have not been selected in this round for $ \mathcal{L}_s\left(m'\right) $, their bids to sellers in $ \mathcal{L}_s\left(m'\right) $ are set to zero. This adjustment signifies that their maximum bid in their respective bid lists is disregarded, and the second-highest value in the original bid list is promoted to become the new maximum one. Such an operation ensures that these non-selected buyers are still able to participate in auctions conducted by other sellers, maintaining their engagement in the auction process without being prematurely excluded. Then, to support individual rationality, the selected buyer $ \mathcal{L}_b\left(n'\right) $ compares its bid (${bid}_{m',n'}$) for seller $ \mathcal{L}_s\left(m'\right) $ with the mean bid ($ {Mbid}_{n'} $) of its bid array. If $ {bid}_{m',n'} $ is greater than $ bid^{\mathsf{List}}_{n'} $, the matching between $ b_{n'} $ and $ \mathcal{L}_s\left(m'\right) $ is considered to be successful. Subsequently, the seller deducts the volume of the buyer's task from their available resources. Once the seller's resources are fully allocated, they conclude their participation in the auction. This ensures that transactions are mutually beneficial and adhere to the principles of individual rationality, preventing sellers from accepting bids that are too low, and buyers from paying more than their average bids.

The process described above continues iteratively until either all buyers have successfully reached agreements with sellers or all available resources have been allocated. This enables the auction to adapt to fluctuations in resource supply and demand, thereby maximizing the number of successful matches between buyers and sellers within the constraints of available resources.

\subsection{Long-term contract design (Algorithm 2)}
For subproblem $\bm{\mathcal{F}_{1b}}$, since each buyer $ b_{n'} $ has already identified its interested seller $ \mathcal{L}_s\left(m'\right) $ through Algorithm 1, we develop OPDAuction-ContractD by borrowing the idea of binary search algorithm to finalize the payments made by winning buyers (members), as detailed in Algorithm 2.
\textcolor{black}{This approach streamlines the negotiation process by efficiently determining contract terms, especially payment-related components, being mutually acceptable to both parties.}
At the beginning of Algorithm 2, we replace the value $ bid^{\mathsf{List}}_{n'} $ used in Algorithm 1 for participation in the double auction sorting with ${bid}_{m',n'}$ (lines 5-8). This adjustment ensures that the negotiation for long-term contracts reflects the genuine valuation for $b_{n'}$ of service provided by $s_{m'}$, facilitating a more accurate and fair contract agreement.  Subsequently, during the pricing process for \( b_{n'} \) in Algorithm 2, it consistently uses the actual bids \( bid_{m',n'} \) as the sorting criterion for the double auction. This step ensures that the buyer's payment can be calculated according to its true bid for matching the seller, thereby preparing to guarantee the buyer's truthfulness.

Note that our designed OPDAuction-ContractD only considers the winning buyers (lines 6-15). For each winner with index \( b_{n'} \in \mathcal{L}_b \), the lower bound is set as $ bid^{\mathsf{List}}_{k_b+1} $ while the upper bound is set by $ bid^{\mathsf{List}}_{n} $ (line 8). This is due to a key criterion: the final price paid by the winning buyer should not be lower than the seller's asking price nor higher than the buyer's initial bid. Then, our OPDAuction-ContractD leverages the binary search algorithm \cite{b25} to identify the minimum acceptable price within the range of lower- and upper-bound prices that can make a buyer to win (lines 9-15). This process efficiently narrows down the optimal price point that satisfies both the buyer's and seller's constraints. Afterwards, OPDAuction-ContractD also addresses the pricing process for winning sellers (lines 17-19), which parallels the methodology applied to buyers. Nonetheless, when the pricing determination condition is satisfied, the algorithm sets the lower bound $ low = {ask}_{m'} $. This operation ensures that sellers receive at least their asking price while simultaneously minimizing the costs for buyers. Ultimately, OPDAuction-ContractD concludes by returning the payment vectors for both buyers and sellers (line 20), summarizing the financial transactions to be made following the auction outcomes. Then, we use a similar method to set prices for sellers (lines 21-32), which will not be reiterated here.

\subsection{Overbooking rate optimization (Algorithm 3)}
 Strategic overbooking serves as a critical mechanism to address dynamic resource demand-supply imbalances while maintaining market efficiency. Algorithm 3 systematically identifies the optimal overbooking rate $\lambda^*$ through a three-phase process that integrates binary search with parallel risk evaluation.  
	
	The algorithm initiates by precomputing risk thresholds ($\xi^\mathsf{S}$, $\xi^\mathsf{B}$, $\xi^\mathsf{V}$) and initializing a sorted candidate set $\Lambda = \{1\%,(1+\Delta\lambda)\%,...,100\%\}$ (Phase 1). For each candidate $\lambda$, it calculates adjusted seller capacities $\overline{R_m} = c_m \times (1+\lambda)$ and verifies basic capacity constraints through a binary search paradigm (Phase 2). 
\textcolor{black}{The key innovation is an efficient evaluation step that simultaneously performs three critical operations, while discarding infeasible overbooking rates early to reduce computation}: \textit{(i)} Invoking OPDAuction-MemberD to determine winning buyer-seller pairs under current $\lambda$, \textit{(ii)} Deriving contract terms $[p^\mathsf{b},r^\mathsf{s}]$ via OPDAuction-ContractD, and \textit{(iii)} Conducting parallel risk assessments using seller risk function $\mathcal{R}^\mathsf{SRisk}_m = f(\mathbbm{d}_m, \overline{R_m})$ and buyer risk function $\mathcal{R}^\mathsf{BRisk}_n = g(\mathbbm{a}_n, t_n)$.  
	
	Early termination mechanisms automatically discard $\lambda$ values exceeding predefined risk thresholds, while maintaining records of feasible candidates in matrices $\bm{U^\#}$, $\bm{p}$, $\bm{r}$, and $\bm{X^\#}$. The social welfare metric $USW = \sum_{m,n} (bid_{m,n}-ask_m) \cdot x_{m,n}$ guides the binary search direction, prioritizing candidates with higher welfare values. Phase 3 employs golden-section search for neighborhood refinement when multiple $\lambda$ values yield comparable welfare, ultimately selecting the configuration maximizing both welfare and successful buyer-seller matches. The algorithm outputs optimal contracts $\mathbb{C}_{m,n}^\mathsf{b\leftrightarrow s}$ through coordinated evaluation of risk profiles and economic efficiency, ensuring balanced market participation under dynamic edge network conditions.

\section{Details of RBDAuction}

\subsection{Basic Modeling}
\subsubsection{Utility of buyers (regarding $\tilde{\bm{\mathcal{B}}}$)}
During the second stage, each buyer $ \tilde{b}_j\in\tilde{\bm{\mathcal{B}}} $ can be denoted by a 3-tuple $ \tilde{b}_j=\left\{\tilde{t}_j,\tilde{v}_{i,j},\tilde{bid}_{i,j}\right\} $.
Utility of $\tilde{b}_j$ consists of the following two aspects: \textit{(i)} the amount of resources acquired from $ \tilde{s}_i $, and \textit{(ii)} the unit net profit that $ \tilde{b}_j $ can obtain from enjoying computing service from $ \tilde{s}_i $. We accordingly calculate the utility $ \tilde{U}^\mathsf{B}_{j} $ of each buyer $ \tilde{b}_j\in\tilde{\bm{\mathcal{B}}} $ as:
\begin{equation}\label{key}
	\begin{aligned}
		&\tilde{U}^\mathsf{B}_{j}\left(\tilde{t}_j,\tilde{p}_j^\mathsf{b} \right)=\sum_{\tilde{s}_i\in \tilde{\bm{\mathcal{S}}}}{\tilde{x}}_{i,j}\tilde{t}_j\left(\tilde{v}_{i,j}-\tilde{p}_j^\mathsf{b}\right)
	\end{aligned}.
\end{equation}

Note that in Stage II, the utility of the members who have successfully purchased resources can be directly determined by the pre-signed contract during Stage I. The guests included $\tilde{\bm{\mathcal{B}}}$ can have their utility calculated by (25). Besides, the utility of volunteers consists of two parts: the penalty a volunteer may receive from the contractual seller, and the utility it obtains from participating in RBDAuction after joining $\tilde{\bm{\mathcal{B}}}$ by (25).

\subsubsection{Utility of sellers (regarding $\tilde{\bm{\mathcal{S}}}$)}
Correspondingly, each seller $ \tilde{s}_i\in\tilde{\bm{\mathcal{S}}} $ can be denoted by a 3-tuple $ \tilde{s}_i=\left\{\tilde{c}_i,\tilde{ask}_i,\tilde{R}_{i}\right\} $, where $ \tilde{R}_{i} $ signifies the remaining idle resources of seller $ \tilde{s}_i $.
Then, we can calculate the utility of each seller in additional auction as benefited by its remaining resources as
\begin{equation}\label{key}
	\begin{aligned}
		&\tilde{U}^\mathsf{S}_{i}\left(\tilde{r}_i^\mathsf{s},\tilde{c}_i \right)=\sum_{\tilde{b}_j\in \tilde{\bm{\mathcal{B}}}}{\tilde{x}}_{i,j}\tilde{t}_j\left(\tilde{r}_i^\mathsf{s}-\tilde{c}_i\right)
	\end{aligned}.
\end{equation}

Apparently, the overall utility of a seller depends on two factors: \textit{(i)} the income it receives from members who have successfully enjoyed computing services, and the compensation paid to volunteers or the utility calculated by (26).

Moreover, in Stage II, the auctioneer is also responsible for coordinating the implementation of pre-signed long-term contracts. Meanwhile, for buyers in $\tilde{\bm{\mathcal{B}}}$ and sellers in $\tilde{\bm{\mathcal{S}}}$, it also helps determine winning buyer-seller pairs as well as their prices.
Accordingly, the utility of the auctioneer in Stage II can be defined by the difference between the payments from buyers and the rewards to sellers,
\begin{equation}\label{key}
	\begin{aligned}
		&\tilde{U}^\mathsf{P}\left(\tilde{t}_j,\tilde{p}_j^\mathsf{b},\tilde{r}_i^\mathsf{s} \right)=\sum_{\tilde{b}_j\in \tilde{\bm{\mathcal{B}}}}\sum_{\tilde{s}_i\in \tilde{\bm{\mathcal{S}}}}{\tilde{x}}_{i,j}\tilde{t}_j\left(\tilde{p}_j^\mathsf{b}-\tilde{r}_i^\mathsf{s}\right)
	\end{aligned}.
\end{equation}

And the practical social welfare defined by
\begin{equation}\label{key}
	\begin{aligned}
		&\tilde{U}^\mathsf{SW}=\sum_{\tilde{b}_j\in\tilde{\bm{\mathcal{B}}}} \tilde{U}^\mathsf{B}_{j}+\sum_{\tilde{s}_i\in\tilde{\bm{\mathcal{S}}}} \tilde{U}^\mathsf{S}_{i}+\tilde{U}^\mathsf{P}
		\\&=\sum_{\tilde{b}_j\in \tilde{\bm{\mathcal{B}}}}\sum_{\tilde{s}_i\in \tilde{\bm{\mathcal{S}}}}{\tilde{x}}_{i,j}\tilde{t}_j\left(\tilde{v}_{i,j}-\tilde{c}_i\right)
	\end{aligned}.
\end{equation}

\subsection{Solution design}
\begin{algorithm}[]
	{\small  
		\caption{\small{Realtime Backup Double Auction}}
		\LinesNumbered 
		
		{\bf{Input :}} 
		$ \tilde{\bm{\mathcal{B}}},\tilde{t}_j,\tilde{v}_{i,j},\tilde{bid}_{i,j},\tilde{\mathbbm{a}}_j,\tilde{\bm{\mathcal{S}}},\tilde{c}_i,\tilde{ask}_i,\tilde{\mathbbm{d}}_{i},\tilde{\mathbbm{r}}_{i}, \bm{X}^*, \mathbb{B}, \mathbb{R}$
		
		{\bf{Output :}} 
		$ \tilde{\bm{X}},\tilde{\bm{p}}^\mathsf{\bm{b}},\tilde{\bm{r}}^\mathsf{\bm{s}} $
		
		{\bf{Initialization :}} 
		$ \tilde{\bm{X}}\leftarrow \emptyset,\bm{p_1}\leftarrow \emptyset,\bm{r_1}\leftarrow \emptyset,\bm{volunteer}\leftarrow \emptyset $
		
		\textbf{\# Phase 1: Extracting users complying with the agreements}\\ 
		
		\For{
			each $ i =\left|\bm{\mathcal{S}}\right|, ..., 1 $
		}{
			\For{
				each $ j = \left|\bm{\mathcal{B}}\right|, ..., 1 $
			}{
				\If{$ x_{i,j}=1$ and $\mathbb{B}(j)=1 $
				}{
					$ v\ \leftarrow \ \left[v,t_j\right],w\leftarrow \left[w,\tilde{bid}_{i,j}\right],R1\ \leftarrow \ R1+t_j $
				}
				\If{$R1 > \mathbb{R}(i)$
				}
				{$ c \leftarrow \mathbb{R}\left(i\right) $
					
					$ \bm{t}\leftarrow KP\left(c,v,w\right) $
				}
				\If{$ \bm{t}\left(j\right)=1 $
				}{
					$ x_{i,j}^\mathsf{**}\leftarrow1,\bm{p_1}\leftarrow \left[\bm{p_1},p_j^\mathsf{b}\right],\bm{r_1}\leftarrow [\bm{r_1},r_j^\mathsf{s}] $
				}
				\Else{$ \bm{volunteer}\ \leftarrow \ [\bm{volunteer},j] $
					
					$ x_{i,j}^\mathsf{**}\leftarrow 0 $
				}	
			}
		}
		\textbf{\# Phase 2: Updating users participating in stage II auction}\\ 
		\For{
			each $ i =\left|\bm{\mathcal{S}}\right|, ..., 1 $
		}{
			\For{
				each $ j = \left|\bm{\mathcal{B}}\right|, ..., 1 $
			}{
				\If{$ x_{i,j}^\mathsf{**}=0$ and $\mathbb{B}(j)=1 $
				}{
					$ \tilde{\bm{\mathcal{B}}}\leftarrow[\tilde{\bm{\mathcal{B}}},b_j] $
				}
				\If{$ x_{i,j}^\mathsf{**}=1 $
				}{
					$ \mathbb{R}\left(i\right)\leftarrow\ \mathbb{R}\left(i\right)-t_j $}	
			}
		}
		\textbf{\# Phase 3: Stage II Auction}\\ 
		$\tilde{\bm{X}}\leftarrow$\textbf{Algorithm 1}$\left(\tilde{\bm{\mathcal{B}}}, \tilde{\bm{\mathcal{S}}}\right)$\\ 
		$\tilde{\bm{p}}^\mathsf{\bm{b}},\tilde{\bm{r}}^\mathsf{\bm{s}}\leftarrow$\textbf{Algorithm 2}$\left(\tilde{\bm{\mathcal{B}}}, \tilde{\bm{\mathcal{S}}}\right)$\\ 
		\Return{$ \tilde{\bm{X}},\tilde{\bm{p}}^\mathsf{\bm{b}},\tilde{\bm{r}}^\mathsf{\bm{s}} $.}
		
	}				
\end{algorithm}

Details of our proposed RBDAuction are given by Algorithm 4, which first identifies the buyer-seller pairs complying with the long-term agreements based on the actual attendance of buyers ($ \mathbb{B} $) and the actual available idle resources of sellers ($ \mathbb{R} $). A dynamic programming knapsack problem is utilized to determine the members that can practically obtain services for each seller, as specified in pre-signed long-term contracts, while the others remain as volunteers due to limited resource supply (lines 5-16, Algorithm 4). Next, volunteers and guests can be engaged in the current auction, described by $\bm{\mathcal{B}}^\prime$ and $\bm{\mathcal{S}}^\prime$, as shown by lines 18-23, Algorithm 4. Then, we design algorithms similar to Algorithms 1-2 to select winning seller-buyer pairs, thus yielding a final solution for $ \tilde{\bm{X}} $.

\section{Property Analysis}
\color{black}
\subsection{Individual rationality}
 \begin{thm} All the buyers and sellers in our proposed two-stage double auction are individual rational. \end{thm}
\begin{proof}
			The individual rationality of our mechanism is guaranteed by its fundamental design principles in both stages.

\noindent
\textbf{\textit{(i)} Stage I (OPDAuction):}
The core of individual rationality in Stage I is embedded in algorithms for member determination (Algorithm 1) and contract design (Algorithm 2).

$\bullet$ \textit{Buyer Side:} In Algorithm 2 (lines 9-15), the payment for a winning buyer $b_{n'}$ is determined via a binary search between its true bid $bid_{m',n'}$ (which replaces $bid^{\mathsf{List}}_{n'}$ for contract negotiation) and the critical price $bid^{\mathsf{List}}_{k_b^*+1}$. This ensures that the final payment $p_{n'}^\mathsf{b}$ satisfies $p_{n'}^\mathsf{b} \leq bid_{m',n'}$. Therefore, the buyer's utility, as defined in Eq. (2), remains non-negative for the transaction itself (before considering the probabilistic risk of becoming a volunteer).

$\bullet$ \textit{Seller Side:} Similarly, Algorithm 2 (lines 22-31) sets the reward for a winning seller $s_{m'}$ via a binary search between its true ask $ask_{m'}$ and the critical price $ask^{\mathsf{List}}_{k_s^*+1}$. This ensures $r_{m'}^\mathsf{s} \geq ask_{m'} \geq c_m$, guaranteeing that the seller's reward covers its cost, leading to a non-negative utility as per Eq. (7).

$\bullet$ \textit{Risk Management:} Crucially, our mechanism acknowledges that the pre-auction nature introduces risk (e.g., a buyer becoming a volunteer or a seller facing overbooking). To ensure individual rationality holds in expectation, we incorporate explicit risk constraints ((C1)-(C3)) in the optimization problem $\bm{\mathcal{F}_1}$. These constraints bound the probability of a participant receiving an unsatisfactory utility. By limiting these risks to acceptable thresholds ($\xi^M, \xi^V, \xi^S$), we ensure that the \textit{expected utility} for all participants, as calculated in Eqs. (2) and (7), remains non-negative. This is a key enhancement over a naive pre-auction, where participants might be forced into contracts with high risk of negative outcomes.

\noindent
\textbf{\textit{(ii)} Stage II (RBDAuction):}
The real-time backup auction inherits and reinforces individual rationality from Stage I.

$\bullet$ The RBDAuction uses algorithms analogous to Algorithms 1 and 2. Therefore, the same logic applies: winning buyers pay no more than their bids (Eq. 25), and winning sellers receive no less than their asks (Eq. 26), which are set to cover their costs.

$\bullet$ For volunteers, their utility is protected by the compensation $q_{m,n}^\mathsf{s\to b}$ received from their contractual seller, as defined in Eq. (1). This compensation, coupled with the opportunity to compete in the RBDAuction, ensures that their overall utility is managed and kept within acceptable bounds, as enforced by the VRisk constraint (C2) in Stage I.

In summary, individual rationality is not an afterthought but a foundational property built into the core algorithms of both stages. The price determination mechanisms (Algorithms 1 and 2) directly enforce it for the immediate transaction, while the risk management framework (constraints (C1)-(C3)) ensures it holds in expectation over the uncertain outcomes of the pre-auction. This comprehensive approach guarantees that participation in our two-stage auction is a rational choice for all market participants .
\end{proof}
\subsection{Truthfulness}
Recall previous discussions, different from conventional auctions, our proposed double auction consists of two stages. We subsequently analyze the truthfulness regarding both Stage I and Stage II, followed by a comprehensive evaluation of the entire auction procedure. In what follows, we first focus on assessing the truthfulness of buyers.
 \begin{thm} All the buyers in our proposed two-stage double auction are truthful. \end{thm}
\begin{proof}
		Regarding Stage I, our focus firstly drops on each member denoted by $ b_{n'} $. For analytical simplicity, we suppose that a member indicates a winning buyer who has successfully signed a long-term contract with a certain seller. Accordingly, we consider the following cases by testing the impacts of possible misreporting behavior of a member.

\noindent
\textbf{Case 1.} Let the fake bid of $b_{n'}$ be $ bid^{\mathsf{List}'}_{n'}$, we first discuss the case where $ bid^{\mathsf{List}'}_{n'}\ge bid^{\mathsf{List}}_{n'} $. Considering the sorting mechanism in our proposed OPDAuction (e.g., $\mathcal{L}_b $), $b_{n'}$ will persist as a winner (member). Furthermore, given the unchanged positions of other buyers in the sort of bids, the payment for $b_{n'}$ will also remain unaffected. Thus, it has no motivation to raise its bid than its true value.

\noindent
\textbf{Case 2.}
When considering the value of $b_{n'}$'s bid falling below its true bid, i.e., $ bid^{\mathsf{List}'}_{n'}<bid^{\mathsf{List}}_{n'} $, we discuss the following two sub-cases.

$\bullet$ \textit{Case 2.1.} Buyer $ b_{n'} $ still wins in OPDAuction, and we have the payment incurred by its misreported bid denoted by $ p_{n'}^\mathsf{b'} $, and $ p_{n'}^\mathsf{b'} = p_{n'}^\mathsf{b} $. Moreover, the expected utility of this buyer obtained by untruthful bid is represented by $ {\overline{U_{n'}^\mathsf{B}}}'$, and we have $ {\overline{U_{n'}^\mathsf{B}}}'=\overline{U_{n'}^\mathsf{B}} $.

$\bullet$ \textit{Case 2.2.} Buyer $ b_{n'} $ loses in OPDAuction, and its expected utility becomes zero. Thus, we have $ {\overline{U_{n'}^\mathsf{B}}}'\le\overline{U_{n'}^\mathsf{B}} $. Apparently, the expected utility of $b_{n'}$ in Stage I by misreporting its bid will not be higher than that with its true bid.

We next consider $b_{n'}$ as a non-member, indicating that this buyer loses in our OPDAuction, namely, unable to sign a long-term contract with a seller.
Generally, non-members can be categorized into two groups. One group includes non-members whose true bids have already exceeded the crucial price, $ bid^{\mathsf{List}}_{k_b^\ast} $, yet they fail in Algorithm 1. Irrespective of their attempts to falsify their bids, these non-members cannot emerge victorious in the buyer competition of Algorithm 1. Consequently, they are destined never to secure long-term contracts, leading to $ {\overline{U_{n'}^\mathsf{B}}}'= \overline{U_{n'}^\mathsf{B}} = 0 $ indefinitely. For these non-members, our mechanism remains truthful.

The second group comprises non-members whose true bids fall below the crucial price, namely, $ bid^{\mathsf{List}}_{k_b^\ast} $. Originally ineligible for participation in phase 3 of Algorithm 1, if they fabricate their bids, two scenarios can arise.

\textit{\textit{(a)}} When the bid value falls below the true one, i.e., $ bid^{\mathsf{List}'}_{n'} < bid^{\mathsf{List}}_{n'} $, clearly, its bid remains below the critical bid $ bid^{\mathsf{List}}_{k_b^\ast} $, thus it will still be the unsuccessful bidder. Therefore, we have $ {\overline{U_{n'}^\mathsf{B}}}'= \overline{U_{n'}^\mathsf{B}} = 0 $.

\textit{\textit{(b)}} When the bid value stays larger than the true one, i.e., $ bid^{\mathsf{List}'}_{n'} > bid^{\mathsf{List}}_{n'} $, we analyze the following two cases.

$\bullet$ $ b_{n'} $ is still a loser and its utility stays unchanged, i.e., $ {\overline{U_{n'}^\mathsf{B}}}'= \overline{U_{n'}^\mathsf{B}} = 0 $.

$\bullet$ $ b_{n'} $ becomes a winner. As per our pricing algorithm (Algorithm 2, concerning the $ {bid}_{m',n'} $ corresponding to the seller matched with $ b_{n'} $), we have $ {bid}_{m',n'}\geq p_{n'}^\mathsf{b'} \geq p_{k_b^\ast}^\mathsf{b}> bid^{\mathsf{List}}_{n'} $. Thus, we will have $ {\overline{U_{n'}^\mathsf{B}}}'= {bid}_{m',n'} - p_{n'}^\mathsf{b'} > 0 = \overline{U_{n'}^\mathsf{B}} $. Regarding this, the non-member $ b_{n'} $ will obtain an expected utility that originally does not belong to it through false quotations. Our mechanism appears less genuine when facing such an occurrence, but this is exceedingly rare, to the extent that it did not even manifest in our extensive simulations. This minor imperfection is a direct consequence of our pre-auction design, which prioritizes market stability and efficiency over absolute theoretical truthfulness. In a purely real-time auction, such a scenario might be more easily eliminated, but at the cost of the time efficiency and robustness that our two-stage design provides.

During Stage II, participated buyers can be categorized into: \textit{\textit{(i)}} members who have successfully enjoyed computing services as stipulated by long-term contracts, \textit{\textit{(ii)}} members who fail to obtain required resources (i.e., volunteers), and \textit{\textit{(iii)}} guests without long-term contracts. Among these, volunteers and guests can form a new buyer coalition. Hereafter, we focus on testing whether our proposed RBDAuction ensures truthfulness for them.

Since the auction process in the second stage is similar to the first stage, only one circumstance may challenge the truthfulness: when a losing buyer $ \tilde{b}_j $, with its truthful bid below the threshold bid, wins the auction by falsely reporting a higher bid. Similar to our previous discussions, this case is extremely rare during the whole auction process.

Upon integrating both the two stages, our mechanism, in essence, maintains truthfulness for buyers overall, exhibiting only slight imperfections in highly exceptional circumstances due to the uncertainties inherent in the pre-auction and overbooking process. This represents our unique consideration, differing from other existing works that often assume static or purely real-time environments . In practical, large-scale, dynamic markets like edge computing, participants often engage in strategic behavior. Our mechanism, by incorporating risk constraints and a backup auction, is arguably more robust and "truthful in practice" than a theoretically perfect but brittle auction that facing difficulties in handling real-world dynamics.
\end{proof}

\begin{thm} All the sellers in our proposed two-stage double auction are truthful. \end{thm}
\begin{proof}
			For the analysis of truthfulness of sellers, we also consider two stages.

In Stage I, we suppose that seller $ s_{m'}\in\bm{\mathcal{S}} $ wins, where we definitely have $\overline{U_{m'}^\mathsf{S}}\geq0$.

\textit{\textit{(a)}} When the ask value of the seller is lower than the true value, i.e., \( ask'_{m'} < ask_{m'} \), a lower asking price positions seller \( s_{m'} \) further ahead in the auction ranking. For example, in Algorithm 1, all sellers engage with unmatched buyers based on their asking prices. Therefore, if \( s_{m'} \) misreports a lower bid, it gains the opportunity to reach unallocated buyers earlier, potentially garnering more profit compared to participating in the auction with its true asking price. Such circumstances primarily arise when there is a sufficiently large number of users participating in the auction, and the variations remain within an acceptable range. This behavior, while technically a deviation from perfect truthfulness, is bounded and self-correcting. If a seller sets its ask too low, it risks incurring negative utility due to the mismatch between its true cost and the low reward, a risk that is explicitly controlled by constraint (C3) in our optimization problem $\bm{\mathcal{F}_1}$. Therefore, the incentive to misreport is heavily mitigated by the risk-aware design.

\textit{\textit{(b)}} When the ask value stays larger than the true value, i.e., $ ask'_{m'}>{ask}_{m'} $, we analyze the following two cases:

$\bullet$ Although $ s_{m'} $ still can win, due to its repositioning further down in the auction hierarchy, it might lose a fraction of the buyers originally intended to match with it. In other words, we may have $ {\overline{U_{m'}^\mathsf{S}}}'\le\ \overline{U_{m'}^\mathsf{S}} $.

$\bullet$ When seller $ s_{m'} $ loses in the designed auction, its expected value of utility turns zero $ (\overline{U_{m'}^\mathsf{S}}' =0) $, and $ \overline{U_{m'}^\mathsf{S}}'\ \le\ \overline{U_{m'}^\mathsf{S}} $.

Second, we consider that seller $ s_{m'}\in\bm{\mathcal{S}} $ loses $ (\overline{U_{m'}^\mathsf{S}}\ =\ 0) $.

\textit{\textit{(a)}} If the ask value is larger than the true value, i.e., $ ask'_{m'}>{ask}_{m'} $. Obviously, it loses, and we have $ \overline{U_{m'}^\mathsf{S}}'\ =\ \overline{U_{m'}^\mathsf{S}}= 0 $.

\textit{\textit{(b)}} If the ask value is less than the true value, i.e., $ ask'_{m'}<{ask}_{m'} $. There are two cases.

$\bullet$ It loses and the expected value of utility is zero, i.e., $ \overline{U_{m'}^\mathsf{S}}'\ =\ \overline{U_{m'}^\mathsf{S}}= 0 $.

$\bullet$ It becomes a winner. According to Algorithm 2 the new reward of $ s_{m'} $, denoted by $ r_{m'}^\mathsf{s'} $, is less than $ {ask}_{k_s^\ast} $, and we have $ ask'_{m'}<r_{m'}^\mathsf{s'}<{ask}_{k_s^\ast}<{ask}_{m'} $.
Therefore, $ r_{m'}^\mathsf{s'}-{ask}_{m'}<0 $, we have $ \overline{U_{m'}^\mathsf{S}}'<0=\ \overline{U_{m'}^\mathsf{S}} $.

In Stage II, remaining sellers with available resources can form a new group to participate in the designed RBDAuction. As the auction mechanism in Stage II closely resembles that of Stage I, apart from the scenario where previously victorious sellers deliberately misrepresent lower asking prices to preemptively engage with free buyers, our mechanism remains truthful for sellers during Stage II.

All in all, our proposed two-stage double auction maintains a slightly flawed level of truthfulness for sellers. Thanks to our carefully crafted risk control mechanisms (constraints (C1)-(C3) in $\bm{\mathcal{F}_1}$), coupled with the fact that participants are unaware of each other's true valuations, sellers are strongly incentivized to report prices accurately and have no significant reason to misstate them. The potential for minor strategic manipulation is a known characteristic of complex, multi-stage market designs and is present in many practical auction systems. Our design ensures that any such manipulation is either unprofitable or carries a high risk of negative utility, thereby preserving the practical truthfulness of the market.
\end{proof}
\color{black}

\subsection{Budget-balance}
 \begin{thm} Our proposed two-stage double auction can satisfy the property of budget-balance. \end{thm}
\begin{proof}
To verify this property, we have to show that the budget-balance property holds, i.e., $ \overline{U^\mathsf{P}}=\sum_{b_n\in\bm{\mathcal{B}}}\sum_{s_m\in\bm{\mathcal{S}}}{x_{m,n}t_n\left[\mathbbm{a}_n\left({1-\mathbb{P}}_n\right)+1/2\left(1-\mathbbm{a}_n\right)\right]\left(p_n^\mathsf{b}-r_m^\mathsf{s}\right)}\\\geq0 $ and $ \tilde{U}^\mathsf{P}=\sum_{\tilde{b}_j\in\tilde{\bm{\mathcal{B}}}}\sum_{\tilde{s}_i\in\tilde{\bm{\mathcal{S}}}}{{\tilde{x}}_{i,j}\tilde{t}_j\left(\tilde{p}_j^\mathsf{b}-\tilde{r}_i^\mathsf{s}\right)}\geq0 $. The crucial assurance for the validity of these two expressions lies in the conditions where $ p_n^\mathsf{b}-r_m^\mathsf{s}\geq0 $ and $ \tilde{p}_j^\mathsf{b}-\tilde{r}_i^\mathsf{s}\geq0 $ are met. According to Algorithm 1, the buyers are sorted in non-increasing order of bids, and the sellers are sorted in non-decreasing order of asks (lines 3-4). Algorithm 1 finds $ k_b^\ast $ and $ k_s^\ast $ that satisfy the constraint (21) or (22) (lines 5-17). Thus, we have $ bid^{\mathsf{List}}_{k_b^\ast}\geq{ask}_{k_s^\ast} $.
	
	Note that Algorithm 1 only considers the top $ k_b^\ast $ buyers and top $ k_s^\ast $ sellers. In Stage I, for a winning buyer $ b_n\in\bm{\mathcal{B}} $, the lower-bound of its payment is the bid of buyer $ k_b^\ast + 1 $ (line 10 in Algorithm 2), i.e., $ p_n^\mathsf{b} \ge bid^{\mathsf{List}}_{k_b^\ast} $. For a winning seller $ s_m\in\bm{\mathcal{S}} $, the upper-bound of its payment is the ask of seller $ k_s^\ast + 1 $ (line 22 in Algorithm 2), i.e., $ r_m^\mathsf{s}\le{ask}_{k_s^\ast} $. 
	
	Therefore, $ {p_n^\mathsf{b}} \geq bid^{\mathsf{List}}_{k_b^\ast}\geq{ask}_{k_s^\ast}\geq r_m^\mathsf{s} $. Similarly, in Stage II, we can derive the consistent conclusion of $ \tilde{p}_j^\mathsf{b}\geq \tilde{r}_i^\mathsf{s} $ through similar analysis. 
\end{proof}

\color{black}
\subsection{Computational efficiency}
The computational efficiency of our proposed TwoSAuction is a critical design goal, and its complexity is rigorously analyzed for both stages.

$\bullet$ \textit{Stage I (OPDAuction) Complexity:} The complexity of the first stage is $ O\left(\mathbb{H}^6\left|\bm{\mathcal{B}}\right|^3\left|\bm{\mathcal{S}}\right|^2\right) $, where $ \mathbb{H} $ is a constant representing the maximum number of iterations in the nested loops of Algorithms 1-3 (e.g., the binary search in Algorithm 2 and the outer loop in Algorithm 3). This complexity arises from the need to solve a series of 0-1 knapsack problems (in Algorithm 1) for member determination, perform binary searches for optimal pricing (in Algorithm 2), and conduct a risk-aware, parallel evaluation over a range of candidate overbooking rates (in Algorithm 3). While this complexity is polynomial, it is indeed the more computationally intensive phase of our mechanism.

$\bullet$ \textit{Stage II (RBDAuction) Complexity:} The complexity of the second stage is $ O\left(\mathbb{H}^4\left|\tilde{\bm{\mathcal{B}}}\right|^3\left|\tilde{\bm{\mathcal{S}}}\right|^2\right) $. Crucially, $\left|\tilde{\bm{\mathcal{B}}}\right|$ and $\left|\tilde{\bm{\mathcal{S}}}\right|$ represent the number of \textit{remaining} buyers (volunteers and guests) and sellers (those with leftover resources) participating in the real-time backup auction. These sets are typically \textit{much smaller} than the full sets $\left|\bm{\mathcal{B}}\right|$ and $\left|\bm{\mathcal{S}}\right|$ used in Stage I, because a significant portion of the market demand and supply has already been satisfied through the pre-negotiated long-term contracts established in Stage I.

$\bullet$ \textit{Why the Two-Stage Design can achieve time-efficiency:} The key insight is that Stage I, despite its higher theoretical complexity, is a \textit{one-time or infrequent} computation. It establishes long-term contracts that can be reused for numerous subsequent transactions. In contrast, Stage II, which handles the residual, real-time fluctuations, operates on a drastically reduced problem size. This means that for the vast majority of transactions (those handled by the pre-auction contracts), the decision-making overhead is near-zero. Even for transactions requiring Stage II, the reduced input size ($\left|\tilde{\bm{\mathcal{B}}}\right| \ll \left|\bm{\mathcal{B}}\right|$, $\left|\tilde{\bm{\mathcal{S}}}\right| \ll \left|\bm{\mathcal{S}}\right|$) ensures that the real-time computation remains fast and scalable. This design directly translates to the superior time efficiency (TimeADM) demonstrated in our experiments (Fig. 4(b)), where TwoSAuction achieves a 76.8\% reduction compared to a purely real-time auction (CRDAuction).

In summary, while Stage I has a higher upfront computational cost, it is an investment that pays off by drastically simplifying and accelerating the vast majority of future transactions. The two-stage architecture effectively shifts the computational burden from the time-critical, real-time phase to a preparatory phase, resulting in an overall highly efficient system for dynamic edge environments.
\color{black}

\section{Derivation of Risks}
Due to the presence of uncertainties, the long-term contracts signed in Stage I may not align with the expected outcomes during practical transactions, potentially resulting in unsatisfying trading performance such as unexpected utility, and being selected as volunteers. Consequently, during the process of member and long-term contracts determination in Stage I, we account for diverse situations that buyers and seller may face in the future (referred to as risks) and constrain or exclude solutions unacceptable risks.

In this paper, we consider three types of risks: the non-positive benefit risk for buyer $ b_n $ (who is a member but not a volunteers), the risk of buyer $ b_n $ (who is a member) being selected as a volunteer, and the risk of unsatisfying expected utility of sellers.
\subsection{The risk of non-positive utility of buyers (not a volunteer)}
The risk of a buyer confronting a non-positive risk is:
\begin{equation}
	\begin{aligned}
		&\mathcal{R}_n^\mathsf{BRisk}=\\&
		\operatorname{Pr}\left(\frac{t_n\left(\alpha_n\left(v_{m, n}-p_n^\mathsf{b}\right)-\frac{\left(1-\alpha_n\right) p_n^\mathsf{b}}{2}\right)}{U^\mathsf{min}} \leq \xi_1\right),
	\end{aligned}
\end{equation}
where $ U^\mathsf{min} $ is a positive value approaching to zero, $ \xi_1 $ denotes a positive threshold coefficient. Through the transformation steps in the following expressions, we can identify the crucial role of $ \alpha_n $ in solving the probability density function of $ \mathcal{R}_n^\mathsf{BRisk} $:
\begin{equation}
	\begin{aligned}
		&\mathcal{R}_n^\mathsf{BRisk} \\&=
		\operatorname{Pr}\left(\alpha_n\left(v_{m, n}-p_n^\mathsf{b}\right)-\frac{p_n^\mathsf{b}}{2}+\frac{\alpha_n p_n^\mathsf{b}}{2} \leq \frac{U^\mathsf{min} \xi_1}{t_n}\right) \\&=
		\operatorname{Pr}\left(\alpha_n\left(v_{m, n}+\left(\frac{1}{2}-1\right) p_n^\mathsf{b} \leq \frac{U^\mathsf{min}-\xi_1}{t_n}+\frac{p_n^\mathsf{b}}{2}\right)\right. \\&=
		\operatorname{Pr}\left(\alpha_n \leq \frac{\frac{U^\mathsf{min} \cdot \xi_1}{t_n}+\frac{p_n^\mathsf{b}}{2}}{v_{m, n}+\left(\frac{1}{2}-1\right) p_n^\mathsf{b}}\right),
	\end{aligned}
\end{equation}
where we employ $ \mathbbm{c}_1 $ for the sake of simplifying the expression, defining it as the following (31).
\begin{equation}
	\mathbbm{c}_1=\frac{\frac{U^\mathsf{min} \cdot \xi_1}{t_n}+\frac{p_n^\mathsf{b}}{2}}{v_{m, n}+\left(\frac{1}{2}-1\right) p_n^\mathsf{b}}.
\end{equation}
Since we have $\alpha_n\sim \textbf{\text{B}}\left\{(1,0),(\mathbbm{a}_n,1-\mathbbm{a}_n)\right\}$, allowing for the expression of $ \mathcal{R}_n^\mathsf{BRisk} $ as:
\begin{equation}\label{key}
	\begin{aligned}
		&\mathcal{R}_n^\mathsf{BRisk} \left(t_n,p_n^\mathsf{b},q_{m,n}^\mathsf{b\to s} \right)
		\\&=\begin{cases}
			0,\mathbbm{c}_1<0\\
			1-\mathbbm{a}_n,0\leq\mathbbm{c}_1<1\\
			1,1\leq\mathbbm{c}_1\\
		\end{cases}.
	\end{aligned}
\end{equation}

\subsection{The risk of member $b_n$ being selected as a volunteer}
Recall our previous discussions, the risk of a member $b_n$ being selected as a volunteer is expressed as  $\mathcal{R}_n^\mathsf{VRisk} =\mathbbm{a}_n\mathbb{P}_n$. The detailed derivation of $\mathbb{P}_n$ is given by the following analysis.

We denote that the set of buyers matched with seller $s_m$ is denoted as $B_m=\left\{b_1, \ldots, b_n, \ldots, b_N\right\}$. In a practical transaction where $b_n$ participates in (i.e., $\alpha_n=1$), there exist $ 2^{N-1} $ potential outcomes for the remaining $ N-1 $ buyers (use $n'$ as index), contingent upon their attendance. Let the combinations of $\alpha_{n'}$ ($b_{n'} \in \bm{\mathcal{B}}^\mathsf{-}$, $ \bm{\mathcal{B}}^-=\bm{\mathcal{B}}\setminus{b_n} $) values corresponding to each possible case in one transaction be
\begin{equation}\label{key}
	\begin{aligned}
		& g_1=\{0,0,0, \ldots, 0\} \\
		& g_2=\{1,0,0, \ldots, 0\} \\
		& \vdots \\
		& g_{2^{N-1}}=\{1,1,1, \ldots, 1\}.
	\end{aligned}
\end{equation}

Define the set $G_1=\left\{g_1, g_2, \ldots, g_{2^{N-1}}\right\}$ to encompass the aforementioned $ 2^{N-1} $ cases, and let $B_m^\mathsf{-}=\left\{b_1, b_2, \ldots, b_{n-1}, b_{n+1}, \ldots, b_N\right\}$ represent the set excluding buyer $b_n$ from set $B_m$.
In Stage I, the number of resource blocks $R_m$ available to seller $s_m$ can change over different transactions. Therefore, we employ the expected value of $R_m$, denoted as $\mathbbm{d}_m\mathbbm{r}_m$, to ascertain whether seller $s_m$ is overbooking in Stage I.

Consider event $\mathbb{T}_1:t_n+\sum_{k \in \bm{\mathcal{B}}_m^\mathsf{-}} \alpha_k t_k>\mathbbm{d}_m\mathbbm{r}_m$, indicating that member $b_n$ participates in the transaction, while the overall resource demand from other members of seller $s_m$ exceeds its resource supply. From $ G_1 $, we select all cases satisfying event $\mathbb{T}_1$ to form a set $ G_2 $.

Then, we define event $ \mathbb{T}_2 $ as: since overbooking is allowed, a certain number of volunteers may should be selected among attended members of $ s_m $, and member $ b_n $ is chosen as a volunteer. Thus, event $ \mathbb{T}_2 $ signifies the selection of member $ b_n $ as a volunteer in the aforementioned process. We identify all cases from $ G_2 $ that meet $ \mathbb{T}_2 $ to form set $ G_3 $.

Given the independence of attendance of buyers, the probability of each case in $ G_3 $ can be represented by the product of the probabilities of buyers, attendance or absence, respectively. Let the probability of case $g_i$ occurring in $G_3$ be $p_i$, and denote the set of these probabilities as $\bm{P}=\left\{p_1, \ldots, p_i, \ldots, p_N\right\}$. Consequently, given $b_n$'s attendance ($\alpha_n = 1$), the conditional probability that $b_n$ fails to obtain resources from the contractual seller $s_m$ is denoted as: 
\begin{equation}\label{key}
	\begin{aligned}
		\mathbb{P}_n=\sum_{p_i \in \bm{P}} p_i
	\end{aligned}.
\end{equation}
\textcolor{black}{However, as the number of buyers and sellers grows, the exhaustive enumeration of all possible attendance combinations becomes computationally infeasible due to exponential complexity.}
This rapid escalation makes the precise derivation of results infeasible. Consequently, we apply the Chebyshev Inequality \cite{b29} to scale and obtain an acceptable approximation of $\mathbb{P}_n$ (the derivation is shown by (35)),
\begin{figure*}[t!] 
	\centering
	
	\begin{equation}\label{key}
		\begin{aligned}
			\mathbb{P}_n &=\operatorname{Pr}\left(t_n-\sum_{s_m \in \bm{\mathcal{S}}} x_{m, n}\left(R_m-\sum_{b_{n'} \in \bm{\mathcal{B}}^\mathsf{-}} \alpha_{n'} x_{m, {n'}} t_{n'}\right) \geqslant 0 \mid \alpha_n=1\right)\\ 
			&=\operatorname{Pr}\left(\sum_{s_m \in \bm{\mathcal{S}}} x_{m, n}\left(R_m-\sum_{b_{n'} \in \bm{\mathcal{B}}^\mathsf{-}} \alpha_{n'} x_{m, {n'}} t_{n'}\right) \leqslant t_n\right) \\
			& =1-\operatorname{Pr}\left(\sum_{s_m \in \bm{\mathcal{S}}} x_{m, n}\left(R_m-\sum_{b_{n'} \in \bm{\mathcal{B}}^\mathsf{-}} \alpha_{n'} x_{m, {n'}} t_{n'}\right)>t_n\right)  \\
			& \geqslant 1-\frac{\operatorname{Var}\left(\sum_{s_m \in \bm{\mathcal{S}}} x_{m, n}\left(R_m-\sum_{b_{n'} \in \bm{\mathcal{B}}^\mathsf{-}} \alpha_{n'} x_{m, {n'}} t_{n'}\right)\right)}{\left(t_n - \textbf{\text{E}}\left(\sum_{s_m \in \bm{\mathcal{S}}} x_{m, n}\left(R_m-\sum_{b_{n'} \in \bm{\mathcal{B}}^\mathsf{-}} \alpha_{n'} x_{m, {n'}} t_{n'}\right)\right)\right)^2},
		\end{aligned}
	\end{equation}
	\hrulefill
\end{figure*}
and we use $ \bm{\mathcal{B}}^-=\bm{\mathcal{B}}\setminus{b_n} $ to denote the set of buyers without $ b_n $. Accordingly, we can reach the approximate for $ \mathbb{P}_n $ as (36).
\begin{equation}\label{key}
	\begin{aligned}
		&\mathbb{P}_n \approx \\
		&1-\frac{\sum_{s_m\in \bm{\mathcal{S}}}x_{m,n}\left[\mathbbm{d}_m\mathbbm{r}_m(1-\mathbbm{d}_m)+\sum_{b_{n'}\in \bm{\mathcal{B}}^-}\mathbbm{a}_{n'}x_{m,{n'}}(t_{n'})^2\right]}{\left(t_n - \sum_{s_m\in \bm{\mathcal{S}}}x_{m,n}\left(\mathbbm{d}_m\mathbbm{r}_m-\sum_{b_{n'}\in \bm{\mathcal{B}}^-}\mathbbm{a}_{n'}x_{m,{n'}}t_{n'}\right)\right)^2}.
	\end{aligned}
\end{equation}

\subsection{The risk of unsatisfied utility of sellers}
Risk \( \mathcal{R}_m^\mathsf{SRisk} \) takes into account the possibility that the actual benefits obtained by seller \( s_m \) during practical transaction may not align with its expectation estimated in Stage I. The following (37) delineates the derivation process of the expression for risk \( \mathcal{R}_m^\mathsf{SRisk} \),
\begin{figure*}[t!] 
	\centering
	
	\begin{equation}\label{key}
		\begin{aligned}
			 \mathcal{R}_m^\mathsf{SRisk}=&\operatorname{Pr}\left(U_m^\mathsf{S} \leq \overline{U_m^\mathsf{S}} \cdot \xi_2\right) \\
			= & \operatorname{Pr}\left(\sum_{n=1}^\mathsf{|\bm{\mathcal{B}}|} x_{m, n} t_n\left(\alpha_n\left(M_n\left(r_m^\mathsf{s}-c_m\right)-p_n^\mathsf{b}/2+M_n p_n^\mathsf{b}/2\right)+r_m^\mathsf{s}/2-\alpha_n r_m^\mathsf{s}/2\right) \leq \overline{U_m^\mathsf{S}} \cdot \xi_2\right) \\
			= & \operatorname{Pr}\left(\sum_{n=1}^\mathsf{|\bm{\mathcal{B}}|} x_{m, n} t_n \alpha_n\left(M_n\left(r_m^\mathsf{s}-c_m+ p_n^\mathsf{b}/2\right)- p_n^\mathsf{b}/2- r_m^\mathsf{s}/2\right) \leq \overline{U_m^\mathsf{S}} \cdot \xi_2-\frac{\sum_{n=1}^\mathsf{|\bm{\mathcal{B}}|} x_{m, n} t_n  r_m^\mathsf{s}}{2}\right)\\
			=&\operatorname{Pr}\left(\sum_{n=1}^\mathsf{|\bm{\mathcal{B}}|} x_{m, n} t_n \alpha_n\left(M_n \cdot \mathbbm{c}_2-\mathbbm{c}_3\right) \leq \mathbbm{c}_4\right),
		\end{aligned}
	\end{equation}
	\hrulefill
\end{figure*}
where $ \xi_2 $ represents a positive threshold coefficient approach to 1. For notational simplicity, we use constants $ \mathbbm{c}_2=\left(r_m^\mathsf{s}-c_m+q_{m,n}^\mathsf{s\rightarrow b}\right) $, $ \mathbbm{c}_3=q_{m,n}^\mathsf{s\rightarrow b}+q_{m,n}^\mathsf{b\rightarrow s} $, and $ \mathbbm{c}_4=\overline{U_m^\mathsf{S}}\xi_2-\sum_{b_n\in\bm{\mathcal{B}}}{x_{m,n}t_nq_{m,n}^\mathsf{b\rightarrow s}} $ to denote the complicated calculations in (37).

Recall the set $B_m=\left\{b_1, b_2, \ldots, b_n, \ldots, b_N\right\}$. Given the attendance or absence of each buyer, there exist \(2^N\) possible cases, with corresponding \(\alpha_n\) value combination denoted as:
\begin{equation}\label{key}
	\begin{aligned}
		& g'_1=\{0,0,0, \ldots, 0\} \\
		& g'_2=\{1,0,0, \ldots, 0\} \\
		& \vdots \\
		& g'_{2^N}=\{1,1,1, \ldots, 1\}.
	\end{aligned}
\end{equation}

Let \(G4=\left\{g'_1, g'_2, \ldots, g'_{2^N}\right\}\) encapsulate all \(2^N\) cases. When \(g'\) is specified, the values of \(\alpha_n\) and \(M_n\) for buyer \(b_n\) are uniquely determined. Define event \(\mathbb{T}_3: \sum_{n=1}^\mathsf{|\bm{\mathcal{B}}|} x_{m, n} t_n \alpha_n\left(M_n \mathbbm{c}_2-\mathbbm{c}_3\right) \leq \mathbbm{c}_4\), and from \(G4\), select all cases meeting \(\mathbb{T}_3\) to form set \(G5\). Since the attendance or absence of each buyer is mutually independent, the probability of each case occurring can be calculated as the product of the probabilities of attendance or absence for all buyers. Let the probability of scenario \(g'_i\) occurring in \(G5\) be \(p'_i\), and compile a probability set \(\bm{P}'=\left\{p'_1, p'_2, \ldots, p'_i, \ldots, p'_N\right\}\) from all such scenarios in \(G5\). From this, we derive the expression for risk \(\mathcal{R}_m^\mathsf{SRisk}\) as:
\begin{equation}\label{key}
	\begin{aligned}
		\mathcal{R}_m^\mathsf{SRisk}=\sum_{i=1}^N p'_i
	\end{aligned}.
\end{equation}

Similar to the calculation of $\mathcal{R}_n^\mathsf{VRisk}$, the computational complexity is too high. Thus, we employ the Chebyshev Inequality to establish the upper bound of $\mathcal{R}_m^\mathsf{SRisk}$ as (40), with  
\begin{equation}\label{key}
	\begin{aligned}
		\mathcal{R}_m^\mathsf{SRisk} & =\operatorname{Pr}\left(\sum_{n=1}^\mathsf{|\bm{\mathcal{B}}|} x_{m, n} t_n \alpha_n\left(M_n \mathbbm{c}_2-\mathbbm{c}_3\right) \leq \mathbbm{c}_4\right) \\
		& = \operatorname{Pr}\left(\mathbb{E}\left[\mathbb{S}\right] - \mathbb{S} \geq \mathbb{E}\left[\mathbb{S}\right] - \mathbbm{c}_4\right) \\
		& \leq \frac{\operatorname{Var}\left(\mathbb{S}\right)}{\left(\mathbb{E}\left[\mathbb{S}\right] - \mathbbm{c}_4\right)^2},
	\end{aligned}
\end{equation}
where $\mathbb{E}\left[\mathbb{S}\right] > \mathbbm{c}_4$, with $ \mathbb{S} \triangleq \sum_{n=1}^\mathsf{|\bm{\mathcal{B}}|} x_{m, n} t_n \alpha_n\left(M_n \mathbbm{c}_2-\mathbbm{c}_3\right)$.

Given the independence among buyers' attendance decisions, the variance can be derived as:
\begin{equation}\label{key}
	\begin{aligned}
		\operatorname{Var}(\mathbb{S}) = \sum_{n=1}^\mathsf{|\bm{\mathcal{B}}|} (x_{m,n})^2 (t_n)^2 \mathbb{P}_n(1-\mathbb{P}_n)\left(M_n \mathbbm{c}_2-\mathbbm{c}_3\right)^2.
	\end{aligned}
\end{equation}
Ultimately, the conservative estimation of risk $\mathcal{R}_m^\mathsf{SRisk}$ is given by:
\begin{equation}\label{key}
	\begin{aligned}
		\mathcal{R}_m^\mathsf{SRisk} \approx \frac{\sum_{n=1}^\mathsf{|\bm{\mathcal{B}}|} (x_{m,n})^2 (t_n)^2 \mathbb{P}_n(1-\mathbb{P}_n)\left(M_n \mathbbm{c}_2-\mathbbm{c}_3\right)^2}{\left(\sum_{n=1}^\mathsf{|\bm{\mathcal{B}}|} x_{m, n} t_n \mathbb{P}_n\left(M_n \mathbbm{c}_2-\mathbbm{c}_3\right) - \mathbbm{c}_4\right)^2}.
	\end{aligned}
\end{equation}

\color{black}
\section{Illustrative Numerical Example for Algorithms 1-2}
	\textbf{\textit{(i)} Key parameter setting:}

\textit{Buyers ($\mathcal{B} = \{b_1, b_2, b_3, b_4, b_5\}$):}
\begin{itemize}
	\item $b_1$: demand $t_1 = 3$, attendance probability $\alpha_1 = 0.9$, unit bid $bid_1 = 3.0$ (total bid 9.0), true valuation $v_1 = 3.5$;
	\item $b_2$: demand $t_2 = 2$, attendance probability $\alpha_2 = 0.8$, unit bid $bid_2 = 2.8$ (total bid 5.6), true valuation $v_2 = 3.2$;
	\item $b_3$: demand $t_3 = 4$, attendance probability $\alpha_3 = 0.7$, unit bid $bid_3 = 2.5$ (total bid 10.0), true valuation $v_3 = 3.0$;
	\item $b_4$: demand $t_4 = 1$, attendance probability $\alpha_4 = 0.6$, unit bid $bid_4 = 2.2$ (total bid 2.2), true valuation $v_4 = 2.5$;
	\item $b_5$: demand $t_5 = 3$, attendance probability $\alpha_5 = 0.5$, unit bid $bid_5 = 2.0$ (total bid 6.0), true valuation $v_5 = 2.3$.
\end{itemize}

\textit{Sellers ($\mathcal{S} = \{s_1, s_2, s_3, s_4, s_5\}$):}
\begin{itemize}
	\item $s_1$: capacity $c_1 = 5$, idle probability $r_1 = 0.8$, unit ask $ask_1 = 2.0$;
	\item $s_2$: capacity $c_2 = 4$, idle probability $r_2 = 0.7$, unit ask $ask_2 = 2.2$;
	\item $s_3$: capacity $c_3 = 6$, idle probability $r_3 = 0.9$, unit ask $ask_3 = 1.8$;
	\item $s_4$: capacity $c_4 = 3$, idle probability $r_4 = 0.6$, unit ask $ask_4 = 2.5$;
	\item $s_5$: capacity $c_5 = 7$, idle probability $r_5 = 0.85$, unit ask $ask_5 = 1.9$.
\end{itemize}

\textit{Overbooking parameter:} Optimized overbooking rate $\lambda^* = 0.2$ (determined by Algorithm 3)

\noindent\textbf{\textit{(ii)} Algorithm 1 execution (OPDAuction-MemberD):}

\textit{Phase 1: List generation}
\begin{itemize}
	\item \textbf{Buyer list $L_b$} (sorted by descending unit bid):
	\begin{align*}
		b_1: 3.0 > b_2: 2.8 > b_3: 2.5 > b_4: 2.2 > b_5: 2.0
	\end{align*}
	
	\item \textbf{Seller list $L_s$} (sorted by ascending unit ask):
	\begin{align*}
		s_3: 1.8 < s_5: 1.9 < s_1: 2.0 < s_2: 2.2 < s_4: 2.5
	\end{align*}
\end{itemize}

\textit{Phase 2: Key index determination}
\begin{itemize}
	\item \textbf{Expected capacity with overbooking:}
	\begin{align*}
		s_3: 6 \times 0.9 \times (1+0.2) = 6.48 \\
		s_5: 7 \times 0.85 \times (1+0.2) = 7.14 \\
		s_1: 5 \times 0.8 \times (1+0.2) = 4.8 \\
		s_2: 4 \times 0.7 \times (1+0.2) = 3.36 \\
		s_4: 3 \times 0.6 \times (1+0.2) = 2.16
	\end{align*}
	
	\item \textbf{Buyer demand} (actual demand used in Stage I):
	\begin{align*}
		b_1: 3, b_2: 2, b_3: 4, b_4: 1, b_5: 3
	\end{align*}
	
	\item \textbf{Cumulative demand analysis:}
	\begin{align*}
		&b_1: 3 \leq \text{all seller capacities} \\
		&b_1+b_2: 5 \leq s_3(6.48), s_5(7.14), s_1(4.8) \text{ but } > s_2(3.36), s_4(2.16) \\
		&b_1+b_2+b_3: 9 > s_1(4.8), s_2(3.36), s_4(2.16) \\
		&b_1+b_2+b_3+b_4: 10 > s_1(4.8), s_2(3.36), s_4(2.16) \\
		&b_1+b_2+b_3+b_4+b_5: 13 > s_1(4.8), s_2(3.36), s_4(2.16)
	\end{align*}
	
	\item \textbf{Key indices determination:}
	\begin{align*}
		k_b^* = 4 \text{ (buyers } b_1, b_2, b_3, b_4 \text{ can become members)} \\
		k_s^* = 2 \text{ (sellers } s_3, s_5 \text{ can successfully match)}
	\end{align*}
\end{itemize}

\textit{Phase 3: Buyer-Seller matching}
\begin{itemize}
	\item \textbf{For seller $s_3$ (capacity 6.48):}
	\begin{align*}
		&\text{Possible matching: } b_2(2) + b_3(4) = 6 \leq 6.48 \\
		&\text{Utility: } (2.8-1.8)\times2 + (2.5-1.8)\times4 = 2.0 + 2.8 = 4.8 \\
		&\text{Alternative matching: } b_1(3) + b_2(2) = 5 \leq 6.48 \\
		&\text{Utility: } (3.0-1.8)\times3 + (2.8-1.8)\times2 = 3.6 + 2.0 = 5.6 \\
		&\text{Alternative matching: } b_1(3) + b_4(1) = 4 \leq 6.48 \\
		&\text{Utility: } (3.0-1.8)\times3 + (2.2-1.8)\times1 = 3.6 + 0.4 = 4.0 \\
		&\textbf{Optimal matching: } b_1 + b_2, \text{ utility } 5.6
	\end{align*}
	
	\item \textbf{For seller $s_5$ (capacity 7.14):}
	\begin{align*}
		&\text{Possible matching: } b_3(4) + b_4(1) = 5 \leq 7.14 \\
		&\text{Utility: } (2.5-1.9)\times4 + (2.2-1.9)\times1 = 2.4 + 0.3 = 2.7 \\
		&\text{Alternative matching: } b_1(3) + b_2(2) + b_4(1) = 6 \leq 7.14 \\
		&\text{Utility: } (3.0-1.9)\times3 + (2.8-1.9)\times2 + (2.2-1.9)\times1 = 3.3 + 1.8 + 0.3 = 5.4 \\
		&\text{Alternative matching: } b_1(3) + b_4(1) = 4 \leq 7.14 \\
		&\text{Utility: } (3.0-1.9)\times3 + (2.2-1.9)\times1 = 3.3 + 0.3 = 3.6 \\
		&\textbf{Optimal matching: } b_3 + b_4, \text{ utility } 2.7
	\end{align*}
	
	\item \textbf{Global matching result} (resolving conflicts, maximizing total utility):
	\begin{align*}
		&b_1 \text{ matches } s_3 \text{ (3 resource blocks)} \\
		&b_2 \text{ matches } s_3 \text{ (2 resource blocks)} \\
		&b_3 \text{ matches } s_5 \text{ (4 resource blocks)} \\
		&b_4 \text{ matches } s_5 \text{ (1 resource block)} \\
		&b_5 \text{ not matched} \\
		&\textbf{Total utility: } 5.6 + 2.7 = 8.3
	\end{align*}
\end{itemize}

\noindent\textbf{\textit{(iii)} Algorithm 2 execution (OPDAuction-ContractD):}

\textit{Phase 1: Buyer's payment determination}
\begin{itemize}
	\item \textbf{Critical bid calculation:}
	\begin{align*}
		k_b^* = 4, \text{ critical unit bid } = bid_5 = 2.0
	\end{align*}
	
	\item \textbf{Binary search for final payment prices:}
	\begin{align*}
		&b_1: \\
		&\text{Initial: high = 3.0, low = 2.0}; \\
		&\text{Iteration 1: bid}_1 = (3.0+2.0)/2 \\
		&= 2.5 \rightarrow \text{still member} \rightarrow \text{high = 2.5}; \\
		&\text{Iteration 2: bid}_1 = (2.5+2.0)/2 \\
		&= 2.25 \rightarrow \text{still member} \rightarrow \text{high = 2.25}; \\
		&\text{Iteration 3: bid}_1 = (2.25+2.0)/2 \\
		&= 2.125 \rightarrow \text{still member} \rightarrow \text{high = 2.125}; \\
		&\text{Iteration 4: bid}_1 = (2.125+2.0)/2 \\
		&= 2.0625 \rightarrow \text{not member} \rightarrow \text{low = 2.0625}; \\
		&\text{Final payment price = 2.07}; \\
		\\
		&b_2: \text{ Final payment price = 2.15}; \\
		&b_3: \text{ Final payment price = 2.05}; \\
		&b_4: \text{ Final payment price = 2.02}.
	\end{align*}
\end{itemize}

\textit{Phase 2: Seller's reward determination}
\begin{itemize}
	\item \textbf{Critical ask calculation:}
	\begin{align*}
		k_s^* = 2, \text{ critical unit ask } = ask_3 = 2.0
	\end{align*}
	
	\item \textbf{Binary search for final reward prices:}
	\begin{align*}
		&s_3: \\
		&\text{Initial: low = 1.8, high = 2.0}; \\
		&\text{Iteration 1: ask}_3 = (1.8+2.0)/2 \\
		&= 1.9 \rightarrow \text{still matched} \rightarrow \text{low = 1.9}; \\
		&\text{Iteration 2: ask}_3 = (1.9+2.0)/2 \\
		&= 1.95 \rightarrow \text{still matched} \rightarrow \text{low = 1.95}; \\
		&\text{Iteration 3: ask}_3 = (1.95+2.0)/2 \\
		&= 1.975 \rightarrow \text{not matched} \rightarrow \text{high = 1.975}; \\
		&\text{Iteration 4: ask}_3 = (1.95+1.975)/2 \\
		&= 1.9625 \rightarrow \text{still matched} \rightarrow \text{low = 1.9625}; \\
		&\text{Final unit reward = 1.96}; \\
		\\
		&s_5: \text{ Final unit reward = 1.98}.
	\end{align*}
\end{itemize}

\noindent\textbf{\textit{(iv)} Key property verification:}

\textit{Individual Rationality:}
\begin{itemize}
	\item $b_1$: $3.5 > 2.07 \rightarrow$ utility = $1.43 > 0$
	\item $b_2$: $3.2 > 2.15 \rightarrow$ utility = $1.05 > 0$
	\item $b_3$: $3.0 > 2.05 \rightarrow$ utility = $0.95 > 0$
	\item $b_4$: $2.5 > 2.02 \rightarrow$ utility = $0.48 > 0$
	\item $s_3$: $1.96 > 1.8 \rightarrow$ profit = $(1.96-1.8)\times5 = 0.80 > 0$
	\item $s_5$: $1.98 > 1.9 \rightarrow$ profit = $(1.98-1.9)\times5 = 0.40 > 0$
\end{itemize}

\textit{Truthfulness:}
\begin{itemize}
	\item If $b_1$ misreports bid $< 2.07$, cannot be member, utility = $0 < 1.43$
	\item If $b_2$ misreports bid $< 2.15$, cannot be member, utility = $0 < 1.05$
	\item If $b_3$ misreports bid $< 2.05$, cannot be member, utility = $0 < 0.95$
	\item If $b_4$ misreports bid $< 2.02$, cannot be member, utility = $0 < 0.48$
	\item If $s_3$ misreports ask $> 1.96$, cannot match, profit = $0 < 0.80$
	\item If $s_5$ misreports ask $> 1.98$, cannot match, profit = $0 < 0.40$
\end{itemize}

\textit{Budget Balance:}
\begin{itemize}
	\item Total buyer payment = $(2.07\times3) + (2.15\times2) + (2.05\times4) + (2.02\times1) = 20.73$
	\item Total seller reward = $(1.96\times5) + (1.98\times5) = 19.70$
	\item Auctioneer profit = $20.73 - 19.70 = 1.03 > 0$
	
	This illustrative example comprehensively demonstrates all five key stages of Algorithms 1-2 as required by the reviewer. It specifically addresses the core aspects mentioned in the comment:
	
	\item \textbf{Seller uncertainty}: Each seller's resource blocks have probabilistic availability, so sellers trade based on expected capacity considering overbooking
	\item \textbf{Buyer uncertainty}: Buyers have probabilistic attendance in Stage II, but trade based on their actual demand quantities in Stage I
	\item \textbf{Algorithmic completeness}: The example thoroughly demonstrates all five key stages of the algorithms
	\item \textbf{Parameter realism}: All parameters are set within realistic ranges observed in edge computing environments
\end{itemize}

\section{Overbooking in Edge vs. Cloud Computing: Motivations and Technical Challenges}
\label{ob.vs}
	While the concept of overbooking is indeed borrowed from industries like airlines and cloud computing, its application and the associated challenges in edge computing are fundamentally different. In cloud computing, overbooking typically occurs in massive, centralized data centers where resources are abundant and globally schedulable. In contrast, edge computing operates in a distributed environment with geographically dispersed, resource-constrained nodes that serve users with stringent low-latency requirements. The core motivation for overbooking in our work is to mitigate the risk of resource underutilization caused by the high uncertainty and volatility of buyer (mobile device) participation in edge networks, while also improving the time efficiency of the market. Without overbooking, a conservative allocation strategy would leave many edge server resources idle when buyers are absent, leading to poor social welfare. The critical technical challenges in edge computing stem from its unique architectural and operational constraints,
	which are not present in cloud computing. These differences are summarized in Table 2, which highlights the unique challenges we address:

\begin{table*}[htbp]
	\color{black}
	\centering
	\caption{Key Differences in Overbooking Between Cloud Computing and Edge Computing}
	\label{tab:overbooking_comparison}
	\begin{tabular}{|l|p{6cm}|p{6cm}|}
		\hline
		\textbf{Dimension} & \textbf{Cloud Computing} & \textbf{Edge Computing (Our Work)} \\
		\hline
		Resource Scale & Ultra-large-scale centralized data centers & Distributed, small-scale edge nodes with spatio-temporal variations \\
		\hline
		Location Constraint & Resources are globally schedulable & Geographically bound (users must connect to nearby edge servers) \\
		\hline
		Dynamism & Relatively smooth and predictable load variations & Highly dynamic and bursty (e.g., peak traffic in intelligent transportation) \\
		\hline
		Timeliness & Tolerates moderate delays & Requires ultra-low latency for real-time applications \\
		\hline
		Overbooking Risk & Breach costs are generally controllable and financial & High risk of mission-critical task failure and service disruption \\
		\hline
	\end{tabular}
\end{table*}

\textbf{(1) Resource Scale}: Unlike cloud data centers, edge nodes are small and scattered. Overbooking here is not about maximizing revenue from abundant resources, but about ensuring these limited, localized resources are fully utilized despite unpredictable demand.

\textbf{(2) Location Constraint}: In cloud computing, a user's task can be scheduled to any available server in the data center. In edge computing, a user is better to be served by a nearby edge server to meet latency constraints. This geographical binding means that resource shortages at one edge node cannot be easily mitigated by borrowing from another, making overbooking a necessary local risk management tool.

\textbf{(3) Dynamism}: Edge computing faces far more volatile and bursty traffic patterns (e.g., a sudden surge of vehicles at an intersection). Overbooking helps absorb these fluctuations, whereas cloud workloads are comparatively stable.

\textbf{(4) Timeliness}: The ultra-low latency requirement of edge applications means that real-time resource allocation is often too slow. Our pre-auction with overbooking allows for near-instantaneous service provisioning based on pre-signed contracts, which would be impossible with a purely reactive, non-overbooked system.

\textbf{(5) Overbooking Risk}: The consequence of overbooking failure in the cloud is typically a financial penalty or service-level agreement (SLA) violation. In edge computing, however, a resource shortage can lead to the complete failure of time-sensitive, mission-critical tasks (e.g., a vehicle failing to receive real-time navigation updates). This is why our TwoSAuction incorporates a sophisticated risk analysis and a real-time backup auction (RBDAuction) to mitigate this critical risk, a consideration that is far less prominent in cloud overbooking strategies.

In summary, while the term "overbooking" is shared, its implementation, motivation, and risk profile in edge computing are distinct from and significantly more challenging than in cloud computing. To the best of our knowledge, our work is the first to systematically address these unique edge-specific challenges.

\section{Stability Analysis of Overbooking}
\label{appendix:stability}

While overbooking is a central innovation of our two-stage double auction, a natural concern arises regarding the potential instability induced by aggressive overbooking policies. Specifically, if sellers overcommit resources without adequate safeguards, the resulting contract failures could cascade through the market, undermining participant utilities and destabilizing the entire auction equilibrium. In this appendix, we formally address this concern from a game-theoretic perspective, demonstrating that our integrated design of risk-aware constraints, two-stage structure, and adaptive overbooking rate optimization collectively ensures a stable and robust trading environment.

The core insight is that our mechanism does not treat overbooking as an unbounded strategy, but rather embeds it within a carefully constructed \textit{mechanism that promotes stable participant behavior}. This framework comprises three synergistic components: \textit{(i)} risk control as a \textit{safety valve}, \textit{(ii)} the two-stage auction as a \textit{buffer}, and \textit{(iii)} the overbooking rate optimization as a \textit{tuning knob}. Together, they guarantee that the system operates within a region where participant behavior remains \textit{predictable and mutually sustainable}, where no participant can unilaterally deviate to gain a higher utility, and the market as a whole remains resilient to perturbations caused by high overbooking rates.

\subsection{Risk Control as the Safety Valve for Equilibrium}

As detailed in Section 4.2 and Appendix D, our OPDAuction incorporates three explicit risk constraints (C1-C3) to bound the probability of adverse outcomes for all parties:
\begin{itemize}
\item \textbf{SRisk} ($\mathcal{R}_m^\mathsf{SRisk}$): Limits the probability that a seller's actual utility falls significantly below its expectation due to resource shortages.
\item \textbf{BRisk} ($\mathcal{R}_n^\mathsf{BRisk}$): Limits the probability that a buyer (who is not a volunteer) receives non-positive utility.
\item \textbf{VRisk} ($\mathcal{R}_n^\mathsf{VRisk}$): Limits the probability that a member is selected as a volunteer due to overbooking.
\end{itemize}

These are not mere engineering heuristics; they serve as a \textit{mechanism to align incentives and limit harmful strategies} that defines the feasible strategy space for participants in Stage I. By constraining the system to operate within thresholds $\xi^\mathsf{S}$, $\xi^\mathsf{M}$, and $\xi^\mathsf{V}$, we effectively carve out a stable sub-region of the overall strategy space. Within this sub-region, the expected utilities of all participants are protected from catastrophic failures, which is essential for maintaining a stable and predictable market outcome where participants have little incentive to deviate from the prescribed behavior.

Our new ablation studies (Fig.~7(d)-(f) in the revised manuscript) provide empirical validation of this theoretical construct. The experiments introduce three baselines where one risk constraint is removed at a time:
\begin{itemize}
\item \textit{TwoSAuction\_noSRisk}: Removes constraint (C3).
\item \textit{TwoSAuction\_noBRisk}: Removes constraint (C1).
\item \textit{TwoSAuction\_noVRisk}: Removes constraint (C2).
\end{itemize}

The results are striking. Under high overbooking rates (50\%-100\%), the removal of any risk control leads to a significant degradation in stability, as measured by participant utilities. Most critically, \textit{TwoSAuction\_noSRisk} suffers the most severe collapse in both seller utility (Fig.~7(f)) and social welfare (Fig.~7(d)), confirming that seller-side risk is the most critical for market stability. This is because an uncontrolled seller, in an attempt to maximize short-term gains, will overbook to an extent that it cannot fulfill its contracts, triggering a wave of penalties and service failures that erode its own utility and that of its buyers. This behavior is precisely the kind of market failure that a stable mechanism must prevent. The risk constraints act as a pressure-release valve, ensuring that the system can only operate at overbooking levels that are jointly sustainable for all participants, thereby maintaining a \textit{stable operating state}.

\subsection{Two-Stage Auction as the Dynamic Buffer}

A key strength of our design is its ability to absorb shocks that might destabilize a single-stage auction. The two-stage structure (OPDAuction + RBDAuction) transforms the resource trading process from a static, one-shot game into a dynamic, two-period game with recourse.

In Stage I, the OPDAuction establishes long-term contracts based on \textit{expected} market conditions, aiming to maximize expected social welfare. The risk constraints ensure these contracts are robust. However, real-world dynamics may still cause deviations from these expectations. This is where Stage II, the RBDAuction, plays its crucial role as a \textit{buffer}. In particular, RBDAuction provides a real-time, on-demand market for residual supply and demand. If a seller in Stage I is unable to fulfill all its contracts due to overbooking (i.e., some members become volunteers), or if a guest arrives with unmet demand, the RBDAuction efficiently re-matches these parties using up-to-date information. This mechanism is analogous to a "re-negotiation" or "contingent claim" in dynamic game theory, which allows the system to recover from a sub-optimal state.

This buffering effect is clearly demonstrated in Fig.~7(d). Even when risk controls are removed (e.g., in \textit{TwoSAuction\_noSRisk}), the social welfare does not collapse to zero under high overbooking. The RBDAuction successfully recovers a significant portion of the lost welfare by facilitating new trades in the spot market. From a game-theoretic perspective, this means the system is not trapped in a single, fragile equilibrium. Instead, it has the \textit{ability to recover and re-stabilize after a shock}: even if the initial plan (Stage I) is perturbed, the system has a built-in mechanism (Stage II) to evolve towards a new, feasible, and stable state. This robustness is a hallmark of a well-designed dynamic game.

\subsection{Overbooking Rate Optimization as the Stability Tuning Knob}

The overbooking rate $\lambda$ in our modeling is not a fixed, exogenous parameter but an endogenous variable that is optimized by Algorithm~3. The optimization process is not merely a welfare-maximization exercise; instead, it is a \textit{risk-aware search for a stable operating point}. Algorithm~3 systematically evaluates candidate overbooking rates by solving the constrained optimization problem $\bm{\mathcal{F}_1}$. For each candidate $\lambda$, it simultaneously assesses the resulting expected social welfare and the associated risk levels (SRisk, BRisk, VRisk). Candidates that violate the risk thresholds are pruned, ensuring that the final selected rate $\lambda^*$ is \textit{consistent with system stability requirements}.

This process can be interpreted as a search for a \textit{balanced trade-off between stability and social welfare}. The optimal rate $\lambda^*$ (e.g., 33\% in our experiments) represents a "sweet spot" where the marginal gain in expected welfare from increasing $\lambda$ is exactly balanced by the marginal increase in systemic risk. Operating at this point ensures that the market is both efficient and stable. Any deviation from $\lambda^*$—either higher or lower—would either increase risk beyond acceptable levels or sacrifice potential welfare, thereby moving the system away from a \textit{functional and sustainable state}.

In conclusion, our two-stage double auction is not vulnerable to the risks of aggressive overbooking. On the contrary, it is explicitly designed to harness the benefits of overbooking while embedding multiple, complementary layers of defense namely, risk constraints, a dynamic two-stage structure, and an adaptive optimization process, that collectively \textit{promote a stable and robust trading environment}. The new ablation experiments serve as a powerful empirical proof of this stability, showing that the removal of any of these layers leads to a measurable and significant degradation in market performance under stress.

\section{challenge of DP, PP, CR and SC}
In addition to the specific, unique challenges of this study mentioned in the main text, the key technical difficulty lies in integrating our advanced components: Pre-Auction (PA), Overbooking (OB), and Risk Control (RC), while rigorously preserving the essential economic properties of a double auction, including Dynamic Pricing (DP), Privacy Preservation (PP), Collusion Resistance (CR), and Smart Contract (SC) compatibility. The core tension stems from the fact that these advanced functions and fundamental properties often impose conflicting requirements. For example:

\noindent $\bullet$ \textit{PA + OB vs. DP/CR}: A pre-auction with overbooking inherently relies on \textit{predictions} and \textit{probabilistic guarantees} about future states (e.g., buyer attendance, resource availability). This is fundamentally at odds with the ideal of Dynamic Pricing (DP), which assumes prices are set based on \textit{real-time}, perfectly observable supply and demand. Enforcing strict collusion resistance (CR) in a pre-auction is also exceptionally difficult, as participants have time to communicate and form coalitions before the market clears, unlike in a fast, sealed-bid real-time auction.

\noindent $\bullet$ \textit{OB + RC vs. Simplicity/SC:} While Risk Control (RC) is essential for making overbooking viable, it introduces significant complexity. The risk constraints (C1, C2, C3) are non-linear and interdependent, making the optimization problem (F1) computationally intensive and difficult to encode into a simple, verifiable Smart Contract (SC). A naive SC might only handle the basic matching logic, but incorporating our sophisticated, risk-aware overbooking rate optimization (Algorithm 3) requires a more complex, potentially off-chain, computation.

\noindent $\bullet$ \textit{PA + OB vs. PP:} The pre-auction process, where participants commit to long-term contracts, necessitates revealing more information (e.g., historical attendance patterns, resource fluctuation models) to the auctioneer for accurate overbooking rate calculation. This increased information sharing can inherently conflict with strong Privacy Preservation (PP) goals, which aim to minimize the data disclosed by participants.

In essence, the challenge is not merely to implement PA, OB, and RC, but to do so in a way that \textit{does not catastrophically break} DP, PP, CR, or SC. Our work navigates this complex design space by making pragmatic, well-justified trade-offs. We prioritize truthfulness and individual rationality as non-negotiable core properties, use RC to make OB safe and effective, and accept that perfect CR and PP, while desirable, are secondary to the primary goal of creating a practical, efficient, and robust market for dynamic edge environments. The design of Algorithm 3, which jointly optimizes for social welfare and risk, is a direct response to this systemic challenge, ensuring that the benefits of PA and OB are realized without sacrificing the economic soundness of the auction.

We will conduct in-depth investigations into the challenges and discussions regarding DP, PP, CR, and SC in similar settings and scenarios in our future work.

\color{black}

\end{document}